\documentclass[twoside,11pt]{article}

\usepackage[preprint]{jmlr2e}
\hypersetup{hidelinks}

\input{packages.sty}
\input{commands.sty}

\ShortHeadings{Partial optimality in cubic correlation clustering for general graphs}{Stein, Andres and Di Gregorio}
\firstpageno{1}

\begin{document}
	\setlength{\belowdisplayskip}{4pt}
	\setlength{\abovedisplayskip}{4pt}
	\setlength{\belowdisplayshortskip}{1pt}
	\setlength{\abovedisplayshortskip}{1pt}
	
	\setlength{\marginparwidth}{2.5cm}
	
	\setlength{\textfloatsep}{4pt plus 4pt minus 4pt} 
	\setlength{\floatsep}{0pt} 
	\setlength{\intextsep}{0pt} 
	
	\setlength{\abovecaptionskip}{0pt}
	\setlength{\belowcaptionskip}{0pt}

\title{
	Partial Optimality in Cubic Correlation Clustering for General Graphs
}

\author{\name David Stein \email David.Stein1@tu-dresden.de \\
       \addr Faculty of Computer Science, TU Dresden, 01069 Dresden, Germany
       \AND
       \name Bjoern Andres \email bjoern.andres@tu-dresden.de \\
       \addr Faculty of Computer Science, TU Dresden, 01069 Dresden, Germany \\
       \addr Center for Scalable Data Analytics and AI Dresden/Leipzig, Dresden, Germany 
       \AND
       \name Silvia Di Gregorio \email digregorio@lipn.fr \\
       \addr Universit\a'e Sorbonne Paris Nord, LIPN, CNRS UMR 7030, F-93430 Villetaneuse, France \\
       \addr Faculty of Computer Science, TU Dresden, 01069 Dresden, Germany
       }

\editor{My editor}

\maketitle

\begin{abstract}
The higher-order correlation clustering problem for a graph $G$ and costs associated with cliques of $G$ consists in finding a clustering of $G$ so as to minimize the sum of the costs of those cliques whose nodes all belong to the same cluster. 
To tackle this NP-hard problem in practice, local search heuristics have been proposed and studied in the context of applications. 
Here, we establish partial optimality conditions for cubic correlation clustering, i.e., for the special case of at most 3-cliques. 
We define and implement algorithms for deciding these conditions and examine their effectiveness numerically, on two data sets.
\end{abstract}


\allowdisplaybreaks


\section{Introduction}
We study an optimization problem whose feasible solutions are all clusterings of a graph that is given as input. 
Here, a clustering of a graph $G = (V, E)$ is a partition $\Pi$ of the node set $V$ such that for any $U \in \Pi$, the subgraph $G[U]$ of $G$ induced by $U$ is connected.
In addition to the edges, we are given a subset $T$ of node triples $pqr$ such that $pq \in E$ and $qr \in E$ and $pr \in E$.
For any $e \in E$ and any $t \in T$, we are given numbers $c_e, c_t \in \mathbb{R}$ called costs.
The objective is to find a clustering $\Pi$ of $G$ so as to minimize the sum of the costs of those edges and triples whose nodes all belong to the same cluster.
More formally: 
\begin{definition}\label{def: first def}
For any graph $G = (V, E)$, any $T \subseteq \{ pqr \in \binom V3 \mid pq \in E, qr \in E, pr \in E \}$ and any $c \colon T \cup E \cup \{\emptyset\} \to \mathbb{R}$, the instance of the \emph{cubic correlation clustering} problem is
\begin{align}
\min_{\Pi \in P_G} \quad
c_\emptyset 
+ \sum_{R \in \Pi} \ \displaystyle\sum_{e \in E \cap \binom R2} c_e
+ \sum_{R \in \Pi} \ \sum_{t \in T \cap \binom R3} c_t
\enspace .
\label{eq:problem}
\end{align}
Here, $P_G$ denotes the set of all clusterings of $G$.
\end{definition}

The cubic correlation clustering problem is \textsc{np}-hard, as it  generalizes the \textsc{np}-hard clique partitioning problem for complete graphs \citep{goetschel-1989}, specializing to the latter in the case that $c_t = 0$ for all $t \in T$.
Potential applications of cubic correlation clustering include the tasks of fitting equilateral triangles to points in a plane, the task of fitting an unknown number of planes containing the origin to a set of points in 3-dimensional space, and the task of fitting an unknown number of lines to a set of points in 2-dimensional space, as discussed, e.g.,~by \cite{LevKarAndKeu22}. 

In this article, we ask whether we can compute a partial solution to the problem efficiently, i.e., decide efficiently for some edges or triples whether their elements are in the same cluster or distinct clusters of an optimal clustering.
In order to find such \emph{partial optimality}, we characterize \emph{improving maps} and state efficiently-verifiable sufficient conditions of their improvingness, 
a technique introduced by \citet{shekhovtsov-2013}; see also \cite{shekhovtsov-2014,shekhovtsov-2015}. 
In order to examine the effectiveness of these partial optimality conditions empirically, we implement algorithms for deciding these conditions, and conduct experiments, cf. Figure~\ref{figure:experiments}.

This article is an extension of a conference article \citep{stein-2023} about partial optimality conditions for the cubic correlation clustering problem with respect to a complete graph. Here, we extend the conference article in three ways:
Firstly, we transfer all conditions by~\cite{stein-2023} from complete to general graphs.
Secondly, we generalize the implementation of~\cite{stein-2023} from complete to general graphs.
Thirdly, we evaluate the effectiveness of these conditions on instances defined with respect to general (also incomplete) graphs.
These extensions make the algorithms for partial optimality applicable to larger instances with incomplete graphs. 
The complete source code of the algorithms and for reproducing the experiments is available at \url{https://github.com/dsteindd/partial-optimality-in-cubic-correlation-clustering-for-general-graphs}.

\begin{figure}[t]
	\centering
	\begin{minipage}[t]{0.5\linewidth}
		\centering
		\hspace{-\linewidth}
		\vspace{-2ex}
		a)
		\linebreak
		\begin{tikzpicture}[scale=1.4,rotate=90]
			\label{fig:illustration-partition}
			\draw (1.6, 0.3) ellipse (0.5 and 1);
			\draw (1.6, -1.2) ellipse (0.5 and 0.3);
			\draw (0.4, -0.8) ellipse (0.4 and 0.6);
			\draw (0.4, 0.65) ellipse (0.4 and 0.6);
			
			\filldraw[tertiary] (1.65, 0.65) circle (1pt);
			\filldraw[tertiary] (1.6, -0.1) circle (1pt);
			
			\draw[tertiary] (1.65, 0.65) -- (1.6, -0.1);
			
			\filldraw[tertiary] (0.3, 0.5) circle (1pt);
			\filldraw[tertiary] (0.6, 0.8) circle (1pt);
			\filldraw[tertiary] (0.1, 0.8) circle (1pt);			
			\draw[tertiary] (0.3, 0.5) -- (0.6, 0.8) -- (0.1, 0.8) -- cycle;
			%
			\filldraw[deep-red] (1.3, 0.2) circle (1pt);
			\filldraw[deep-red] (0.4, 0.3) circle (1pt);
			\draw[deep-red] (1.3, 0.2) -- (0.4, 0.3);
			
			\filldraw[deep-red] (1.35, -0.4) circle (1pt);
			\filldraw[deep-red] (0.7, -0.7) circle (1pt);
			\filldraw[deep-red] (1.3, -1.1) circle (1pt);
			\draw[deep-red] (1.35, -0.4) -- (0.7, -0.7) -- (1.3, -1.1) -- cycle;
			
			\filldraw (0.5, 1) circle (1pt);
			\filldraw (0.6, 0.5) circle (1pt);
			\filldraw (0.2, 0.3) circle (1pt);
			\filldraw (0.3, 0.7) circle (1pt);
			\filldraw (0.25, 0.9) circle (1pt);
			
			\filldraw (1.5, 1.1) circle (1pt);
			\filldraw (1.75, 0.3) circle (1pt);
			\filldraw (1.7, 0.9) circle (1pt);
			\filldraw (2.0, 0) circle (1pt);
			\filldraw (1.5, 0.1) circle (1pt);
			\filldraw (1.6, -0.5) circle (1pt);
			\filldraw (1.7, -0.2) circle (1pt);
			\filldraw (1.9, 0.6) circle (1pt);
			\filldraw (1.2, 0.4) circle (1pt);
			\filldraw (1.3, 0.7) circle (1pt);
			\filldraw (1.5, 0.5) circle (1pt);
			\filldraw (1.3, -0.2) circle (1pt);
			
			\filldraw (0.5, -1.1) circle (1pt);
			\filldraw (0.6, -0.6) circle (1pt);
			\filldraw (0.2, -0.5) circle (1pt);
			\filldraw (0.4, -0.8) circle (1pt);
			\filldraw (0.25, -1.1) circle (1pt);
			
			\filldraw (1.7, -1.4) circle (1pt);
			\filldraw (1.85, -1.2) circle (1pt);
			\filldraw (1.6, -1.1) circle (1pt);
			\filldraw (1.35, -1.25) circle (1pt);
		\end{tikzpicture}%
	\end{minipage}%
	\begin{minipage}[t]{0.5\linewidth}
		\centering
		\hspace{-\linewidth}
		\vspace{-2ex}
		b)
		\linebreak
	  \begin{tikzpicture}[xscale=1.4, yscale=-1.4,rotate=70]
	        	\label{fig:illustration-equilateral}
			\draw[black] (0.110281, 1.62544) -- (-0.596785, 0.782788) -- (0.486504, 0.591775) -- cycle;
			\draw[black] (0.302535, 0.00746821) -- (1.06858, 0.650256) -- (0.128886, 0.992276) -- cycle;
			\draw[black] (1.75981, -0.25) -- (0.980385, -0.7) -- (1.75981, -1.15) -- cycle;
			\draw plot[mark=*, mark size=0.1ex, mark options={draw=black, fill=black}] coordinates {(0.124757, 1.29645)};
			\draw plot[mark=*, mark size=0.1ex, mark options={draw=black, fill=black}] coordinates {(0.0156063, 1.56903)};
			\draw plot[mark=*, mark size=0.1ex, mark options={draw=black, fill=black}] coordinates {(0.0442783, 1.85471)};
			\draw plot[mark=*, mark size=0.1ex, mark options={draw=black, fill=black}] coordinates {(0.352589, 1.85776)};
			\draw plot[mark=*, mark size=0.1ex, mark options={draw=black, fill=black}] coordinates {(0.341288, 1.76396)};
			\draw plot[mark=*, mark size=0.1ex, mark options={draw=black, fill=black}] coordinates {(-0.788512, 0.629457)};
			\draw plot[mark=*, mark size=0.1ex, mark options={draw=black, fill=black}] coordinates {(-0.512841, 0.893121)};
			\draw plot[mark=*, mark size=0.1ex, mark options={draw=black, fill=black}] coordinates {(-0.681334, 0.771977)};
			\draw plot[mark=*, mark size=0.1ex, mark options={draw=black, fill=black}] coordinates {(-0.599445, 0.491229)};
			\draw plot[mark=*, mark size=0.1ex, mark options={draw=black, fill=black}] coordinates {(-0.214179, 0.741795)};
			\draw plot[mark=*, mark size=0.1ex, mark options={draw=black, fill=black}] coordinates {(-0.600574, 0.563377)};
			\draw plot[mark=*, mark size=0.1ex, mark options={draw=black, fill=black}] coordinates {(0.684411, 0.85045)};
			\draw plot[mark=*, mark size=0.1ex, mark options={draw=black, fill=black}] coordinates {(0.431443, 0.609925)};
			\draw plot[mark=*, mark size=0.1ex, mark options={draw=black, fill=black}] coordinates {(0.337814, 0.569987)};
			\draw plot[mark=*, mark size=0.1ex, mark options={draw=black, fill=black}] coordinates {(0.402965, 0.536616)};
			\draw plot[mark=*, mark size=0.1ex, mark options={draw=black, fill=black}] coordinates {(0.807089, 0.664689)};
			\draw plot[mark=*, mark size=0.1ex, mark options={draw=black, fill=black}] coordinates {(0.611389, 0.82912)};
			\draw plot[mark=*, mark size=0.1ex, mark options={draw=black, fill=black}] coordinates {(0.244796, 0.331031)};
			\draw plot[mark=*, mark size=0.1ex, mark options={draw=black, fill=black}] coordinates {(0.422437, -0.0280672)};
			\draw plot[mark=*, mark size=0.1ex, mark options={draw=black, fill=black}] coordinates {(0.371178, 0.0800537)};
			\draw plot[mark=*, mark size=0.1ex, mark options={draw=black, fill=black}] coordinates {(0.242193, -0.101483)};
			\draw plot[mark=*, mark size=0.1ex, mark options={draw=black, fill=black}] coordinates {(0.621105, 0.0697063)};
			\draw plot[mark=*, mark size=0.1ex, mark options={draw=black, fill=black}] coordinates {(0.468242, -0.292406)};
			\draw plot[mark=*, mark size=0.1ex, mark options={draw=black, fill=black}] coordinates {(0.315929, 0.175373)};
			\draw plot[mark=*, mark size=0.1ex, mark options={draw=black, fill=black}] coordinates {(0.810572, 0.490854)};
			\draw plot[mark=*, mark size=0.1ex, mark options={draw=black, fill=black}] coordinates {(0.801826, 0.628778)};
			\draw plot[mark=*, mark size=0.1ex, mark options={draw=black, fill=black}] coordinates {(0.786745, 0.888599)};
			\draw plot[mark=*, mark size=0.1ex, mark options={draw=black, fill=black}] coordinates {(1.09038, 1.10942)};
			\draw plot[mark=*, mark size=0.1ex, mark options={draw=black, fill=black}] coordinates {(1.13157, 0.460952)};
			\draw plot[mark=*, mark size=0.1ex, mark options={draw=black, fill=black}] coordinates {(1.23018, 0.811745)};
			\draw plot[mark=*, mark size=0.1ex, mark options={draw=black, fill=black}] coordinates {(1.00766, 0.916428)};
			\draw plot[mark=*, mark size=0.1ex, mark options={draw=black, fill=black}] coordinates {(0.228776, 0.929358)};
			\draw plot[mark=*, mark size=0.1ex, mark options={draw=black, fill=black}] coordinates {(0.245888, 1.25887)};
			\draw plot[mark=*, mark size=0.1ex, mark options={draw=black, fill=black}] coordinates {(0.105466, 1.03045)};
			\draw plot[mark=*, mark size=0.1ex, mark options={draw=black, fill=black}] coordinates {(0.142065, 1.05879)};
			\draw plot[mark=*, mark size=0.1ex, mark options={draw=black, fill=black}] coordinates {(0.33803, 1.17063)};
			\draw plot[mark=*, mark size=0.1ex, mark options={draw=black, fill=black}] coordinates {(0.41758, 0.794262)};
			\draw plot[mark=*, mark size=0.1ex, mark options={draw=black, fill=black}] coordinates {(1.75563, -0.715098)};
			\draw plot[mark=*, mark size=0.1ex, mark options={draw=black, fill=black}] coordinates {(1.42387, -0.433054)};
			\draw plot[mark=*, mark size=0.1ex, mark options={draw=black, fill=black}] coordinates {(1.67543, -0.238974)};
			\draw plot[mark=*, mark size=0.1ex, mark options={draw=black, fill=black}] coordinates {(1.16393, -0.0127456)};
			\draw plot[mark=*, mark size=0.1ex, mark options={draw=black, fill=black}] coordinates {(1.40133, -0.267431)};
			\draw plot[mark=*, mark size=0.1ex, mark options={draw=black, fill=black}] coordinates {(1.1587, -0.546172)};
			\draw plot[mark=*, mark size=0.1ex, mark options={draw=black, fill=black}] coordinates {(1.00785, -0.949736)};
			\draw plot[mark=*, mark size=0.1ex, mark options={draw=black, fill=black}] coordinates {(0.89444, -0.571423)};
			\draw plot[mark=*, mark size=0.1ex, mark options={draw=black, fill=black}] coordinates {(0.886873, -0.692756)};
			\draw plot[mark=*, mark size=0.1ex, mark options={draw=black, fill=black}] coordinates {(1.00278, -0.755562)};
			\draw plot[mark=*, mark size=0.1ex, mark options={draw=black, fill=black}] coordinates {(1.79729, -1.28684)};
			\draw plot[mark=*, mark size=0.1ex, mark options={draw=black, fill=black}] coordinates {(1.63846, -1.29231)};
			\draw plot[mark=*, mark size=0.1ex, mark options={draw=black, fill=black}] coordinates {(1.88776, -1.06207)};
			\draw plot[mark=*, mark size=0.1ex, mark options={draw=black, fill=black}] coordinates {(1.69246, -1.41393)};
			\draw plot[mark=*, mark size=0.1ex, mark options={draw=black, fill=black}] coordinates {(1.89045, -1.0887)};
			\draw plot[mark=*, mark size=0.1ex, mark options={draw=black, fill=black}] coordinates {(1.89538, -0.978857)};
			\draw plot[mark=*, mark size=0.1ex, mark options={draw=black, fill=black}] coordinates {(1.85791, -1.05489)};
		\end{tikzpicture}
	\end{minipage}
	\caption{In order to examine the effectiveness of partial optimality conditions empirically, we implement algorithms for deciding these conditions and measure the fraction of fixed variables with respect to a parameter controlling the noise of the problem, for (a) synthetic instances with four clusters and noisy costs, and (b) instances for the task of finding equilateral triangles in a noisy point cloud.}
	\label{figure:experiments}
\end{figure}
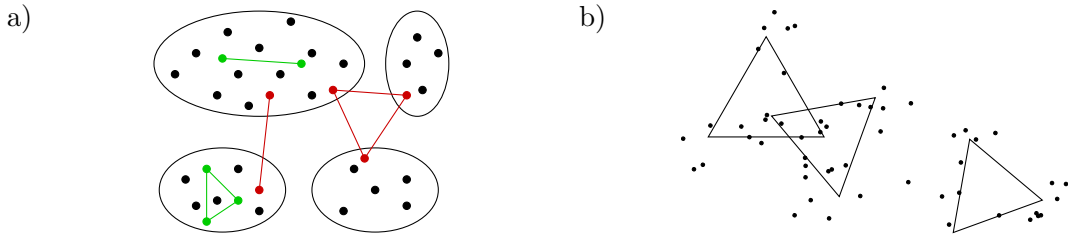


\section{Related Work}\label{sec: related work}

We choose to state the cubic correlation clustering problem 
(Definition~\ref{def: first def})
in the form of a non-linear binary program 
(Proposition~\ref{def: pb}), 
a special case of the higher-order correlation clustering problem introduced by \citet{kim-2014}.
Combinatorial optimization problems like this involving higher-order objective functions 
have applications as accurate mathematical abstractions of intrinsically non-linear tasks
\citep{AgaLimZelPerKriBel05,KapSpeReiSch16,kim-2014,LevKarAndKeu22,OchBro12,PulChiSadSut17}.

Like \citet{kim-2014}, we introduce costs only for node triples $t \in T$ such that $G[t]$ is complete. 
This means that we do not introduce costs for node triples for which the induced subgraph $G[t]$ is incomplete (possibly even disconnected).
Encoding whether three nodes that do not induce a complete subgraph belong to the same cluster would require additional variables and constraints \citep{hornakova-2017} that generalize the correlation clustering problem in a different manner than a cubic objective function does.

Being able to efficiently fix some variables to an optimal value and thus reduce the size of the problem can be valuable in practice.
Consequently, much effort has been devoted to studying partial optimality for non-convex problems 
\citep{AdaLasShe98,BilSut92,HamHanSim84,KapSpeReiSch13,KohEtAl08,shekhovtsov-2014,shekhovtsov-2015}.
In particular, we mention here the impressive application of partial optimality conditions to Potts models for image segmentation in which more than 95\% of the variables can be fixed \citep[Fig.~1]{shekhovtsov-2015}.
In contrast to the customary approach of considering a convex, usually linear, relaxation and establishing partial optimality conditions regarding the variables in the extended formulation, we study such conditions directly in the original discrete variable space.
Unlike the above-mentioned articles, we concentrate on taking advantage of the specific structure of the cubic correlation clustering problem.

To this end, we build on the works of \citet{alush-2012} and \citet{Lange-2018,Lange-2019} who establish partial optimality conditions for problems equivalent to correlation clustering with a linear objective function.
Regarding their terminology, we remark that the correlation clustering problem, the clique partitioning problem, and the multicut problem are equivalent if the objective functions are linear. 
The correlation clustering problem keeps attracting considerable attention by the community also in the context of approximation algorithms \citep{Vel22}.
The cubic correlation clustering problem we consider here generalizes the specialization to complete graphs of both the correlation clustering problem and the multicut problem.
We transfer all partial optimality conditions established by 
\citet{alush-2012} and \citet{Lange-2018,Lange-2019} 
for the correlation clustering problem and the multicut problem to the cubic correlation clustering problem.
In addition, we establish new results. 
Unlike \citet{Lange-2018}, we do not contribute partial optimality conditions for the max cut problem. 
\section{Preliminaries}\label{section:preliminaries}
In order to establish partial optimality conditions for the cubic correlation clustering problem (Definition~\ref{def: first def}), we state this problem in the form of the non-linear integer program introduced by \citet{kim-2014}:
\begin{proposition}\label{def: pb}
    The instance of the cubic correlation clustering problem (\ccp) with respect to a graph $G = (V, E)$, a set $T \subseteq \{ pqr \in \binom V3 \mid pq \in E, pr \in E, qr \in E \}$ and a function $c \colon T \cup E \cup \{\emptyset\} \to \mathbb{R}$ has the form of the cubic integer program
    \begin{align}
        \min_{x: E \to \{0, 1\}} \quad
        & c_\emptyset
        + \sum_{pq \in E} c_{pq} \, x_{pq} 
        + \sum_{pqr \in T} c_{pqr} \, x_{pq} \, x_{pr} \, x_{qr} 
        \label{eq: pb}
        \\ 
        \mathrm{s.t.:}\ \quad 
        & \forall (V_C, E_C) \in \mathrm{cycles}(G) \ 
        \forall e' \in E_C \colon \quad
        \sum_{e \in E_C \setminus \{e'\}}x_e - x_{e'} \leq \vert E_C\vert - 2 \ ,
        \label{eq:def-ccp}
    \end{align}
    where $\mathrm{cycles}(G)$ is the set containing all cycles in $G$.
\end{proposition}
\begin{proof}
    For each graph partition $\Pi\in P_G$ and every edge $pq \in E$, let $x_{pq} = 1$ if and only if $p$ and $q$ are in the same set of $\Pi$.
    This establishes a one-to-one relation between the set $P_G$ of all graph partitions of $G$ and the feasible set $X_G$ of all $x: E \to \{0, 1\}$ that satisfy the above inequalities \eqref{eq:def-ccp}, see \citep{chopra-1993}.
    Under this bijection, the objective functions of Definition~\ref{def: first def} and Proposition~\ref{def: pb}
    are equivalent.
\end{proof}
\vspace{6pt}

Below, let \csp{G}{c} denote this instance of the problem, $\phi_c$ its objective function, and $X_G$ its feasible set, i.e.,~the set of all $x: E \to \{0, 1\}$ that satisfy the inequalities~\eqref{eq:def-ccp}. 


\label{section:improving-maps}
Our main technique is the construction of improving maps \citep{shekhovtsov-2013}, which is based on the following preliminary notions.
\begin{definition}
	Let $X \neq \emptyset$, $\phi\colon X \to \mathbb{R}$ and $\sigma \colon X \to X$. If for every $x \in X$, we have $\phi(\sigma(x)) \leq \phi(x)$, then
	$\sigma$ is called \emph{improving} for the problem $\min_{x\in X}\phi(x)$.
\end{definition}
\begin{proposition}
	\label{lemma:persistency-predicate}
	Let $X \neq \emptyset$, $\phi \colon X \to \mathbb{R}$ and $\sigma \colon X \to X$ an improving map. Moreover, let $Q\subseteq X$. 
	If, for every $x \in X$, $\sigma(x) \in Q$,
	then there is an optimal solution $x^*$ to $\min_{x\in X}\phi(x)$ such that $x^*\in Q$.
\end{proposition}%
\begin{proof}
	Let $x^*$ be an optimal solution to $\min_{x\in X}\phi(x)$ such that $x^*\not\in Q$. Then $\sigma(x^*)$ is also an optimal solution to $\min_{x\in X}\phi(x)$ and $\sigma(x^*)\in Q$.
\end{proof}
\begin{corollary}
	\label{lemma:persistency-variable}
	Let $S \neq \emptyset$, $X \subseteq \{0, 1\}^S$, $\phi \colon X \to \mathbb{R}$ and $\sigma \colon X \to X$ an improving map. 
	Moreover, let $s \in S$ and $\beta \in \{0, 1\}$. 
	If for every $x \in X$, $\sigma(x)_s = \beta$,
	then there is an optimal solution $x^*$ to $\min_{x\in X}\phi(x)$ such that $x^*_s = \beta$.
\end{corollary}
Our construction starts from the elementary maps of \citet{Lange-2019}, i.e.~the map $\sigma_{\delta(U)}$ 
that cuts a set $U \subseteq V$ from its complement, and the map $\sigma_U$ that joins all sets intersecting with a set $U \subseteq V$:
\begin{definition}
	\label{def:elementary-cut-map}
	For any node 
	set $V$ and $U \subseteq V$,
	the \emph{elementary cut map} $\sigma_{\delta(U)}\colon \cp_G \to \cp_G$ is such that for all $x \in \cp_G$ and all $pq \in E$:
	\begin{equation}
		\sigma_{\delta(U)}(x)_{pq} := \begin{cases}
			0 & \textnormal{if $| \{p, q\} \cap U| = 1$} \\
			x_{pq} & \textnormal{otherwise}
		\end{cases}.
	\end{equation}
\end{definition}
\begin{lemma}
	For any node set $V$ and $U \subseteq V$, the map $\sigma_{\delta(U)}$ is well-defined, i.e., for all $x\in X_G$ we have that $\sigma_{\delta(U)}(x)\in X_G$.
\end{lemma}%
\begin{proof}
	Let $x\in X_G$  
    and let $x' = \sigma_{\delta(U)}(x)$. For the sake of contradiction, we assume that $x'\notin X_G$. Thus, there is a cycle $(V_C, E_C)\in \mathrm{cycles}(G)$ and an edge $e'\in E_C$ such that 
	$\sum_{e\in E_C\setminus \{e'\}} x'_e - x'_{e'} > \vert E_C\vert - 2.$
	This is only possible if $x'_e = 1$ for all $e\in E_C\setminus \{e'\}$ and $x'_{e'} = 0$. 
    By definition of $x'$, we have that $x'_e = x_e = 1$ for all $e\in E_C\setminus \{e'\}$ and all $e\in E_C\setminus \{e'\}$ are not cut by $U$. 
    For $e'$ we distinguish two cases:
	
	(1) If $e'$ is not cut by $U$, we have that $x'_{e'} = x_{e'}$.
    It follows that $0 = x'_{e'} = x_{e'}$. Thus 
	$\sum_{e\in E_C\setminus \{e'\}} x_e - x_{e'} > \vert E_C\vert - 2.$
	This is a violation of $x\in X_G$.
	
	(2) If $\lvert e'\cap U\rvert = 1$, namely $e'$ is cut by $U$, then $e'$ is the only edge in $E_C$ that has exactly one element in common with $U$. Therefore $(V_C, E_C)$ cannot be a cycle, which is a contradiction.
\end{proof}%
\begin{definition}
	\label{def:elementary-join-map}
	For any node sets $V$ and $U \subseteq V$,
	the \emph{elementary join map} $\sigma_U \colon \cp_G \to \cp_G$ is such that for all $x \in  \cp_G$ and all $pq \in E$:
	\begin{equation}
		\sigma_U(x)_{pq} := \begin{cases}
			1 & \textnormal{if $pq \in \tbinom U2 \cap E$} \\
			1 & \textnormal{if there is a $pq$-path $(V_P, E_P)$ such that $\forall e \in E_P \colon x_e = 1 \vee e\in \tbinom U2 \cap E$}\\
			x_{pq} & \textnormal{otherwise}
		\end{cases}
	\end{equation}
\end{definition}%
\begin{lemma}
	For any node set $V$ and $U \subseteq V$, the map $\sigma_{U}$ is well-defined, i.e., for all $x\in X_G$ we have $\sigma_{U}(x)\in X_G$.
\end{lemma}
\begin{proof}
	Let $x\in X_G$ and let $x' = \sigma_{U}(x)$. For the sake of contradiction, we assume that $x'\notin X_G$. Thus there is a cycle $(V_C, E_C)\in \mathrm{cycles}(G)$ and an edge $f\in E_C$ such that 
	\begin{equation}
		\label{eq:conflicted-cycle-well-definedness-join-map}
		\sum_{e\in E_C\setminus \{f\}} x'_e - x'_{f} > \vert E_C\vert - 2\enspace.
	\end{equation}
	Therefore $x'_e = 1$ for all $e\in E_C\setminus \{f\}$ and $x'_{f} = 0$. Moreover $x_f = 0$ since $x' \geq x$ by definition of $\sigma_U$.
	
	For any $E'\subseteq E$, let $E_{0 \to 1}(E') := \{e\in E'\mid x_e = 0 , x'_e = 1\}\setminus \tbinom U2$. 
    Since $x'_f = 0$ we have that $f\notin E_{0\to 1}(E_C)$.
	Next we show that $E_{0\to 1}(E_C) = \emptyset$. 
	For the sake of contradiction, let us assume that $E_{0\to 1}(E_C) \neq \emptyset$. 
    We pick a cycle $(V_C, E_C) \in \mathrm{cycles}(G)$ with minimal $\lvert E_{0 \to 1}(E_C)\rvert$ and that fulfills~\eqref{eq:conflicted-cycle-well-definedness-join-map}. 
	We proceed to construct a different cycle $(V_{C'}, E_{C'})\in \mathrm{cycles}(G)$ such that $\lvert E_{0\to 1}(E_{C'}) \rvert = \lvert E_{0\to 1}(E_C)\rvert - 1$ that still fulfills ~\eqref{eq:conflicted-cycle-well-definedness-join-map}, which would lead to a contradiction. 
	Let $uv\in E_{0\to 1}(E_C)$ be any arbitrary edge. We know that $uv\in E_C\setminus \{f\}$. Therefore $x_{uv} = 0$ and $x'_{uv} = 1$ by definition of $E_{0\to 1}(E_C)$.
	Then, according to the definition of $\sigma_U$, there is a $uv$-path $(V_P, E_P)$ such that $x_e = 1$ for all $e\in E_P \setminus \tbinom U2$. 
	Let $(V_C', E_C')\in \mathrm{cycles}(G)$ be the symmetric difference of the two cycles $(V_C, E_C)$ and $(V_P, E_P \cup \{uv\})$, i.e., $E_{C'} = \left(E_{C}\setminus \left( E_P \cup \{uv\}\right)\right) \cup \left(\left(E_P \cup \{uv\}\right)\setminus E_C\right)$ and $V_{C'} = \{v\in V\mid \exists e\in E_{C'}\colon v\in e\}$. 
	As $E_{0\to 1}(E_P \cup \{uv\}) = \{uv\}$, we have that $E_{0\to 1}(\left(E_P \cup \{uv\}\right)\setminus E_C) = \emptyset$.
    Next, $E_{0\to 1}(E_{C'}) = E_{0\to 1}(E_C\setminus \left( E_P \cup \{uv\}\right)) \cup E_{0\to 1}(\left(E_P \cup \{uv\}\right)\setminus E_C) = E_{0\to 1}(E_C\setminus \left( E_P \cup \{uv\}\right)) = E_{0\to 1}(E_C) \setminus \{uv\}$. 
    The last equality follows from the fact that $x_e = 1$ for all $e \in E_P \setminus \tbinom U2$.
    Therefore $\lvert E_{0\to 1}(E_{C'})\rvert = \lvert E_{0\to 1}(E_C)\rvert - 1$ which contradicts minimality of $(V_C, E_C)$ with respect to $\vert E_{0\to 1}(E_C)\vert$. 
	Thus we must have $E_{0\to 1}(E_C) = \emptyset$.
	
	Since $E_{0\to 1}(E_C) = \emptyset$ and \eqref{eq:conflicted-cycle-well-definedness-join-map} holds, 
    we have that $x_e = 1$ for all $e \in E_C\setminus \tbinom U2$. Thus $(V_{C}, E_C\setminus \{f\})$ is a path connecting the endpoints of $f$ such that $x_e = 1$ for all $e\in \left(E_{C}\setminus \{f\}\right)\setminus \tbinom U2$. Therefore
	$\sum_{e\in E_C\setminus \{f\}}x_e - x_f = \lvert E_C\rvert - 1 > \lvert E_C\rvert - 2,$
	which contradicts to $x\in X_G$.
\end{proof}
\section{Partial Optimality Conditions}
\label{section:partial-optimality-criteria}
In this section, we establish partial optimality conditions for the cubic correlation clustering problem by constructing improving maps, starting from the elementary maps $\sigma_{\delta(U)}$ and $\sigma_U$ defined in Section~\ref{section:preliminaries}.
For simplicity, we introduce some notation.
For any $r \in \mathbb{R}$, let $r^\pm := \max\{0, \pm r\}$.
For any function $f\colon X \to Y$ and any $X'\subseteq X$, let $f\vert_{X'}\colon X' \to Y$ denote the restriction of $f$ to $X'$.
For any  
$U, U', U'' \subseteq V$, let
\begin{gather}
	\delta(U, U') := \left\{pq \in E \mid p \in U \land q \in U' \right\} \quad , \quad
	\delta(U) := \delta(U, V \setminus U), \\
	T_{UU'U''} := \left\{pqr \in T \mid p \in U \land q \in U' \land r \in U'' \right\}.
\end{gather}
For $I(G) := T \cup E \cup \{\emptyset\}$ and any $c \colon I(G) \to \mathbb{R}$, let 
\begin{gather}
	E^\pm := \left\{ pq \in E \mid c_{pq} \gtrless 0 \right\} \quad, \quad
	T^\pm := \left\{ pqr \in T \mid c_{pqr} \gtrless 0 \right\}, \\
	T_{E'} := \left\{ pqr \in T \mid E' \cap \tbinom{pqr}{2} \neq \emptyset\right\} \quad \forall E	' \subseteq E 
	\enspace .	
\end{gather}
For any subset $V_H \subseteq V$, let
$E_H = E \cap \tbinom{V_H}2, 
T_H = T \cap \tbinom{V_H}3.$
\subsection{Cut Conditions}
\label{section:partial-optimality-criteria-cuts}
Here, we establish partial optimality conditions that imply the existence of an optimal solution $x^*$ to \csp{G}{c} such that $x^*_{ij} = 0$ for some $ij \in E$ or $x^*\in \{x\in \cp_G \mid x_{ij}x_{ik}x_{jk} = 0\}$ for some $ijk \in T$.
The following Proposition~\ref{lemma:persistency-subset-separation} generalizes Theorem~1 of \citet{alush-2012} to cubic objective functions. 
\begin{proposition}
	\label{lemma:persistency-subset-separation}	
	Let $G = (V, E)$ a graph and $c \in \mathbb{R}^{I(G)}$. 
	If there exists $U \subseteq V$ such that
	\begin{align}
		\label{eq:edge-cut-condition-1}
		c_{pq} &\geq 0 \quad \forall pq\in \delta(U) \\
		\label{eq:edge-cut-condition-2}
		c_{pqr} &\geq 0 \quad \forall pqr\in T_{\delta(U)}
	\end{align}
	then there is an optimal solution $x^*$ to \csp{G}{c} such that $x^*_{ij} = 0$ for all $ij \in \delta(U)$.
\end{proposition}
\begin{proof}
We define $\sigma \colon \cp_G \to \cp_G$ such that for all $x \in \cp_G$ we have 
\begin{equation}
\sigma(x) := \begin{cases}
x & \text{ if } x_{ij} = 0 \:\:	 \forall ij \in \delta(U) \\
\sigma_{\delta(U)}(x) & \textnormal{otherwise}
\end{cases}
\enspace .
\end{equation}
For any $x\in \cp_G$, let $x' = \sigma(x)$. 
Firstly, the map $\sigma$ is such that $x'_{ij} = 0$
for all $ij \in \delta(U)$. 
Secondly, for any $x\in \cp_G$ such that there exists $ij\in \delta(U)$ such that $x_{ij} = 1$, we have 
\begin{align}
	\phi_c(x') - \phi_c(x) & = - \smashoperator[r]{\sum_{pqr \in T_{\delta(U)}}} c_{pqr}x_{pq}x_{pr}x_{qr} - 
	\smashoperator[r]{\sum_{pq\in \delta(U)}} c_{pq}x_{pq} \leq - 
	\smashoperator[r]{\sum_{pqr\in T_{\delta(U)}\cap T^-}} c_{pqr} - 
	\smashoperator[r]{\sum_{pq\in \delta(U)\cap E^-}}c_{pq} = 0.
\end{align}
The last equality is due to the fact that those sums vanish by Assumptions~\eqref{eq:edge-cut-condition-1} and \eqref{eq:edge-cut-condition-2}.
Applying Corollary~\ref{lemma:persistency-variable} concludes the proof.
\end{proof}
\vspace{6pt}
This condition can be exploited: When satisfied for a subset $U$, \csp{G}{c} decomposes into two independent subproblems. Firstly,
\begin{equation}
	\min_{x\in \cp_G} \phi_c(x) 
	= 
	\min_{x\in \cp_{G[U]}} \phi_{c\vert_{I(G[U])}}(x) 
	+ 
	\min_{x\in \cp_{G[V \setminus U]}} \phi_{c\vert_{I(G[V \setminus U])}}(x)
	\enspace ,
\end{equation}
where we denote the subgraph of $G$ induced by $U$ by $G[U]$, for any $U \subseteq V$. 
Secondly, given solutions 
\begin{gather}
x' \in \argmin_{x \in \cp_{G[U]}} \phi_{c\vert_{I(G[U])}}(x), \quad x'' \in \argmin_{x \in \cp_{G[V \setminus U]}} \phi_{c\vert_{I(G[V \setminus U])}}(x)
\enspace ,
\end{gather}
an optimal solution to the problem $\min_{x \in \cp_G} \phi_c(x)$ is given by the vector $x\in \cp_G$ such that 
\begin{align}
x_{pq} = \begin{cases}
	x'_{pq} & \text{if\ } pq \in E \cap \tbinom U2 \\
	x''_{pq} & \text{if\ } pq \in E \cap \tbinom{V \setminus U}{2} \\
	0 & \text{if\ } pq\in \delta(U)
\end{cases}
\enspace .
\end{align}
The following Proposition~\ref{proposition:edge-cut-persistency}, together with Proposition~\ref{lemma:edge-join-persistency} further below, generalizes 
Theorem~1 of \citet{Lange-2019} to cubic objective functions. 
\begin{proposition}
	\label{proposition:edge-cut-persistency}
	Let $G = (V, E)$ a graph 
	and $c \in \mathbb{R}^{I(G)}$. 
	Moreover, let $ij \in E$. 
	If there exists $U \subseteq V$ such that $ij \in \delta(U)$ and 
	\begin{equation}
		\label{eq:assumption-edge-cut-inequality}
		c_{ij}^+ \geq \sum_{pqr\in T_{\delta(U)}} c_{pqr}^- + \sum_{pq\in \delta(U)} c_{pq}^-
		\enspace ,
	\end{equation}
	then there is an optimal solution $x^*$ to \csp{G}{c} such that $x^*_{ij} = 0$.
\end{proposition}
\begin{proof}
Let  
$\sigma \colon \cp_G \to \cp_G$ be constructed as 
\begin{equation}
\sigma(x) := \begin{cases}
x & \textnormal{if $x_{ij} = 0$} \\
\sigma_{\delta(U)}(x) & \textnormal{otherwise}
\end{cases}
\enspace .
\end{equation}
For any $x\in \cp_G$, let $x' = \sigma(x)$. 
First of all, the map $\sigma$ is such that $x'_{ij} = 0$ for all $x\in \cp_G$. 
Next, for any $x\in \cp_G$ such that $x_{ij} = 1$, we have 
\begin{align}
&\phi_c(x') - \phi_c(x) 
= -c_{ij} -\sum_{pqr\in T_{\delta(U)}} c_{pqr}x_{pq} x_{pr}x_{qr} - \sum_{\substack{pq\in \delta(U) : \\ pq \neq ij}}c_{pq}x_{pq} \\
\leq &-c_{ij} + \sum_{pqr\in T_{\delta(U)}}c_{pqr}^- + \sum_{\substack{pq \in \delta(U) : \\ pq \neq ij}}c_{pq}^- = -c_{ij}^+ + 
\smashoperator[r]{\sum_{pqr\in T_{\delta(U)}}}c_{pqr}^- + \sum_{pq \in \delta(U)}c_{pq}^- \leq 0
\enspace .
\end{align}
The last inequality follows from Assumption~\eqref{eq:assumption-edge-cut-inequality}.
We conclude the proof by applying Corollary~\ref{lemma:persistency-variable}.
\end{proof}
\vspace{6pt}
The following Proposition~\ref{lemma:persistency-triplet-cut} establishes a partial optimality condition that implies the existence of an optimal solution $x^*$ such that $x^*_{ij}x^*_{ik}x^*_{jk} = 0$ for some $ijk\in T$. 
Note that one cannot conclude which of the variables $x^*_{ij}$, $x^*_{ik}$ or $x^*_{jk}$ equal zero.
\begin{proposition}
	\label{lemma:persistency-triplet-cut}
	Let $G = (V, E)$ a graph %
	and $c\in \mathbb{R}^{I(G)}$. 
	Moreover, let $ijk\in T$ and $U \subseteq V$ such that $ij, ik \in \delta(U)$.
	If
	\begin{align}
		c_{ijk}^+ + c_{ij}^+ + c_{ik}^+
		\label{eq:triplet-cut-condition}
		\geq &\sum_{pqr\in T_{\delta(U)}} c_{pqr}^-+ \sum_{pq\in \delta(U)} c_{pq}^-
		\enspace ,
	\end{align}
	there is an optimal solution $x^*$ to \csp{G}{c} such that $x^*_{ij}x^*_{ik}x^*_{jk} = 0$.
\end{proposition}
\begin{proof}
We define $\sigma \colon \cp_G\to \cp_G$ as 
\begin{equation}
\sigma(x) := \begin{cases}
x & \textnormal{if $x_{ij}x_{ik}x_{jk} = 0$}\\
\sigma_{\delta(U)}(x) & \textnormal{otherwise}
\end{cases}.
\end{equation}
For any $x\in \cp_G$, we denote $\sigma(x)$ by $x'$. 
Firstly, we observe that 
$x'_{ij}x'_{ik}x'_{jk} = 0$ for all $x\in \cp_G$. 
Secondly, for any $x\in \cp_G$ such that $x_{ij}x_{ik}x_{jk} = 1$, we have
\small{
\begin{align}
&\phi_c(x') - \phi_c(x) = 
- \smashoperator[r]{\sum_{pqr \in T_{\delta(U)}}} c_{pqr}x_{pq}x_{pr}x_{qr} - 
\smashoperator[r]{\sum_{pq\in \delta(U)}} c_{pq}x_{pq} \\
& \leq -c_{ijk} - c_{ij} - c_{ik} + 
\smashoperator[r]{\sum_{\substack{pqr\in T_{\delta(U)} : \\ pqr \neq ijk}}} c_{pqr}^- + 
\smashoperator[r]{\sum_{\substack{pq\in \delta(U) : \\ pq\not\in \{ij, ik\}}}} c_{pq}^- 
= -c_{ijk}^+ - c_{ij}^+ - c_{ik}^+ + 
\smashoperator[r]{\sum_{pqr\in T_{\delta(U)}}} c_{pqr}^- + 
\smashoperator[r]{\sum_{pq\in \delta(U)}} c_{pq}^- \leq 0.
\end{align}%
}%
\normalsize
The last inequality holds by
Assumption~\eqref{eq:triplet-cut-condition}.
Applying Proposition~\ref{lemma:persistency-predicate} with $Q = \{x\in \cp_G \mid x_{ij}x_{ik}x_{jk} = 0\}$ concludes the proof.
\end{proof}
\subsection{Join Conditions}
\label{section:partial-optimality-criteria-joins}
Next, we establish partial optimality conditions that imply the existence of an optimal solution $x^*$ to \csp{G}{c} such that $x^*_{ij} = 1$ for some $ij\in E$. 
This property can be applied to simplify a given instance by joining the nodes $i$ and~$j$. 
Below, Proposition~\ref{lemma:edge-join-persistency} transfers a result of \citet{Lange-2019} to the cubic correlation clustering problem. 
Lemma~\ref{lemma:auxiliary-edge-join} is an auxiliary lemma applied in the proof of Proposition~\ref{lemma:edge-join-persistency}.
\begin{lemma}
	\label{lemma:auxiliary-edge-join}
	Let $G = (V, E)$ be a graph, $c\in \mathbb{R}^{I(G)}$, $U\subseteq V$ and $ij\in \delta(U)$. Moreover, let $x\in X_G$ such that $x_{ij} = 0$ and 
	$\bar x = \left(\sigma_{\{i,j\}}\circ \sigma_{\delta(U)}\right)(x).$
	Then, $\bar x_{pq} = x_{pq}$ for all $pq\notin\delta(U)$.
\end{lemma}
\begin{proof}
	Let $\hat x = \sigma_{\delta(U)}(x)$. We have that $\hat x_{pq} = 0$ for all $pq \in\delta(U)$, and $\hat x_{pq} = x_{pq}$ for all $pq \notin \delta(U)$. Let us assume that there exists some $pq\notin\delta(U)$ such that $\bar x_{pq} \neq x_{pq}$. First, consider the case where $x_{pq} = \hat x_{pq} = 0$ and $\bar x_{pq} = 1$. Then, there must be a path from $p$ to $i$ and from $q$ to $j$ such that the edge variables along this path in $\hat x$ all have value 1. 
	Analogously, we could have a path from $p$ to $j$ and from $q$ to $i$ such that the edge variables along these paths in $\hat x$ all have value 1.
	Without loss of generality, we only consider the first case.
	Since at least one of these paths must contain an edge $e\in \delta(U) \setminus \{ ij \}$, it follows that $\hat x_e$ has value equal to zero. 
	This leads to a contradiction.
	Secondly, consider the case where $x_{pq} = \hat x_{pq} = 1$ and $\bar x_{pq} = 0$. 
	We observe that this case is impossible, as the elementary join map never sets to 0 a variable that was previously equal to 1 by definition.
\end{proof}
\begin{proposition}
	\label{lemma:edge-join-persistency}
	Let $G = (V, E)$ a graph 
	and $c \in \mathbb{R}^{I(G)}$. 
	Moreover, let $ij\in E$. 
	If there exists $U \subseteq V$ such that $ij \in \delta(U)$ and 
	\begin{align}
		2c_{ij}^- + \sum_{pqr\in T_{\{ij\}}} c_{pqr}^- 
		\label{eq:edge-join-inequality}
		\geq  \sum_{pqr\in T_{\delta(U)}}\vert c_{pqr}\vert + \sum_{pq\in \delta(U)}\vert c_{pq}\vert
		\enspace ,
	\end{align}
	there is an optimal solution $x^*$ to \csp{G}{c} such that $x^*_{ij} = 1$.
\end{proposition}
\begin{proof}
Let $\sigma \colon \cp_G \to \cp_G$ such that for all $x \in \cp_G$, we have 
\begin{equation}
\sigma(x) := \begin{cases}
x & \textnormal{if $x_{ij} = 1$} \\
\left(\sigma_{\{i,j\}}\circ \sigma_{\delta(U)}\right)(x) & \textnormal{otherwise}
\end{cases}
\enspace .
\end{equation}
For any $x\in \cp_G$, let $x' = \sigma(x)$. 
The map $\sigma$ is such that $x'_{ij} = 1$ for all $x\in \cp_G$. 
We show that $\sigma$ is improving. 
In particular, let $x \in \cp_G$ such that $x_{ij} = 0$. 
By Lemma~\ref{lemma:auxiliary-edge-join} we have that $x'_{pq} = x_{pq}$ for all $pq \not\in \delta(U)$. 
Therefore,
\small{
\begin{align}
& \phi_c(x') - \phi_c(x) = \sum_{pqr\in T_{\delta(U)}}c_{pqr}\left(x'_{pq}x'_{pr}x'_{qr} - x_{pq}x_{pr}x_{qr}\right) + \sum_{pq\in \delta(U)}c_{pq}\left(x'_{pq} - x_{pq}\right) \\
& = \sum_{\substack{pqr\in T_{\delta(U)}: \\ pqr\notin T_{\{ij\}}}}c_{pqr}\left(x'_{pq}x'_{pr}x'_{qr} - x_{pq}x_{pr}x_{qr}\right) +\sum_{pqr\in T_{\{ij\}}}c_{pqr}x'_{pq}x'_{pr}x'_{qr}  + c_{ij} + \sum_{\substack{pq\in \delta(U) : \\ pq \neq ij}}c_{pq}\left(x'_{pq} - x_{pq}\right)\\
& \leq \; 
\smashoperator[r]{\sum_{\substack{pqr\in T_{\delta(U)}: \\ pqr \notin T_{\{ij\}}}}}\vert c_{pqr}\vert + 
\smashoperator[r]{\sum_{pqr\in T_{\{ij\}}}}c_{pqr}^+ + c_{ij} + 
\smashoperator[r]{\sum_{\substack{pq\in \delta(U) : \\ pq \neq ij}}}\vert c_{pq}\vert  
= 
\smashoperator[r]{\sum_{pqr\in T_{\delta(U)}}}\vert c_{pqr}\vert - 
\smashoperator[r]{\sum_{pqr\in T_{\{ij\}}}}c_{pqr}^- - 2c_{ij}^-  + 
\smashoperator[r]{\sum_{pq\in \delta(U)}}\vert c_{pq}\vert \leq  0 .
\end{align}
}%
\normalsize
The last inequality is due to Assumption~\eqref{eq:edge-join-inequality}.
\end{proof}
\vspace{6pt}
While Proposition~\ref{lemma:edge-join-persistency} is about deciding whether two adjacent nodes $i, j$ should end up together, Proposition~\ref{proposition:triplet-join} asks the same question for a triple $ijk$. 
\begin{proposition}
	\label{proposition:triplet-join}
	Let $G = (V, E)$ a graph 
	and $c\in \mathbb{R}^{I(G)}$. 
	Moreover, let $ijk\in T$ and $U \subseteq V$ such that $ij, ik \in \delta(U)$. If 
	\begin{align}
		\smallmath{
			2c_{ijk}^- + 2c_{ij}^- + 2c_{ik}^- + c_{jk}^- 
		- \sum_{\mathclap{pqr\in T}}c_{pqr}^+ - 
		\sum_{\mathclap{pq\in E}}c_{pq}^+  
		+ 
		\min_{\substack{x\in \cp_{G[\{i, j, k\}]} : \\ x_{ij}x_{ik}x_{jk} = 0}}\hspace{-0.25ex}
		\smashoperator[r]{\sum_{pq\in \tbinom{ijk}{2}}}c_{pq} x_{pq} 
		\label{eq:assumption-triplet-join}
		\geq 
		\smashoperator[r]{\sum_{\mathclap{pqr \in T_{\delta(U)}}}}c_{pqr}^- +  
		\smashoperator[r]{\sum_{\mathclap{pq \in \delta(U)}}}c_{pq}^-
		\enspace ,
	}
	\end{align}%
	there is an optimal solution $x^*$ to \csp{G}{c} such that $x^*_{ij}x^*_{ik}x^*_{jk} = 1$.
\end{proposition}
\begin{proof}
We construct 
$\sigma \colon \cp_G \to \cp_G$ as follows. 
\begin{equation}
\sigma(x) := \begin{cases}
x & \textnormal{if $x_{ij}x_{ik}x_{jk} = 1$} \\
\left(\sigma_{\{i, j, k\}} \circ \sigma_{\delta(U)}\right)(x) & \textnormal{otherwise}
\end{cases}
\end{equation}
For any $x \in \cp_G$, let $x' = \sigma(x)$. 
The map $\sigma$ is such that $x'_{ij}x'_{ik}x'_{jk} = 1$ for all $x\in \cp_G$. 
Note that for any $x\in \cp_G$ such that $x_{ij}x_{ik}x_{jk} = 0$, we have 
$x'_{pq} \geq x_{pq}$ for all $pq \not\in \delta(U)$. 

In order to see this, let $pq\notin\delta(U)$. We observe that $\sigma_{\delta(U)}(x)_{pq} = x_{pq}$ and that $\sigma_{\{i, j, k\}}(x)_{pq} \geq x_{pq}$ by definitions of $\sigma_{\delta(U)}$ and $\sigma_{\{i, j, k\}}$.
Moreover, we make the following observations: 
For every $pqr\in T_{\delta(U)}\setminus \{ijk\}$ we have that $c_{pqr}\left(x'_{pq}x'_{pr}x'_{qr} - x_{pq}x_{pr}x_{qr}\right) \leq \vert c_{pqr}\vert$ since $x'_{pq}x'_{pr}x'_{qr} - x_{pq}x_{pr}x_{qr} \in \{-1, 0, +1\}$. 
For every $pq\in \delta(U)\setminus \{ij, ik\}$ we have that $c_{pq}\left(x'_{pq} - x'_{pq}\right)\leq \vert c_{pq}\vert$ since $x'_{pq} - x_{pq}\in \{-1, 0, 1\}$.
For every $pqr\notin T_{\delta(U)}$ we have that $c_{pqr}\left(x'_{pq}x'_{pr}x'_{qr} - x_{pq}x_{pr} x_{qr}\right) \leq \max(0, c_{pqr}) = c_{pqr}^+$ since $x'_{pq}x'_{pr}x'_{qr} - x_{pq}x_{pr}x_{qr}\in \{0, 1\}$. For every $pq\notin\delta(U)\cup \{jk\}$ we have that $c_{pq}\left(x'_{pq} - x_{pq}\right)\leq \max(0, c_{pq}) = c_{pq}^+$ since $x'_{pq} - x_{pq}\in \{0, 1\}$.
It follows that 
\small{\begin{align}
& \phi_c(x') - \phi_c(x) =
 \smashoperator[r]{\sum_{\substack{pqr\in T_{\delta(U)} : \\ pqr \neq ijk}}}c_{pqr}(x'_{pq}x'_{pr}x'_{qr} - x_{pq}x_{pr}x_{qr}) 
 + 
 c_{ijk} 
 + 
 \smashoperator[r]{\sum_{pqr\not\in T_{\delta(U)}}}c_{pqr}(x'_{pq}x'_{pr}x'_{qr} - x_{pq}x_{pr}x_{qr}) \\					
&\quad + \sum_{pq\in \{ij, ik, jk\}}c_{pq}(1-x_{pq}) 
+ 
\sum_{\substack{pq\in \delta(U) : \\ pq \not\in \{ij, ik\}}}c_{pq}(x'_{pq} - x_{pq}) + \sum_{pq\not\in \delta(U)\cup \{jk\}}c_{pq}(x'_{pq} - x_{pq})\\	
&\leq \; c_{ijk} 
+  
\max_{\substack{x\in \cp_{G[\{i, j, k\}]} : \\ x_{ij}x_{ik}x_{jk} = 0}}
\smashoperator[r]{\sum_{pq\in \tbinom{ijk}{2}}}c_{pq} (1 - x_{pq}) + 
\smashoperator[r]{\sum_{\substack{pqr\in T_{\delta(U)} : \\ pqr \neq ijk}}}\vert c_{pqr}\vert  + 
\smashoperator[r]{\sum_{\substack{pqr\in T^+ : \\ pqr \notin T_{\delta(U)}}}} c_{pqr}  + 
\smashoperator[r]{\sum_{\substack{pq\in \delta(U) : \\ pq \notin \{ij, ik\}}}} \vert c_{pq} \vert + 
\smashoperator[r]{\sum_{\substack{pq \in E^+ : \\ pq \notin \left(\delta(U)\cup \{jk\}\right)}}}c_{pq} \\	
&= \; c_{ijk} +  \max_{\substack{x\in \cp_{G[\{i, j, k\}]} : \\ x_{ij}x_{ik}x_{jk} = 0}}\sum_{pq\in \tbinom{ijk}{2}}c_{pq} (1 - x_{pq}) + \sum_{\substack{pqr\in T^+ \cup T_{\delta(U)} : \\ pqr \neq ijk}}\vert c_{pqr}\vert + \sum_{\substack{pq\in E^+ \cup \delta(U) : \\ pq \not\in \{ij, ik, jk\}}}\vert c_{pq} \vert \\
&= \; -2c_{ijk}^- -  2c_{ij}^- - 2c_{ik}^- - c_{jk}^- - 
\hspace{-2ex}\min_{\substack{x\in \cp_{G[\{i, j, k\}]} : \\ x_{ij}x_{ik}x_{jk} = 0}}
\smashoperator[r]{\sum_{pq\in \tbinom{ijk}{2}}}c_{pq} x_{pq} + 
\sum_{\mathclap{pqr\in T_{\delta(U)}}}c_{pqr}^- +
\sum_{\mathclap{pqr\in T}}c_{pqr}^+ + 
\sum_{\mathclap{pq\in E}}c_{pq}^+ + 
\sum_{\mathclap{pq\in \delta(U)}}c_{pq}^- \leq  0.
\end{align}}%
\normalsize
Assumption \eqref{eq:assumption-triplet-join} provides the last inequality.
We arrive at the thesis by applying Proposition~\ref{lemma:persistency-predicate} with $Q = \{x\in \cp_G \mid x_{ij}x_{ik}x_{jk} = 1\}$. 
\end{proof}
\vspace{6pt}

The following Proposition~\ref{lemma:persistency-triangle-edge-join} expands Corollary~1 of \citet{Lange-2019} to the cubic correlation clustering problem.
\begin{proposition}
	\label{lemma:persistency-triangle-edge-join}
	Let $G = (V, E)$ a graph 
	and $c\in \mathbb{R}^{I(G)}$. 
	Moreover, let $ijk \in T$ and $U, U' \subseteq V$ such that $ij, ik \in \delta(U)$ and $jk, ik\in \delta(U')$. 
	If all of the following conditions hold, there exists an optimal solution $x^*$ to \csp{G}{c} such that $x^*_{ik} = 1$.
	\begin{align}
		c_{ijk}^- + 2c_{ij}^- + 2c_{ik}^- + \sum_{pqr\in T_{\{ij, ik\}}}c_{pqr}^-
		& \geq \sum_{pqr\in T_{\delta(U)}}\vert c_{pqr}\vert + \sum_{pq\in \delta(U)}\vert c_{pq} \vert  \label{eq:triangle-edge-join-1}		
		\\
		c_{ijk}^- + 2c_{jk}^- + 2c_{ik}^- + \sum_{pqr\in T_{\{jk, ik\}}}c_{pqr}^-  		
		& \geq \sum_{pqr\in T_{\delta(U')}}\vert c_{pqr}\vert + \sum_{pq\in \delta(U')}\vert c_{pq} \vert  \label{eq:triangle-edge-join-2}		
		\\
		c_{ijk} + c_{ij} + c_{ik} + c_{jk} 	
		& \leq - \sum_{\substack{pqr\in T_{\delta(ijk)}\cap T^- : \\ pqr\not\in T_{\{ij, ik, jk\}}}}\vert c_{pqr}\vert
		- \sum_{pq\in \delta(ijk)\cap E^-}\vert c_{pq} \vert \label{eq:triangle-edge-join-3}
	\end{align}%
\end{proposition}%
\begin{proof}
Let $\sigma \colon \cp_G \to \cp_G$ be defined as 
\begin{equation}
\sigma(x) := \begin{cases}
x & \textnormal{if $x_{ik} = 1$}\\
(\sigma_{\{i, k\}} \circ \sigma_{\delta(U)})(x) & \textnormal{if $x_{ik} = x_{ij} = 0 ,x_{jk} = 1$} \\
(\sigma_{\{i, k\}}\circ \sigma_{\delta(U')})(x) & \textnormal{if $x_{ik} = x_{jk} = 0 , x_{ij} = 1$} \\
(\sigma_{\{i, j, k\}}\circ \sigma_{\delta(ijk)})(x) & \textnormal{if $x_{ik} = x_{ij} = x_{jk} = 0$}
\end{cases}
\enspace .
\end{equation}
We use the notation $x' = \sigma(x)$ for all $x \in \cp_G$.
Firstly, the map $\sigma$ is such that $x'_{ik} = 1$ for all $x\in \cp_G$. 
Secondly, we consider $x \in \cp_G$ such that $x_{ik} = x_{ij} = 0$ and $x_{jk} = 1$. The map $\sigma_{\{i, k\}} \circ \sigma_{\delta(U)}$ is such that $x'_{ij} = x'_{ik} = x'_{jk} = 1$ and $x'_{pq} = x_{pq}$ for any $pq \notin \delta(U)$. 
The difference of the objective values for such $x\in X_G$ is given by 
\begin{align}
&\phi_c(x') - \phi_c(x) = \; c_{ijk} + c_{ij} + c_{ik} + \sum_{\substack{pqr \in T_{\{ij, ik\}} : \\ pqr\neq ijk}} c_{pqr}x'_{pq}x'_{pr}x'_{qr} \\
&\quad + \sum_{\substack{pqr\in T_{\delta(U)} : \\ pqr \notin T_{\{ij, ik\}}}} c_{pqr}\left(x'_{pq}x'_{pr}x'_{qr} - x_{pq}x_{pr}x_{qr}\right)+ 
\sum_{\substack{pq\in \delta(U) :\\ pq \not\in \{ij, ik\}}} c_{pq}(x'_{pq} - x_{pq}) 
\\
&\leq \; c_{ijk} + c_{ij} + c_{ik} 
+
\sum_{\substack{pqr\in T_{\{ij, ik\}}\cap T^+ : \\ pqr\neq ijk}}c_{pqr} 
+ 
\sum_{\substack{pqr\in T_{\delta(U)} : \\ pqr \notin T_{\{ij, ik\}}}}\vert c_{pqr}\vert 
+ 
\sum_{\substack{pq\in \delta(U) : \\ pq \not\in \{ij, ik\}}}  \vert c_{pq}\vert \\
&= -c_{ijk}^- - 2c_{ij}^- - 2c_{ik}^- 
+
\sum_{pqr\in T_{\{ij, ik\}}}c_{pqr}^+ 
+ 
\sum_{pqr\in T_{\delta(U)}}\vert c_{pqr}\vert 
- 
\sum_{pqr\in T_{\{ij, ik\}}}\vert c_{pqr}\vert
+ 
\sum_{pq\in \delta(U)}  \vert c_{pq}\vert 
\\
&= -c_{ijk}^- - 2c_{ij}^- - 2c_{ik}^- -\sum_{pqr\in T_{\{ij, ik\}}}c_{pqr}^- + \sum_{pqr\in T_{\delta(U)}}\vert c_{pqr}\vert + \sum_{pq\in \delta(U)}  \vert c_{pq}\vert \leq \;0
\enspace .
\end{align}	
The last inequality follows from 
Assumption~\eqref{eq:triangle-edge-join-1}. 
Thirdly, we consider $x \in \cp_G$ such that $x_{ik} = x_{jk} = 0$ and $x_{ij} = 1$. The map $\sigma_{\{i, k\}} \circ \sigma_{\delta(U')}$ is improving by analogous arguments and Assumption~\eqref{eq:triangle-edge-join-2}.
Finally, we consider $x\in \cp_G$ such that $x_{ik} = x_{jk} = x_{ij} = 0$. The map $\sigma_{\{i, j, k\}	}\circ \sigma_{\delta(ijk)}$ is such that 
\begin{equation}
(\sigma_{\{i, j, k\}}\circ \sigma_{\delta(ijk)})_{pq} = \begin{cases}
0 & \textnormal{if $pq \in \delta(ijk)$} \\
1 & \textnormal{if $pq \in \{ij, ik, jk\}$} \\
x_{pq} & \textnormal{otherwise}
\end{cases}
\enspace .
\end{equation}
Therefore, 
\begin{align}
\phi_c(x') - \phi_c(x) & = c_{ijk} + c_{ij} + c_{ik} + c_{jk} - \sum_{pqr\in T_{\delta(ijk)}\setminus T_{\{ij, ik, jk\}}} c_{pqr}x_{pq}x_{pr}x_{qr} - \sum_{pq\in \delta(ijk)} c_{pq} x_{pq} 
\\
& \leq c_{ijk} + c_{ij} + c_{ik} + c_{jk}  +\sum_{\substack{pqr\in T_{\delta(ijk)} \cap T^- : \\ pqr\not\in T_{\{ij, ik, jk\}}}} \vert c_{pqr} \vert + \sum_{pq\in \delta(ijk)\cap E^-} \vert c_{pq}\vert 
\overset{\eqref{eq:triangle-edge-join-3}}{\leq} 0
\enspace .
\end{align}%
Applying Corollary~\ref{lemma:persistency-variable} concludes the proof.
\end{proof}%
\vspace{6pt}
Next, we discuss a generalization of Theorem~2 of \citet{Lange-2019}.
To this end, let $V_H \subseteq V$, $E_H = E\cap \tbinom{V_H}2$ and $T_H = T\cap \tbinom{V_H}3$. 
We define $c'\in \mathbb{R}^{I(G[V_H])}$ by the equations written below.
\begin{align}
	c'_\emptyset 
	& = \frac 12 \sum_{pqr \in T_H}c_{pqr} + \sum_{pq \in E_H}c_{pq}	
	\label{eq: defcprime1} 
	\\
	\forall pq \in E_H\colon \quad 
	c'_{pq} 
	& = - c_{pq} + \frac 12 \sum_{\substack{r\in V_H \setminus \{p,q\} : \\ pqr \in T_H}} c_{pqr} 
	\label{eq: defcprime2} 
	\\
	\forall pqr \in T_H \colon \quad 
	c'_{pqr} 
	& = -2c_{pqr}  
	\label{eq: defcprime3}
\end{align}
Proposition~\ref{lemma:general-subgraph-edge-join} studies subsets $V_H \subseteq V$ such that $\mathbbm{1}_{E_H}$ is a trivial solution to the problem $\max_{x\in \cp_{G[V_H]}}\phi_{c'}(x)$, i.e.~$\max_{x\in \cp_{G[V_H]}}\phi_{c'}(x) = 0$, since 
\begin{align}
	&\phi_{c'}\left(\mathbbm{1}_{E_H}\right) = \sum_{pqr\in T_H}c'_{pqr} + \sum_{pq\in E_H}c'_{pq} + c'_\emptyset \\
	&= 
	-2\sum_{pqr\in T_H}c_{pqr} 
	+
	\sum_{pq\in E_H}
	\Bigl(-c_{pq} + \frac{1}{2}\sum_{\substack{r\in V_H : \\ pqr\in T_H}}c_{pqr}\Bigr) 
	+ 
	\frac{1}{2}\sum_{pqr\in T_H}c_{pqr} 
	+
	\sum_{pq\in E_H}c_{pq} \\
	&= 
	-\frac{3}{2}\sum_{pqr\in T_H}c_{pqr}  + \frac{1}{2}\sum_{pq\in E_H}\sum_{\substack{r\in V_H : \\ pqr\in T_H}}c_{pqr} 
	= -\frac{3}{2}\sum_{pqr\in T_H}c_{pqr}  + \frac{3}{2}\sum_{pqr\in T_H}c_{pqr} = 0 \enspace.
\end{align}
Lemma~\ref{lemma:subgraph-helping-lemma-1} is an auxiliary lemma applied in the proof of Proposition~\ref{lemma:general-subgraph-edge-join}.
\begin{lemma}
	\label{lemma:subgraph-helping-lemma-1}
	Let $G = (V, E)$ a graph and $c \in \mathbb{R}^{I(G)}$. 
	Define $c'\in \mathbb{R}^{I(G)}$ as in \eqref{eq: defcprime1}, \eqref{eq: defcprime2}, \eqref{eq: defcprime3} for $V_H = V$, $E_H = E$ and $T_H = T$.
	Then for any graph partition $\Pi \in P_G$:
	\begin{align}
		\smallmath{\phi_{c'}(x^\Pi) }
		&\smallmath{= \;\frac 12 \sum_{pqr\in T} c_{pqr}\prod_{\mathclap{uv\in \tbinom{pqr}{2}}}(1-x^\Pi_{uv}) + 
		\sum_{pqr\in T}c_{pqr}\sum_{\mathclap{uv \in \tbinom{pqr}{2}}}x^\Pi_{uv}
		\smashoperator[r]{\prod_{\substack{u'v'\in \tbinom{pqr}{2} \\ u'v' \neq uv}}}(1- x^\Pi_{u'v'}) + \sum_{pq\in E}c_{pq}(1-x^\Pi_{pq})	
	}\label{eq:partition-subgraph-helper-1} 
		\\ 
		& \smallmath{= \;\frac 12 \hspace{-1.5ex}\sum_{UU'U''\in \tbinom{\Pi}{3}}
		\smashoperator[r]{\sum_{pqr\in T_{UU'U''}}}c_{pqr} +
		\smashoperator[l]{\sum_{UU'\in \tbinom{\Pi}{2}}}
		\smashoperator[r]{\sum_{pq\in \delta(U, U')}}c_{pq} + 
		\sum_{UU'\in \tbinom{\Pi}{2}}
		\Bigl(
		\smashoperator[r]{\sum_{pqr\in T_{UUU'}}}c_{pqr} 
		+ 
		\hspace{-1.5ex}\smashoperator[r]{\sum_{pqr\in T_{UU'U'}}}c_{pqr}
		\Bigr) 
		\enspace ,}
		\label{eq:partition-subgraph-helper-2}
	\end{align}
	where $x^\Pi$ denotes the feasible vector corresponding to the graph partition $\Pi$.
\end{lemma}
\begin{proof}
	We use the fact that for any graph partition $\Pi\in P_G$ and any $pqr\in T$ we have that
	$x_{pq}^\Pi x_{qr}^\Pi=x_{pq}^\Pi x_{pr}^\Pi  = x_{pr}^\Pi x_{qr}^\Pi  = x_{pq}^\Pi x_{pr}^\Pi x_{qr}^\Pi $ thanks to the transitivity property, i.e. given a 3-clique either the three nodes are all apart, or only two are together, or they are all together. 
	Expanding the inner products and the inner sums leads to
	$\prod_{uv\in \tbinom{pqr}{2}}\left(1-x^\Pi _{uv}\right) = 1 - x_{pq}^\Pi  - x_{pr}^\Pi -x_{qr}^\Pi  + 2x_{pq}^\Pi x_{pr}^\Pi x_{qr}^\Pi$, and $\sum_{uv\in \tbinom{pqr}{2}}x_{uv}^\Pi \prod_{\substack{u'v'\in \tbinom{pqr}{2} : \\ u'v' \notin \{uv\}}}(1-x_{u'v'}^\Pi ) = x^\Pi _{pq} + x^\Pi _{pr} + x^\Pi _{qr} -3 x^\Pi _{pq}x^\Pi _{qr}x^\Pi _{pr}.$
	By substituting and collecting terms, we conclude the proof for Equation~\eqref{eq:partition-subgraph-helper-1}. 
	Equation~\eqref{eq:partition-subgraph-helper-2} follows instead from the following observations: 
	\begin{align}
		\prod_{ab\in \tbinom{pqr}{3}}\left(1-x^\Pi _{ab}\right) = 1 
		& \quad \Leftrightarrow \quad \exists \;UU'U''\in \tbinom{\Pi }{3}\colon pqr\in T_{UU'U''}
		\\
		\sum_{ab\in \tbinom{pqr}{2}}x_{ab}^\Pi \prod_{a'b'\in \tbinom{pqr}{2}\setminus \{ab\}} \hspace{-2ex} (1-x_{a'b'}^\Pi ) = 1 
		& \quad \Leftrightarrow \quad \exists \;UU'\in \tbinom{\Pi }{2}\colon 
		\left(pqr\in T_{UUU'} \lor pqr\in T_{UU'U'}\right)
		\\
		1 - x^\Pi _{pq} = 1 
		& \quad \Leftrightarrow \quad \exists \; UU'\in \tbinom{\Pi }{2}\colon pq\in \delta(U, U') 
		\enspace .
	\end{align}
	This concludes the proof.
\end{proof}
\begin{proposition}
	\label{lemma:general-subgraph-edge-join}
	Let $G = (V, E)$ a graph and $c\in \mathbb{R}^{I(G)}$. 
	Moreover, let $V_H \subseteq V$, $E_H = E\cap \tbinom{V_H}{2}$, $T_H = T\cap \tbinom{V_H}3$ and $ij \in E_H$.
	Define $c'\in \mathbb{R}^{I(G[V_H])}$ as in \eqref{eq: defcprime1}, \eqref{eq: defcprime2}, \eqref{eq: defcprime3}.
	If $\max_{x\in \cp_{G[V_H]}}\phi_{c'}(x) = 0$
	and for all $U \subseteq V_H$ with $i \in U$ and $j \in V_H \setminus U$ we have
	\begin{equation}\label{eq:assumption-subgraph-criterion-uv-cuts}
		\sum_{pq\in \delta(V_H)\cap E^-}c_{pq} + \sum_{pqr\in T_{\delta(V_H)}\cap T^-} c_{pqr} \geq \sum_{pq\in \delta(U, V_H \setminus U)}c_{pq} + \sum_{pqr\in T_{\delta(U, V_H \setminus U)}\cap T_H} c_{pqr} \enspace ,		
	\end{equation}%
	then there exists an optimal solution $x^*$ to \csp{G}{c} such that $x^*_{ij} = 1$.
\end{proposition}%
\begin{proof}
We define $\sigma \colon \cp_G \to \cp_G$ as
\begin{equation}
\sigma(x) := \begin{cases}
x & \textnormal{if $x_{ij} = 1$}\\
(\sigma_{V_H} \circ \sigma_{\delta(V_H)})(x) & \textnormal{otherwise}
\end{cases}
\enspace .
\end{equation}
Let $x' = \sigma(x)$ for any $x \in \cp_G$. 
It is easy to see that $x'_{ij} = 1$ for all $x\in \cp_G$.
Similarly to before, 
we show that $\sigma$ is an improving map. 
For any $x \in \cp_G$ such that $x_{ij} = 1$ we have $\phi_c(x') - \phi_c(x) = 0$, by definition of $x'$. 
Now, we consider $x\in \cp_G$ such that $x_{ij} = 0$.  
We let $x \vert_{E_H}$ denote the restriction of $x$ containing only components corresponding to edges in $E_H$. 
Let $\Pi$ be the graph partition of $V$ such that $x = x^\Pi$, and let $\Pi_H$ be the induced graph partition of $V_H$ such that $x \vert_{E_H} = x^{\Pi_H}$. 
Since $x_{ij} = 0$, there exist $U_1, U_2\in \Pi_H$ such that $i \in U_1$, $j \in U_2$. 
We have for all $pq\in E$:
\begin{equation}
x'_{pq} = \begin{cases}
1 & \textnormal{if $pq \in E_H$}\\
0 & \textnormal{if $pq \in \delta(V_H)$} \\
x_{pq} & \textnormal{otherwise}
\end{cases}
\enspace .
\end{equation}
Therefore, it follows that
\begin{equation}
	\smallmath{%
\phi_c(x') - \phi_c(x) = \sum_{\mathclap{pq\in E_H}} c_{pq}(1- x_{pq}) + \sum_{\mathclap{pqr\in T_H}}c_{pqr}(1-x_{pq}x_{pr}x_{qr}) - \sum_{\mathclap{pq\in \delta(V_H)}}c_{pq}x_{pq} - \sum_{\mathclap{pqr\in T_{\delta(V_H)}}}c_{pqr}x_{pq}x_{pr}x_{qr}.
}
\label{eq:subgraph-plugin-map-1} 
\end{equation}%
In order to find an upper bound for the sums over $E_H$ and $T_H$, we show that there exists a subset $U \subset V_H$ with $i \in U$ and $j \in V_H \setminus U$ such that 
\begin{equation} \label{eq:subgraph-simplification-1}
	\smallmath{\sum_{pq \in E_H}c_{pq}\left(1-x_{pq}\right) + \sum_{pqr\in T_H}c_{pqr}(1-x_{pq}x_{pr}x_{qr}) \leq \sum_{pq \in \delta(U, V_H \setminus U)}c_{pq} + \sum_{pqr\in T_{\delta(U, V_H \setminus U)}\cap T_H} c_{pqr}.}
\end{equation}
For the sake of contradiction, we assume that there is no such $U \subset V_H$.
For any $\Pi'\subset \Pi_H$, let $R_{\Pi'} = \bigcup_{P'\in \Pi'}P'$. 
Furthermore, we define $t \colon\tbinom{\Pi_H}{2}\cup \tbinom{\Pi_H}{3} \to \mathbb{R}$ and $e \colon \tbinom{\Pi_H}{2} \to \mathbb{R}$ such that
\begin{gather}
t_{UU'U''} = \sum_{pqr \in T_{UU'U''}}c_{pqr} \: \forall UU'U'' \in \tbinom{\Pi_H}{3}, \qquad t_{UU'} = \sum_{pqr \in T_{UUU'} \cup T_{UU'U'}}c_{pqr} \: \forall UU'\in \tbinom{\Pi_H}{2} 
\\
e_{UU'} = \sum_{pq \in \delta(U, U')}c_{pq} \: \forall UU'\in \tbinom{\Pi_H}{2}.
\end{gather}
Therefore, let $\Pi'\subset \Pi_H$ with $U_1\in \Pi'$ and $U_2\not\in \Pi'$. 
We observe that this implies that $i \in R_{\Pi'}$ and $j \notin R_{\Pi'}$, since $\Pi_H$ is a partition of $V_H$. 
By the assumption \eqref{eq:subgraph-simplification-1} of the proof by contradiction, we have that
\begin{equation}\label{eq:subgraph-simplification-2}
\sum_{pq \in E_H}c_{pq}\left(1-x_{pq}\right) + \sum_{pqr\in T_H}c_{pqr}(1-x_{pq}x_{pr}x_{qr}) 
> \hspace{-4ex}
\sum_{pq \in \delta(R_{\Pi'}, V_H\setminus R_{\Pi'})}\hspace{-4ex}c_{pq} 
+ 
\sum_{pqr\in T_{\delta(R_{\Pi'}, V_H\setminus R_{\Pi'})}\cap T_H} \hspace{-6ex}c_{pqr}
\enspace .
\end{equation}
We evaluate the terms in \eqref{eq:subgraph-simplification-2} one-by-one, and express them as sums over elements in $\Pi'$ and $\Pi_H \setminus \Pi'$.
Firstly, we observe that for any $pq \in E_H$ we have $x_{pq} = 0$ if and only if there exist $UU'\in \tbinom{\Pi_H}{2}$ such that $pq\in \delta(U, U')$.
Therefore,
$\sum_{pq \in E_H}c_{pq}\left(1-x_{pq}\right) = \sum_{UU'\in \tbinom{\Pi_H}{2}} e_{UU'}$
whereas
$\sum_{pq \in \delta(R_{\Pi'}, V_H \setminus R_{\Pi'})}c_{pq} = \sum_{U \in \Pi'} \sum_{U'\in \Pi_H \setminus \Pi'} e_{UU'}.$
For the first sum, we use the decomposition 
$\tbinom{\Pi_H}{2} = \tbinom{\Pi'}{2}\cup \left\{UU' \mid U \in \Pi'\land U'\in \Pi_H\setminus \Pi'\right\} \cup \tbinom{\Pi_H\setminus \Pi'}{2}$, 
where the subsets are mutually disjoint.
Consequently:
\begin{equation}\label{eq:subgraph-edge-difference}
\sum_{pq \in E_H}c_{pq}\left(1-x_{pq}\right) - \sum_{pq\in \delta(R_{\Pi'}, V_H\setminus R_{\Pi'})}c_{pq} = \sum_{UU'\in \tbinom{\Pi'}{2}} e_{UU'} + \sum_{UU'\in \tbinom{\Pi_H\setminus \Pi'}{2}} e_{UU'}
\enspace .
\end{equation}
Secondly, for any $pqr \in T_H$, we have $x_{pq}x_{pr}x_{qr} = 0$ if and only if there exist $UU'\in \tbinom{\Pi_H}{2}$ such that $pqr\in T_{UUU'}\cup T_{UU'U'}$ or there exist $UU'U''\in \tbinom{\Pi_H}{3}$ such that $pqr\in T_{UU'U''}$.
Therefore,
\begin{equation}\label{eq:subgraph-triangle-difference-1}
\sum_{pqr\in T_H}c_{pqr} \left(1 - x_{pq}x_{pr}x_{qr}\right) = \sum_{UU'U''\in \tbinom{\Pi_H}{3}}t_{UU'U''} + \sum_{UU'\in \tbinom{\Pi_H}{2}}t_{UU'} 
\end{equation}
whereas
\begin{equation}\label{eq:subgraph-triangle-difference-2}
	\smallmath{
		\sum_{\substack{pqr\in T_{\delta(R_{\Pi'}, V_H \setminus R_{\Pi'})} \\ pqr\in T_H}} \hspace{-2ex} c_{pqr} = 
		\sum_{UU'\in \tbinom{\Pi'}{2}}
		\smashoperator[r]{\sum_{\substack{U''\in \Pi_H \\ U'' \notin \Pi'}}}t_{UU'U''} + \sum_{U \in \Pi'}\sum_{U'U''\in \tbinom{\Pi_H \setminus \Pi'}{2}}t_{UU'U''} + \sum_{U \in \Pi'}\sum_{\substack{U'\in \Pi_H \\ U' \notin \Pi'}}t_{UU'}.
}
\end{equation}
For the first sum, we use the decomposition
	\begin{align} \label{eq:decomposition2}
		\scriptsizemath{
	\tbinom{\Pi_H}{3} = \ \tbinom{\Pi'}{3}\cup \left\{UU'U''\mid UU'\in \tbinom{\Pi'}{2}\land U''\in \Pi_H \setminus \Pi'\right\} \cup \left\{UU'U''\mid U \in \Pi'\land U'U''\in \tbinom{\Pi_H\setminus \Pi'}{2}\right\} \cup \tbinom{\Pi_H \setminus \Pi'}{3}
	}
	\end{align}%
where again the subsets are mutually disjoint.	
By \eqref{eq:subgraph-triangle-difference-1} together with the decomposition \eqref{eq:decomposition2}, and \eqref{eq:subgraph-triangle-difference-2}
it follows that
\begin{align}
& \sum_{pqr\in T_H}c_{pqr} \left(1 - x_{pq}x_{pr}x_{qr}\right) - \sum_{pqr\in T_{\delta(R_{\Pi'}, V_H \setminus R_{\Pi'})}\cap T_H}c_{pqr} 
\\ 
= & \sum_{UU'U''\in \tbinom{\Pi'}{3}}t_{UU'U''} + \sum_{UU'U''\in \tbinom{\Pi_H\setminus \Pi'}{3}}t_{UU'U''}
+ \sum_{UU'\in \tbinom{\Pi'}{2}}t_{UU'} + \sum_{UU'\in \tbinom{\Pi_H\setminus \Pi'}{2}}t_{UU'}
\enspace . 
\label{eq:subgraph-triplet-difference}
\end{align}
\noindent
Combining \eqref{eq:subgraph-simplification-2}, \eqref{eq:subgraph-edge-difference} and \eqref{eq:subgraph-triplet-difference} 
yields
\begin{align}
0 < 
& \sum_{pq\in E_H} c_{pq}(1-x_{pq}) 
+ \sum_{pqr\in T_H} c_{pqr}(1-x_{pq}x_{pr}x_{qr}) 
- \hspace{-2ex}\sum_{pq\in \delta(R_{\Pi'}, V_H \setminus R_{\Pi'})} \hspace{-3ex} c_{pq} 
- \hspace{-2ex}\sum_{pqr\in T_{\delta(R_{\Pi'}, V_H \setminus R_{\Pi'})}\cap T_H} \hspace{-4ex} c_{pqr}
\\
& = \sum_{UU'\in \tbinom{\Pi'}{2}} e_{UU'} + \sum_{UU'\in \tbinom{\Pi_H\setminus \Pi'}{2}} e_{UU'} + \sum_{UU'U''\in \tbinom{\Pi'}{3}}t_{UU'U''} + \sum_{UU'U''\in \tbinom{\Pi_H\setminus \Pi'}{3}}t_{UU'U''}
\\
\label{eq:individual-sum-subset-partition}
& + \sum_{UU'\in \tbinom{\Pi'}{2}}t_{UU'} + \sum_{UU'\in \tbinom{\Pi_H\setminus \Pi'}{2}}t_{UU'} =: W_{\Pi'}
\enspace .
\end{align}
Let $k = \vert \Pi_H\vert$ be the number of components in the partition $\Pi_H$, and $W_{\Pi'}$ the right-hand side of the last inequality. 
Recall that $U_1, U_2 \in \Pi_H$, $U_1 \in \Pi'$, and $U_2 \notin \Pi'$. 
As $W_{\Pi'} > 0$, it follows that at least one of the sums in its definition must not be vacuous. 
Moreover, since its sums are indexed by pairs or triples of subsets all belonging either to $\Pi'$ or to $\Pi_H \setminus \Pi'$, we observe that there must exist at least another subset 
of nodes in $\Pi_H$ different from $U_1$ and $U_2$. 
In fact, each sum is indexed over at least two subsets of $\Pi_H$ that are either both of them in $\Pi'$ or in $\Pi_H \setminus \Pi’$. 
Hence, at least one set between $\Pi’$ and $\Pi_H \setminus \Pi’$ has two elements. 
Given also that neither $\Pi’$ nor $\Pi \setminus \Pi’$ can be empty, it follows that $k \geq 3$. 

We calculate 
\begin{equation}
	\label{eq:sum-subsets-partition}
W = \sum_{\substack{\Pi'\subseteq \Pi_H\colon \\U_1\in \Pi', U_2\not\in \Pi'}}W_{\Pi'}.
\end{equation}
We need this in order to contradict $\max_{x\in \cp_{G[V_H]}}\phi_{c'}(x) = 0$. 

	First, we consider the special case $k = 3$. Without loss of generality, let $\Pi_H = \{U_1, U_2, U_3\}$. We observe that there are only two subsets of $\Pi_H$ that can serve as $\Pi'$ in the sum \eqref{eq:sum-subsets-partition} are $\Pi'_1 = \{U_1\}$ and $\Pi'_2 = \{U_1, U_3\}$.
	We have that
		$W_{\Pi'_1} = e_{U_2U_3} + t_{U_2U_3}, 
		W_{\Pi'_2} = e_{U_1U_3} + t_{U_1U_3}.$
	Therefore, we arrive at
		$0 < W = e_{U_1U_3} + e_{U_2U_3} + t_{U_1U_3} + t_{U_2U_3} = \phi_{c'}(x^{\Pi''})$, 
	where $\Pi'' = \left(\Pi_H \setminus \{U_1, U_2\}\right)\cup \{U_1 \cup U_2\}$ is the partition obtained by merging $U_1$ and $U_2$. The last equality follows from Lemma~\ref{lemma:subgraph-helping-lemma-1}. That contradicts $\max_{x\in X_{G\left[V_H\right]}}\phi_{c'}(x) = 0$.
	
	Next, we consider the case $k \geq 4$. We evaluate the values of all the sums in \eqref{eq:individual-sum-subset-partition}. We start by looking at the possible values $e_{UU'}$ and $t_{UU'}$. We observe that there are four different scenarios for their subscripts. 
	\begin{enumerate}
		\item For any $UU'\in \tbinom{\Pi_H}{2}$ and $\{U, U'\}\cap \{U_1, U_2\} = \emptyset$, there are exactly $2^{k-3}$ 
		subsets $\Pi'\subseteq \Pi_H$ such that $e_{UU'}$ or $t_{UU'}$ occurs in $W_{\Pi'}$ and $U_1\in \Pi', U_2 \not\in \Pi'$, namely $2^{k-4}$ subsets $\Pi'$ for which we have additionally $U, U'\in \Pi'$ plus $2^{k-4}$ subsets $\Pi'$ for which we have additionally $U, U'\notin \Pi'$.		
		\item For any $U\in \Pi_H \setminus \{U_1, U_2\}$ there are exactly $2^{k-3}$ subsets $\Pi'\subseteq \Pi_H$ such that $e_{U_1U}$ or $t_{U_1U}$ occurs in $W_{\Pi'}$ and $U_1\in \Pi', U_2\notin \Pi'$, namely the $2^{k-3}$ subsets $\Pi'$ for which we additionally have $U\in \Pi'$.		
		\item For any $U\in \Pi_H \setminus \{U_1, U_2\}$ there are exactly $2^{k-3}$ subsets $\Pi'\subseteq \Pi_H$ such that $e_{U_2U}$ or $t_{U_2U}$ occurs in $W_{\Pi'}$ and $U_1\in \Pi', U_2\notin \Pi'$, namely the $2^{k-3}$ subsets $\Pi'$ for which we additionally have $U\notin \Pi'$.		
		\item There is no $\Pi'\subseteq \Pi_H$ such that $e_{U_1U_2}$ or $t_{U_1U_2}$ occurs in $W_{\Pi'}$ with $U_1\in \Pi', U_2\not\in \Pi'$. 
	\end{enumerate}	
	We now turn to the remaining sums in \eqref{eq:individual-sum-subset-partition}, namely the ones that have three subsets as indices. We observe that for these there are again four possibilities.
	\begin{enumerate}
		\item For any $UU'U''\in \tbinom{\Pi_H}{3}$ 
		and $\{U, U', U''\}\cap \{U_1, U_2\} = \emptyset$, there are exactly $2^{k-4}$ subsets $\Pi'\subseteq \Pi_H$ such that $t_{UU'U''}$ occurs in $W_{\Pi'}$ and $U_1\in \Pi', U_2\not\in \Pi'$, namely $2^{k-5}$ subsets $\Pi'$ for which we have additionally $U, U', U''\in \Pi'$ plus $2^{k-5}$ subsets $\Pi'$ for which we have additionally $U, U', U''\notin \Pi'$.
        \item For any $UU'\in \tbinom {\Pi_H}2$ and $\{U, U'\}\cap \{U_1, U_2\} = \emptyset$, there are exactly $2^{k-4}$ subsets $\Pi'\subseteq \Pi_H$ such that $t_{U_1UU'}$ occurs in $W_{\Pi'}$ and $U_1 \in \Pi', U_2\notin\Pi'$, namely the $2^{k-4}$ subsets $\Pi'$ for which we have additionally $U, U'\in \Pi'$.
        \item For any $UU'\in \tbinom {\Pi_H}2$ and $\{U, U'\}\cap \{U_1, U_2\} = \emptyset$, there are exactly $2^{k-4}$ subsets $\Pi'\subseteq \Pi_H$ such that $t_{U_2UU'}$ occurs in $W_{\Pi'}$ and $U_1 \in \Pi', U_2\notin\Pi'$, namely the $2^{k-4}$ subsets $\Pi'$ for which we have additionally $U, U'\notin \Pi'$.
		\item There is no $\Pi'\subseteq \Pi_H$ such that $t_{U_1U_2U}$ occurs in $W_{\Pi'}$ for any $U \in \Pi_H\setminus \{U_1, U_2\}$ for which $U_1\in \Pi', U_2\not\in \Pi'$.
	\end{enumerate}
At this point, we put all of this together in \eqref{eq:sum-subsets-partition}:
\begin{align}
0 < W &= \sum_{\substack{\Pi'\subseteq \Pi_H\colon \\U_1\in \Pi', U_2\not\in \Pi'}}W_{\Pi'} 
= 2^{k-3}\sum_{UU'\in \tbinom{\Pi_H}{2}} e_{UU'} - 2^{k-3}e_{U_1U_2} +  2^{k-4} \sum_{UU'U''\in \tbinom{\Pi_H}{3}} t_{UU'U''} \\
&- 2^{k-4}  \sum_{\substack{U \in \Pi_H : \\ U \not\in \{U_1, U_2\}}}t_{U_1U_2U}
+ 
2^{k-3}\sum_{UU'\in \tbinom{\Pi_H}{2}} t_{UU'} - 2^{k-3} t_{U_1U_2} 
\\
& = 2^{k-3}\sum_{UU'\in \tbinom{\Pi''}{2}} e_{UU'}+  2^{k-4} \sum_{\mathclap{UU'U''\in \tbinom{\Pi''}{3}}}t_{UU'U''} + 2^{k-3}\sum_{UU'\in \tbinom{\Pi''}{2}}t_{UU'} = 2^{k-3}\phi_{c'}(x^{\Pi''})
\enspace ,
\end{align}
where $\Pi'' = \left(\Pi_H\setminus \{U_1, U_2\}\right)\cup \{U_1 \cup  U_2\}$ is the partition obtained by merging $U_1$ and $U_2$. 
The last equality follows from Lemma~\ref{lemma:subgraph-helping-lemma-1}. 
This contradicts $\max_{x\in \cp_{G[V_H]}}\phi_{c'}(x) = 0$. 
Therefore, this implies that there exists a subset $U \subset V_H$ with $i \in U$ and $j \in V_H \setminus U$ such that inequality \eqref{eq:subgraph-simplification-1} is fulfilled.

Let $U \subset V_H$ be a subset for which \eqref{eq:subgraph-simplification-1} holds. Therefore, we have that
\small{\begin{align}
&\phi_c(x') - \phi_c(x) \overset{\eqref{eq:subgraph-plugin-map-1}}{=} 
\sum_{pq\in E_H} c_{pq}(1- x_{pq}) + \sum_{pqr\in T_H}c_{pqr}(1-x_{pq}x_{pr}x_{qr}) - 
\sum_{\mathclap{pq\in \delta(V_H)}}c_{pq}x_{pq} -
\sum_{\mathclap{pqr\in T_{\delta(V_H)}}}c_{pqr}x_{pq}x_{pr}x_{qr} 
\\
& \overset{\eqref{eq:subgraph-simplification-1}}{\leq} \sum_{pq\in \delta(U, V_H \setminus U)}c_{pq} + \sum_{pqr\in T_{\delta(U, V_H \setminus U)}\cap T_H} c_{pqr} - \sum_{pq \in \delta(V_H)\cap E^-}c_{pq}- \sum_{pqr \in T_{\delta(V_H)}\cap T^-}c_{pqr} 
\overset{\eqref{eq:assumption-subgraph-criterion-uv-cuts}}{\leq} \;0.
\end{align}}%
\normalsize
Consequently, the map $\sigma$ is improving. 
By applying Corollary~\ref{lemma:persistency-variable}, we conclude the proof.
\end{proof}
\vspace{6pt}
We are unaware of an efficient method for finding subsets $V_H \subseteq V$ and $U \subseteq V$ for which the conditions~\eqref{eq:assumption-subgraph-criterion-uv-cuts} are satisfied. For subsets $V_H$ with $\vert V_H \vert \in \{2,3\}$, two corollaries of  Proposition~\ref{lemma:general-subgraph-edge-join} provide efficiently-verifiable partial optimality conditions:
\begin{corollary}
	\label{corollary:edge-subgraph-edge-join}
	Let $G = (V, E)$ a graph, $c\in \mathbb{R}^{I(G)}$ and $ij \in E$. 
	If  
	$c_{ij} \leq \sum_{pq\in \delta(ij)\cap E^-}c_{pq} + \sum_{pqr\in T_{\delta(ij)}\cap T^-}c_{pqr}$, 
	then there exists an optimal solution $x^*$ to \csp{G}{c} such that $x^*_{ij} = 1$.
\end{corollary}
\begin{corollary}
	\label{corollary:triplet-subgraph-edge-join}
	Let $G = (V, E)$ a graph, $c\in \mathbb{R}^{I(G)}$ and $ijk \in T$. 
	If 
\small{	\begin{gather}
		c_{ij} + c_{ik} \leq 0, \quad
		c_{ij} + c_{jk} \leq 0, \quad
		c_{ik} + c_{jk} \leq 0, \quad
		c_{ij} + c_{ik} + c_{jk} \leq 0, \quad
		c_{ij} + c_{ik} + c_{jk} + \frac 12 c_{ijk} \leq 0,
		\\
		c_{ij} + c_{ik} + c_{ijk} \leq \sum_{\substack{pq \in \delta(ijk) \\ pq \in E^-}} c_{pq} + \sum_{\substack{pqr \in T_{\delta(ijk)} \\ pqr \in T^-}} c_{pqr}, \quad
		c_{jk} + c_{ik} + c_{ijk} \leq \sum_{\substack{pq \in \delta(ijk) \\ pq \in E^-}} c_{pq}	+ \sum_{\substack{pqr \in T_{\delta(ijk)} \\ pqr \in T^-}} c_{pqr}
	\end{gather}}%
	\normalsize
	then there exists an optimal solution $x^*$ to \csp{G}{c} such that $x^*_{ik} = 1$.
\end{corollary}
We now present the last main partial optimality condition of this article.
\begin{proposition}
	\label{proposition:subset-join-proposition}
	Let $G = (V, E)$ a graph and $c\in \mathbb{R}^{I(G)}$. 
	Moreover, let $V_H \subseteq V$, $E_H = E\cap \tbinom{V_H}{2}$ and $T_H = T\cap \tbinom{V_H}{3}$. 
	If for every $ij \in E_H$ we have 
	\begin{align}\label{eq:subset-join-inequality}
		\footnotemath{%
			\max_{\substack{x\in \cp_G : \\ x_{ij} = 0}} \Bigl\{ 
		\hspace{-1ex}
		\smashoperator[r]{\sum_{pqr \in T_H}} c_{pqr}(1-x_{pq}x_{pr}x_{qr}) + 
		\sum_{\mathclap{pq\in E_H}}c_{pq}(1-x_{pq}) \Bigr\}  
		\leq \min_{\substack{x\in \cp_G : \\x_{ij} = 0}} 
		\Bigl\{\hspace{-1ex} 
		\smashoperator[r]{\sum_{pqr\in T_{\delta(V_H)}}}c_{pqr}x_{pq}x_{pr}x_{qr} + 
		\sum_{\mathclap{pq\in \delta(V_H)}}c_{pq}x_{pq} \Bigr\}
	}
	\end{align}%
	then there is an optimal solution $x^*$ to \csp{G}{c} such that $x^*_{ij} = 1$, $\forall ij \in E_H$.
\end{proposition}
\begin{proof}
We define $\sigma \colon \cp_G \to \cp_G$ such that 
\begin{equation}
\sigma(x) := \begin{cases}
x & \textnormal{if $x_{ij} = 1 \: \forall ij \in E_H$}\\
\left(\sigma_{V_H} \circ \sigma_{\delta(V_H)}\right)(x) & \textnormal{otherwise}
\end{cases}.
\end{equation}
Let $x' = \sigma(x)$ for every $x \in \cp_G$. 
Firstly, we have $x'_{ij} = 1$ for every $ij \in E_H$. 
Secondly, we show that $\sigma$ is an improving map. 
Let $x \in \cp_G$ such that $x_{ij} = 1$ for all $ij \in E_H$.
In this case, we have $\phi_c(x') = \phi_c(x)$ by definition of $x'$. 
Now, let us consider the complementary case, i.e.~let $x \in \cp_G$ such that there exists $ij \in E_H$ for which $x_{ij} = 0$. 
Then,
\begin{equation}
x'_{pq} = \begin{cases}
1 & \textnormal{if $pq \in E_H$}\\
0 & \textnormal{if $pq \in \delta(V_H)$}\\
x_{pq} & \textnormal{otherwise}
\end{cases}
\enspace .
\end{equation} 
Therefore, it follows that
\small{\begin{align}
&\phi_c(x') - \phi_c(x) =
\sum_{\mathclap{pq\in E_H}} c_{pq}(1-x_{pq}) - 
\smashoperator[r]{\sum_{pq\in \delta(V_H)}} c_{pq}x_{pq} + 
\smashoperator[r]{\sum_{pqr \in T_H}} c_{pqr}(1-x_{pq}x_{pr}x_{qr}) - 
\smashoperator[r]{\sum_{pqr\in T_{\delta(V_H)}}} c_{pqr}x_{pq}x_{pr}x_{qr}   
\\
&\leq 
\max_{\substack{x\in \cp_G : \\ x_{ij} = 0}} 
\Bigl\{ \smashoperator[r]{\sum_{pqr \in T_H}} c_{pqr}(1-x_{pq}x_{pr}x_{qr}) + 
\sum_{\mathclap{pq\in E_H}} c_{pq}(1-x_{pq}) \Bigr\} - 
\min_{\substack{x\in \cp_G : \\ x_{ij} = 0}}\Bigl\{ \smashoperator[r]{\sum_{pqr\in T_{\delta(V_H)}}} c_{pqr}x_{pq}x_{pr}x_{qr} + 
\sum_{\mathclap{pq\in \delta(V_H)}}c_{pq}x_{pq} \Bigr\} \\
&\overset{\eqref{eq:subset-join-inequality}}{\leq} 0.\qedhere
\end{align}}\normalsize%
\end{proof}%
We are unaware of an efficient method for deciding~\eqref{eq:subset-join-inequality} for arbitrary subsets $V_H \subseteq V$ and costs $c \in \mathbb{R}^{I(G)}$. 
Yet, Corollary~\ref{corollary:subset-join-all-pairs-at-once} below describes one setting in which a suitable subset can be searched for heuristically, in polynomial time. 
\begin{corollary}
	\label{corollary:subset-join-all-pairs-at-once}
	Let $G = (V, E)$ a graph and $c\in \mathbb{R}^{I(G)}$. 
	Moreover, let $V_H \subseteq V$, $E_H = E\cap \tbinom{V_H}{2}$ and $T_H = T\cap \tbinom{V_H}{3}$. 
	If 
	\begin{gather}
		\smallmath{c_{pq} \leq 0 \qquad \forall pq \in E_H,}
		\label{eq:cond-submod-1}
		\\
		\smallmath{c_{pqr} \leq 0 \qquad \forall pqr\in T_H,}
		\label{eq:cond-submod-2} \\
		\smallmath{\max_{\substack{U\subset V_H : \\ U \neq \emptyset}} \Bigl\{ \sum_{pqr\in T_{\delta(U)}\cap T_H}c_{pqr} + \sum_{pq\in \delta(U, V_H \setminus U)}c_{pq} \Bigr\} 
		\leq \quad \sum_{pqr\in T_{\delta(V_H)}\cap T^-}c_{pqr} + \sum_{pq\in \delta(V_H)\cap E^-}c_{pq} \enspace ,}
		\label{eq:subset-join-corollary-all-pairs-condition}
	\end{gather}%
	then there is an optimal solution $x^*$ to \csp{G}{c} such that $x^*_{ij} = 1$, $\forall ij\in E_H$.
\end{corollary}
\begin{proof}
	The claim follows from the following three observations. 
	Let \eqref{eq:cond-submod-1} and \eqref{eq:cond-submod-2} be satisfied. 
	Firstly, for any $ij \in E_H$, we have that the left-hand side of \eqref{eq:subset-join-inequality} is equal to 
	$\max_{\substack{U\subsetneq V_H : \\ i \in U, \, j \not\in U}}\sum_{pqr \in T_{\delta(U)}\cap T_H}c_{pqr} + \sum_{pq \in \delta(U, V_H \setminus U)} c_{pq}.$
	This means that the maximizer is given by a feasible $x\in \cp_G$ whose restriction to $E_H$ corresponds to a partition $\Pi$ of 
	$V_H$ into two subsets. 
	To see this, note that for any $ij\in E_H$, instead of maximizing the left-hand side of \eqref{eq:subset-join-inequality} over $x\in X_G$ such that $x_{ij} = 0$ we can equivalently maximize over all $x'\in X_{G[V_H]}$ such that $x'_{ij} = 0$. 
	Now, let us assume that there exists $ij \in E_H$ such that the maximizer is given by a feasible $x' \in \cp_{G[V_H]}$ corresponding to a partition $\Pi'$ of 
	$V_H$ into more than two subsets.
	Without loss of generality, let $U_1, U_2 \in \Pi'$ such that $i\in U_1$, $j\in U_2$. 
	Then, the vector $x'' \in \cp_{G[V_H]}$ corresponding to the partition $\Pi'' = \left\{ U_1, \bigcup_{U \in \Pi' \setminus \{U_1\}} U \right\}$ has objective value at least the objective value of $x'$.
	This follows from the facts that all the costs are non-positive and that $\Pi'$ is a refinement of $\Pi''$.
	Secondly, by using the trivial lower bound of the right-hand side of \eqref{eq:subset-join-inequality} we have that
	it is at least 
	$\sum_{pqr\in T_{\delta(V_H)}\cap T^-}c_{pqr} + \sum_{pqr\in \delta(V_H)\cap E^-}c_{pq}$.
	Thirdly, for any $ij \in E_H$ we have
	\small{\begin{align}
		\max_{\substack{U\subset V_H : \\ U \neq \emptyset}} \Bigl\{ \sum_{pqr\in T_{\delta(U)}\cap T_H}c_{pqr} + \sum_{pq\in \delta(U, V_H \setminus U)}c_{pq} \Bigr\} \geq \max_{\substack{U\subsetneq V_H : \\ i \in U, \, j \not\in U}}\sum_{pqr \in T_{\delta(U)}\cap T_H}c_{pqr} + \sum_{pq \in \delta(U, V_H \setminus U)} c_{pq}\enspace.
	\end{align}}\normalsize
	These three observations, together with satisfaction of \eqref{eq:cond-submod-1}--\eqref{eq:subset-join-corollary-all-pairs-condition}, imply that~\eqref{eq:subset-join-inequality} holds for every $ij\in E_H$. Hence, Corollary~\ref{corollary:subset-join-all-pairs-at-once} specializes Proposition~\ref{proposition:subset-join-proposition} and the claim follows.
\end{proof}
\section{Deciding Partial Optimality Efficiently}\label{sec: testing}
Next, we describe algorithms for deciding the partial optimality conditions introduced in the previous section.
This includes exact algorithms and heuristics, and we discuss their time complexity.
We denote by 
$\degree{G} = \max_{p\in V} \lvert\{q\in V\setminus \{p\}\mid pq\in E\}\rvert$, and 
$\degreethree{G} = \max_{pq\in E} \lvert \{r\in V\setminus \{p, q\}\mid pqr\in T\}\rvert$,
the maximum number of adjacent vertices of a vertex or the maximum number of triples an edge is a part of.

We start by examining Proposition~\ref{lemma:persistency-subset-separation}. 
Observe that it suffices to find any subset $U \subseteq V$ that fulfills Proposition~\ref{lemma:persistency-subset-separation}. 
Because of that, we focus our attention on minimal subsets, with respect to set inclusion.
\begin{proposition}
	Let $G = (V, E)$ a graph and $c\in \mathbb{R}^{I(G)}$. The minimal subsets $U \subseteq V$ with respect to set inclusion that satisfy \eqref{eq:edge-cut-condition-1}--\eqref{eq:edge-cut-condition-2} are exactly the vertex sets of the connected components of the graph $G' = (V, E')$ 
    with 
	$E' = \left\{pq\in E | c_{pq} < 0 \lor \exists r \in V 
	\text{ s.t. } 
	pqr\in T \colon c_{pqr} < 0\right\}$.
\end{proposition}%
\begin{proof}
	First we show that every $U \subseteq V$ inducing a connected component in $G'$ is such that \eqref{eq:edge-cut-condition-1}--\eqref{eq:edge-cut-condition-2} are satisfied. 
	Let us assume that there is some $U \subseteq V$  inducing a connected component in $G'$, but \eqref{eq:edge-cut-condition-1}--\eqref{eq:edge-cut-condition-2} do not hold. 
	In this case, there is some $pq\in \delta(U)$ such that $c_{pq} < 0$ or $pqr\in T_{\delta(U)}$ such that $c_{pqr} < 0$. Without loss of generality, let us assume that $p\notin U$ and $q\in U$ in both cases. Then we have that $pq\in E'$ by the definition of $E'$. Since $U \cup \{p\}$ also induces a connected subgraph of $G'$, this contradicts the assumption that $U$ induces a (maximal) connected component of $G'$.
		
	Second we show that every minimal $U \subseteq V$ with respect to set inclusion that satisfies \eqref{eq:edge-cut-condition-1}--\eqref{eq:edge-cut-condition-2} induces a connected component of $G'$.
	For the sake of contradiction, let us assume that $U$ does not induce a connected component of $G'$. 
	By definition of $E'$, we have that  
	$\delta(U) \cap E' = \emptyset$, since $U$ satisfies \eqref{eq:edge-cut-condition-1}--\eqref{eq:edge-cut-condition-2}. 
	Therefore every connected component of the subgraph $G'[U]$ is also a connected component of the underlying graph $G'$.
	If $G'[U]$ has exactly one connected component, then $U$ induces a connected component of $G'$, which contradicts our assumption.
	We now consider the case that the subgraph $G'[U]$ has more than one connected components. 
		
	Without loss of generality, let  $U'\subset U$ be a vertex set inducing a connected component in $G'[U]$ and thus also in $G'$. For all edges $pq\in \delta(U')$ we have that $pq\notin E'$ because otherwise $U'$ would not induce a connected component of $G'$. Therefore for all $pq\in \delta(U')$ we have that $c_{pq} > 0$ and $c_{pqr} > 0$ for all $pqr\in T$. Since $T_{\delta(U')}$ is the union of all sets $\{pqr\in T\mid pq\in \delta(U')\}$, we have that $U'$ fulfills~\eqref{eq:edge-cut-condition-1}--\eqref{eq:edge-cut-condition-2}, 
	which contradicts the minimality of $U$ with respect to set inclusion.
\end{proof}
\vspace{6pt}

Hence, we can use any algorithm that computes connected components in a graph. 
In this work we use a disjoint set data structure with path compression and union-by-size, starting with at most $|V|$ singleton sets, proceeding with at most $|E|$ called of union, followed by precisely $|V|$ calls of find.
This has time complexity $\mathcal{O}((|V| + |E|) \alpha(|V|))$ \citep{tarjan-1984} where $\alpha$ denotes the inverse Ackermann function.

Partial optimality according to Propositions~\ref{proposition:edge-cut-persistency}--\ref{lemma:persistency-triangle-edge-join} is conditional to the existence of an edge $ij \in E$ or a triple $ijk \in T$, together with a subset $U \subseteq V$ (in case of Proposition~\ref{lemma:persistency-triangle-edge-join} even a second subset $U' \subseteq V$ independent of $U$) 
for which 
specific conditions 
are satisfied, namely \eqref{eq:assumption-edge-cut-inequality}--\eqref{eq:triangle-edge-join-3}.
For every triple, we decide~\eqref{eq:triangle-edge-join-3} explicitly by enumerating all edges incident to one of the triple vertices and then enumerating all triples this edge is a part of, in time 
$\mathcal{O}(\degree{G} \cdot \degreethree{G})$.
For every pair or triple, we reduce the search of 
subsets $U$ or $U'$ satisfying \eqref{eq:assumption-edge-cut-inequality}--\eqref{eq:triangle-edge-join-2} \emph{with maximum margin} 
to the min $st$-cut problem (in Appendix~\ref{sec:tech-min-cut}). 
In order to decide partial optimality efficiently, we solve the dual max $st$-flow problems by the push-relabel algorithm with FIFO vertex selection rule \citep{goldberg-1988}. 
Note that any subsets $U$, or $U'$, that satisfy \eqref{eq:assumption-edge-cut-inequality}--\eqref{eq:triangle-edge-join-2} would suffice. 
We choose to use an exact method that returns subsets $U$ or $U'$ that fulfill \eqref{eq:assumption-edge-cut-inequality}--\eqref{eq:triangle-edge-join-2} with maximum margin.

As mentioned already in Section~\ref{section:partial-optimality-criteria-joins},
we are unaware of an efficient method for finding subsets that satisfy the conditions of Proposition~\ref{lemma:general-subgraph-edge-join} or \ref{proposition:subset-join-proposition}.
Regarding Proposition~\ref{lemma:general-subgraph-edge-join}, we resort to the special case of Corollary~\ref{corollary:edge-subgraph-edge-join} that we test for each pair in time $\mathcal{O}(\degree{G}\cdot \degreethree{G})$, 
and to the special case of Corollary~\ref{corollary:triplet-subgraph-edge-join} that we test for each triple in time $\mathcal{O}(\degree{G}\cdot \degreethree{G})$.
Regarding Proposition~\ref{proposition:subset-join-proposition}, we employ 
the special case of Corollary~\ref{corollary:subset-join-all-pairs-at-once} and search heuristically for a witness $U$ of \eqref{eq:subset-join-corollary-all-pairs-condition}, as follows:
In an outer loop, we iterate over all pairs 
$ij\in E$ with $c_{ij} \leq 0$, and initialize $U = \{i, j\}$.
For each of these initializations of $U$, we add elements to $U$ for which the costs of all pairs and triples inside $U$ is non-positive.
Upon termination of this inner loop, we take $U$ to be a candidate.
By construction of $U$, all coefficients on the left-hand side of \eqref{eq:subset-join-corollary-all-pairs-condition} are non-positive. 
By applying Proposition~\ref{prop: reduction costs} to the left-hand side of \eqref{eq:subset-join-corollary-all-pairs-condition}, this problem takes the form of a min cut problem with non-negative capacities that we solve exactly using the algorithm by \citet{stoer-1997}. 
As soon as we find a candidate $U$ that fulfills~\eqref{eq:subset-join-corollary-all-pairs-condition}, we merge the nodes in $U$ as described in Section~\ref{sec: join-conditions} and obtain a smaller instance of~$\ccp$. Then, we proceed as described in Section~\ref{section:mixing-conditions}.
\section{Combining Partial Optimality Conditions}\label{sec: combining}
Next, we discuss how we apply partial optimality conditions iteratively and why this requires special attention. 
Let $G = (V, E)$ a graph and $c\in \mathbb{R}^{I(G)}$. 
Furthermore, let $Q_1, Q_2 \subseteq \cp_{G}$. 
If there is an optimal solution $x^*_1\in \cp_G$ to \csp{G}{c} such that $x^*_1\in Q_1$, and an optimal solution $x^*_{2}\in X_G$ to \csp{G}{c} such that $x^*_2\in Q_2$, then there is not necessarily an optimal solution $x^*\in \cp_G$ to \csp{G}{c} such that $x^*\in Q_1 \cap Q_2$. Thus, we cannot na\"ively apply the partial optimality conditions simultaneously (see Example~\ref{example:simultaneous-po}).
\begin{example}
	\label{example:simultaneous-po}
	Let $V = \{1, 2, 3\}$, $E = \{ 12, 13, 23 \}$, $T = \{ 123 \}$ and $c \in \mathbb{R}^{I(G)}$ such that $c_{123} = 5$, $c_{12}=c_{13}=c_{23} = -2$, $c_\emptyset = 0$.
	Then, $\min_{x\in \cp_G}\phi_c(x) = -2$.  
	Furthermore, let $Q_1 = \{x\in \cp_G \mid x_{12} = 1\}$ and $Q_2 = \{x\in \cp_G \mid x_{13} = 1\}$. 
	It follows that $Q_1 \cap Q_2 = \{ x \in \cp_G \mid x_{12} = 1, x_{13} = 1 \} = \{ (1, 1, 1) \}$.
	The set of optimal solutions is the set of all $x\in \cp_G$ for which there is exactly one $pq\in E$ with $x_{pq} = 1$. 
	Thus, the feasible $x'\in \cp_G$ such that $x'_{12} = x'_{13} = x'_{23} = 1$ is not optimal.
\end{example}%
\subsection{Cut Conditions}
\label{section:cut-conditions-parallel}
	Let us assume that we have found partially optimal assignments $E_0 \subseteq E$ by applying Proposition~\ref{proposition:edge-cut-persistency} and partially optimal constraints $T_0 \subseteq T$ by applying Proposition~\ref{lemma:persistency-triplet-cut}, i.e., we know that there exists a solution $x^*$ to $\ccp$ such that $x^*_{pq} = 0$ for all $pq\in E_0$ and $x^*_{pq}x^*_{qr}x^*_{pr} = 0$ for all $pqr\in T_0$. Moreover, we denote by 
	$X_G^{E_0, T_0} = \{x\in X_G\mid \forall pq\in E_0\colon x_{pq} = 0 \land \forall pqr\in T_0\colon x_{pq}x_{pr}x_{qr} = 0\}$
	the set of feasible vectors $x\in X_G$ that fulfill the partially optimal assignments and partially optimal constraints.
	Thus we have 
	$\min_{x\in X_G}\phi_c(x) = \min_{x\in X_G^{E_0, T_0}} \phi_c(x).$
	In order not to lose already established partial optimality, we may from now on only apply improving maps $\sigma'\colon X_G^{E_0, T_0} \to X_G^{E_0, T_0}$.
 	This is the case for the elementary cut map restricted to $X_G^{E_0, T_0}$, i.e., $\sigma_{\delta(U)}(X_G^{E_0, T_0}) \subseteq X_G^{E_0, T_0}$ for every $U\subseteq V$. Since the cut conditions Proposition~\ref{proposition:edge-cut-persistency} and Proposition~\ref{lemma:persistency-triplet-cut} use only the elementary cut map or the identity map to improve solutions, we can apply these cut conditions simultaneously.
	
 	On the contrary, it can happen that the image of $X_G^{E_0, T_0}$ under $\sigma_U$ is not a subset of $X_G^{E_0, T_0}$ for some $U\subseteq V$, i.e., $\sigma_U(X_G^{E_0, T_0}) \nsubseteq X_G^{E_0, T_0}$. 
 	In particular, $\sigma_V(X_G^{E_0, T_0}) = \{\mathbbm{1}_V\}\nsubseteq X_G^{E_0, T_0}$ for non-empty $E_0 \neq \emptyset$ and $T_0\neq \emptyset$.
 	Therefore, we cannot apply join conditions if we want to maintain already-established partial optimality due to Propositions~\ref{proposition:edge-cut-persistency} and \ref{lemma:persistency-triplet-cut}.
\subsection{Join Conditions}
\label{sec: join-conditions}
Let us assume the existence of an optimal solution $x^*$ to \csp{G}{c} such that $x^*_{ij} = 1$ for some $ij\in E$. 
We define $\cp_G \vert_{x_{ij} = 1} = \{x\in \cp_G \mid x_{ij} = 1\}$.
Then, we have 
\begin{equation}
	\label{eq:one-persistency-minimization-problem}
	\min_{x\in \cp_G} \phi_c(x) = \min_{x \in \cp_G \vert_{x_{ij} = 1}}\phi_c(x)
	\enspace .
\end{equation}
Let $V' = V \setminus \{j\}, E' = \left(E\cap \tbinom{V'\setminus \{i\}}{2}\right)\cup \{pi\in \tbinom{V'}{2}\mid pi\in E\lor pj\in E\}$ and $G' = (V', E')$, that is, we contract $i$ and $j$. 
Now, we relate feasible vectors of $\cp_G \vert _{x_{ij} = 1}$ to feasible vectors of $\cp_{G'}$. 
We observe that for any $x \in \cp_G \vert_{x_{ij} = 1}$ we have $\forall p\in V \setminus \{i, j\} \colon x_{pi} = x_{pj}$ whenever $pi, pj \in E$ by transitivity. 
We define $\varphi_{ij}\colon \cp_G \vert_{x_{ij} = 1} \to \cp_{G'}$ as 
\begin{align}
	\varphi_{ij}(x)_{pi} &= \begin{cases}
			x_{pi} & \textnormal{if $pi\in E \land pj\notin E$}\\
			x_{pj} & \textnormal{if $pi\notin E \land pj\in E$}\\
			x_{pi} = x_{pj} & \textnormal{if $pi\in E\land pj\in E$}
		\end{cases}
		& \forall p\in V'\setminus \{i\} \colon pi \in E'
	\\
	\varphi_{ij}(x)_{pq} &= x_{pq} &  \forall pq \in E \cap \tbinom{V'\setminus \{i\}}{2} \enspace .
\end{align}
It is easy to see that $\varphi_{ij}$ is bijective.
Proposition~\ref{lemma:contraction-cost-adjustments} below shows that solving the right-hand side of \eqref{eq:one-persistency-minimization-problem} is equivalent to solving a smaller instance of the original problem. 
\begin{proposition}
	\label{lemma:contraction-cost-adjustments}
	Let $G = (V, E)$ a graph and $c\in \mathbb{R}^{I(G)}$. 
	Moreover, let $ij \in E$ and $V' = V \setminus \{j\}$, $E' = \left(E\cap \tbinom{V'\setminus \{i\}}{2}\right)\cup \{pi\in \tbinom{V'}{2}\mid pi\in E\lor pj\in E\}$ and $G' = (V', E')$. Define $c'\in \mathbb{R}^{I(G')}$ such that
	\begin{align}
	c'_{pqr} &= c_{pqr} & \forall pqr \in T \cap \tbinom{V'\setminus \{i\}}{3} \\
	c'_{pqi} &= \sum_{\substack{r\in \{i, j\} \colon pqr\in T}}c_{pqr} & \forall pq \in E \cap \tbinom{V'\setminus \{i\}}{2} \colon pqi \in T \lor pqj \in T\\
	c'_{pq} &= c_{pq} & \forall pq \in E \cap \tbinom{V'\setminus \{i\}}{2} \\
	c'_{pi} &= \sum_{\substack{q\in \{i, j\}\colon pq \in E}}c_{pq} & \forall p \in V'\setminus \{i\} \colon \left(pi \in E \lor pj\in E\right) \land pij\notin T\\
	c'_{pi} &= c_{pi} + c_{pj}+ c_{pij} & \forall p \in V'\setminus \{i\} \colon pij\in T \\
	c'_\emptyset & = c_\emptyset + c_{ij}
	\enspace .
	\end{align}
	Furthermore, let $\varphi_{ij}\colon \cp_G\vert_{x_{ij} = 1}\to \cp_{G'}$ the map that relates feasible vectors of $\cp_G\vert_{x_{ij} = 1}$ to feasible vectors of $\cp_{G'}$.
	Then
	$\min_{x\in \cp_G\vert_{x_{ij} = 1}} \phi_c(x) = \min_{x\in \cp_{G'}}\phi_{c'}(x).$
	Moreover, if $x^* \in \argmin_{x\in \cp_G\vert_{x_{ij} = 1}}\phi_c(x)$, then $\varphi_{ij}(x^*)\in \argmin_{x\in \cp_{G[V']}}\phi_{c'}(x)$.
\end{proposition}
\begin{proof}
Let $x\in \cp_G \vert_{x_{ij} = 1}$ and $x' = \varphi_{ij}(x)$.
We show that $\phi_c(x) = \phi_{c'}(x')$. 
Note that we have
\begin{align}
	T' &= \{pqr\in \tbinom{V'}{3}\mid pq\in E'\land pr\in E'\land qr\in E'\} \\ 
	&= \left(T\cap \tbinom{V'\setminus \{i\}}{3}\right)\cup \left\{pqi\in \tbinom {V'}3\mid pq\in \left(E\cap \tbinom {V'\setminus \{i\}}2\right) \land \left(pqi\in T \lor pqj\in T\right)\right\}
\end{align}
We use the fact that $x_{pi} = x_{pj}$ for all $p\in V \setminus \{i, j\}$ such that $pi, pj \in E$, and $x_{ij} = 1$. 
It follows 
\small{
\begin{align}
&\phi_{c'}(x') = \sum_{pqr\in T'}c'_{pqr}x'_{pq}x'_{pr}x'_{qr} + \sum_{pq\in E'} c'_{pq}x'_{pq} + c'_{\emptyset}\\
= &
\sum_{pqr\in T\cap \tbinom{V'\setminus \{i\}}{3}}c'_{pqr}x'_{pq}x'_{pr}x'_{qr} + 
\sum_{\substack{pq\in E\cap \tbinom{V'\setminus \{i\}}{2} : \\ pqi \in T\lor pqj\in T}}c'_{pqi}x'_{pq}x'_{pi}x'_{pj} + 
\smashoperator[r]{\sum_{pq\in E\cap \tbinom{V'\setminus \{i\}}{2}}} c'_{pq}x'_{pq} + 
\smashoperator[r]{\sum_{\substack{p\in V'\setminus \{i\} : \\ pi \in E\lor pj\in E}}}c'_{pi}x'_{pi} + c'_{\emptyset}\\
= &
\sum_{\substack{pq \in E \cap \tbinom{V \setminus\{i, j\}}{2}  :  \\ pqi \in T \lor pqj\in T}} c'_{pqi}x'_{pi}x'_{qi}x'_{pq} 
+ 
\sum_{pqr\in T \cap \tbinom{V \setminus\{i, j\}}{3}}c'_{pqr}x'_{pq}x'_{pr}x'_{qr} 
+
 \sum_{\mathclap{\substack{p\in V \setminus\{i, j\} : \\ pi \in E \lor pj\in E}}}c'_{pi}x'_{pi} + 
 \smashoperator[r]{\sum_{pq \in E \cap \tbinom{V\setminus\{i, j\}}{2}}}c'_{pq}x'_{pq} + c'_{\emptyset}\\
= &\sum_{\substack{pq \in E \cap \tbinom{V \setminus\{i, j\}}{2} : \\ pqi \in T}}c_{pqi}x_{pi}x_{qi}x_{pq} 
+
\sum_{\substack{pq \in E \cap \tbinom{V \setminus\{i, j\}}{2} : \\ pqj \in T}}c_{pqj}x_{pj}x_{qj}x_{pq}
+
\sum_{pqr\in T \cap \tbinom{V \setminus\{i, j\}}{3}}c_{pqr}x_{pq}x_{pr}x_{qr} \\
& + 
\sum_{pq \in E \cap \tbinom{V \setminus\{i, j\}}{2}}c_{pq}x_{pq} 
+ 
\sum_{\substack{p\in V \setminus\{i, j\} : \\ pij \in T}}c_{pij}x_{pi}
+
\sum_{\substack{p\in V \setminus\{i, j\} : \\ pi\in E}}c_{pi} x_{pi}
+ 
\sum_{\substack{p\in V \setminus\{i, j\} : \\ pj\in E}}c_{pj} x_{pj}  + c_{\emptyset} + c_{ij}x_{ij}\\
= &\sum_{pqr\in T}c_{pqr}x_{pq}x_{pr}x_{qr} + \sum_{pq \in E}c_{pq}x_{pq} + c_{\emptyset} = \phi_c(x).
\end{align}}%
\normalsize
Therefore, we have
$\min_{x\in \cp_G \vert_{x_{ij} = 1}} \phi_c(x) = \min_{x\in \cp_G \vert_{x_{ij} = 1}} \phi_{c'}(\varphi_{ij}(x)) = \min_{x\in \cp_{G[V']}}\phi_{c'}(x)$.
\end{proof}
\subsection{Mixing Cut and Join Conditions}
\label{section:mixing-conditions}
Here, we describe how we apply the partial optimality properties recursively. As soon as a condition leads to a smaller instance, we start the procedure again on the smaller set (in case of a join) or sets (in case of Proposition~\ref{lemma:persistency-subset-separation}). 
Firstly, we apply Proposition~\ref{lemma:persistency-subset-separation}, which leaves us with independent sub-problems. 
Secondly, we apply our join conditions until we find an edge $ij \in E$ or triple $ijk \in  T$ to join, starting from Corollary~\ref{corollary:subset-join-all-pairs-at-once} and then moving on to Propositions~\ref{lemma:edge-join-persistency} and~\ref{lemma:persistency-triangle-edge-join}, Corollaries~\ref{corollary:edge-subgraph-edge-join} and~\ref{corollary:triplet-subgraph-edge-join} and Proposition~\ref{proposition:triplet-join}, in this order.
Thirdly, we apply the remaining cut conditions, which can be applied jointly as we have seen in Section~\ref{section:cut-conditions-parallel}.
\section{Experiments}\label{sec: experiments}
We examine the effect of the algorithms empirically on two data sets.
For both of them, we report the percentage of fixed variables and triples, as well as the runtime. 
More specifically, we report the median as well as lower and upper quartile over 20 instances.
We apply all partial optimality conditions jointly, as described in Section~\ref{section:mixing-conditions},
and we also evaluate the effect of each condition separately. 
All algorithms are implemented in C++ and run on one core of an Intel Core i9-12900KF equipped with 64 GB of RAM.
Note that compared to~\cite{stein-2023} the hardware has changed. In the interest of objectivity, we run the experiments from~\cite{stein-2023} on the new hardware and report the results in Appendix~\ref{sec:reproduction-experiments-conference-article}.
\subsection{Partition Data Set}\label{section:experiments-partition}
\input{figure-sparse-partition-joint}
\input{figure-sparse-partition-independent-subset-join-edge-cut-alphas}
\tikzmath{\plotheight=4.0;}
\tikzmath{\plotwidth=4.6;}

\providecommand{\addnodeplot}{}
\renewcommand{\addnodeplot}[1]{\addplot+[
	only marks,
	mark=*,
	mark options={draw=none}
	] table[
	col sep=comma,
	scale only axis,
	x expr=\thisrow{alpha}, 
	y expr=100*\thisrow{medianEliminatedNodes},
	y error minus expr=100*\thisrow{medianEliminatedNodes} - 100*\thisrow{q25EliminatedNodes},
	y error plus expr=100*\thisrow{q75EliminatedNodes} - 100*\thisrow{medianEliminatedNodes} 
	] {#1};}

\providecommand{\addvariableplot}{}	
\renewcommand{\addvariableplot}[1]{\addplot+[
	only marks,
	mark=*,
	mark options={draw=none}
	] table[
	col sep=comma,
	x expr=\thisrow{alpha}, 
	y expr=\thisrow{medianEliminatedVariables}*100,
	y error minus expr=100*(\thisrow{medianEliminatedVariables} - \thisrow{q25EliminatedVariables}),
	y error plus expr=100*(\thisrow{q75EliminatedVariables} - \thisrow{medianEliminatedVariables})
	] {#1};}

\providecommand{\addtriangleplot}{}	
\renewcommand{\addtriangleplot}[1]{\addplot+[
	only marks,
	mark=*,
	mark options={draw=none}
	] table[
	col sep=comma,
	x expr=\thisrow{alpha}, 
	y expr=100*\thisrow{medianEliminatedTriangles},
	y error minus expr=100*\thisrow{medianEliminatedTriangles} - 100*\thisrow{q25EliminatedTriangles},
	y error plus expr=100*\thisrow{q75EliminatedTriangles} - 100*\thisrow{medianEliminatedTriangles} 
	] {#1};}

\providecommand{\addruntimeplot}{}	
\renewcommand{\addruntimeplot}[1]{\addplot+[
	only marks,
	mark=*,
	mark options={draw=none}
	] table[
	col sep=comma,
	x expr=\thisrow{alpha}, 
	y expr=\thisrow{medianDuration} / 1e9,
	y error minus expr=(\thisrow{medianDuration} - \thisrow{q25Duration}) / 1e9,
	y error plus expr=(\thisrow{q75Duration} - \thisrow{medianDuration})  / 1e9
	] {#1};}

\begin{figure}[!t]
	\centering
	
	\begin{tikzpicture}[baseline]
		\begin{groupplot}[group style={group size= 4 by 2, horizontal sep=0.5cm, vertical sep=0.5cm}]
			
			\nextgroupplot[
			title={$(0.25, 184, 4224)$},
			title style={align=center},
			ylabel={Triples [\%]},
			xmin=0.25,
			xmax=0.7,
			ymin=0,
			ymax=100,
			width=\plotwidth cm,
			height=\plotheight cm,
			legend style={font=\tiny},
			legend pos=north east,
			xticklabels=\empty
			]
			\addtriangleplot{./data/sparse/partition-triangle-cut/alphas/n184_beta0.00_edgePercentage0.25_stats.csv}
			\addlegendentry{$\beta = 0.00$};
			\addtriangleplot{./data/sparse/partition-triangle-cut/alphas/n184_beta0.01_edgePercentage0.25_stats.csv}
			\addlegendentry{$\beta = 0.01$};
			\addtriangleplot{./data/sparse/partition-triangle-cut/alphas/n184_beta0.50_edgePercentage0.25_stats.csv}
			\addlegendentry{$\beta = 0.50$};
			\addtriangleplot{./data/sparse/partition-triangle-cut/alphas/n184_beta1.00_edgePercentage0.25_stats.csv}
			\addlegendentry{$\beta = 1.00$};
			
			\nextgroupplot[
			title={$(0.50, 104, 2675.5)$},
			title style={align=center},
			xmin=0.25,
			xmax=0.7,
			ymin=0,
			ymax=100,
			yticklabel=\empty,
			width=\plotwidth cm,
			height=\plotheight cm,
			xticklabels=\empty
			]
			\addtriangleplot{./data/sparse/partition-triangle-cut/alphas/n104_beta0.00_edgePercentage0.50_stats.csv}
			\addtriangleplot{./data/sparse/partition-triangle-cut/alphas/n104_beta0.01_edgePercentage0.50_stats.csv}
			\addtriangleplot{./data/sparse/partition-triangle-cut/alphas/n104_beta0.50_edgePercentage0.50_stats.csv}
			\addtriangleplot{./data/sparse/partition-triangle-cut/alphas/n104_beta1.00_edgePercentage0.50_stats.csv}
			
			\nextgroupplot[
			title={$(0.75, 72, 1918)$},
			title style={align=center},
			xmin=0.25,
			xmax=0.7,
			ymin=0,
			ymax=100,
			yticklabel=\empty,
			width=\plotwidth cm,
			height=\plotheight cm,
			xticklabels=\empty
			]
			\addtriangleplot{./data/sparse/partition-triangle-cut/alphas/n72_beta0.00_edgePercentage0.75_stats.csv}	
			\addtriangleplot{./data/sparse/partition-triangle-cut/alphas/n72_beta0.01_edgePercentage0.75_stats.csv}
			\addtriangleplot{./data/sparse/partition-triangle-cut/alphas/n72_beta0.50_edgePercentage0.75_stats.csv}
			\addtriangleplot{./data/sparse/partition-triangle-cut/alphas/n72_beta1.00_edgePercentage0.75_stats.csv}
			
			\nextgroupplot[
			title={$(1.00, 56, 1540)$},
			title style={align=center},
			xmin=0.25,
			xmax=0.7,
			ymin=0,
			ymax=100,
			yticklabel=\empty,
			width=\plotwidth cm,
			height=\plotheight cm,
			xticklabels=\empty
			]
			\addtriangleplot{./data/sparse/partition-triangle-cut/alphas/n56_beta0.00_edgePercentage1.00_stats.csv}	
			\addtriangleplot{./data/sparse/partition-triangle-cut/alphas/n56_beta0.01_edgePercentage1.00_stats.csv}
			\addtriangleplot{./data/sparse/partition-triangle-cut/alphas/n56_beta0.50_edgePercentage1.00_stats.csv}
			\addtriangleplot{./data/sparse/partition-triangle-cut/alphas/n56_beta1.00_edgePercentage1.00_stats.csv}
			
			
			\nextgroupplot[
			ylabel={Runtime [s]},
			xlabel=$\alpha$,
			xmin=0.25,
			xmax=0.7,
			ymin=0,
			ymax=60,
			width=\plotwidth cm,
			height=\plotheight cm
			]
			\addruntimeplot{./data/sparse/partition-triangle-cut/alphas/n184_beta0.00_edgePercentage0.25_stats.csv}
			\addruntimeplot{./data/sparse/partition-triangle-cut/alphas/n184_beta0.01_edgePercentage0.25_stats.csv}
			\addruntimeplot{./data/sparse/partition-triangle-cut/alphas/n184_beta0.50_edgePercentage0.25_stats.csv}
			\addruntimeplot{./data/sparse/partition-triangle-cut/alphas/n184_beta1.00_edgePercentage0.25_stats.csv}
			
			\nextgroupplot[
			xmin=0.25,
			xmax=0.7,
			ymin=0,
			ymax=60,
			yticklabel=\empty,
			width=\plotwidth cm,
			height=\plotheight cm,
			xlabel=$\alpha$
			]
			\addruntimeplot{./data/sparse/partition-triangle-cut/alphas/n104_beta0.00_edgePercentage0.50_stats.csv}
			\addruntimeplot{./data/sparse/partition-triangle-cut/alphas/n104_beta0.01_edgePercentage0.50_stats.csv}
			\addruntimeplot{./data/sparse/partition-triangle-cut/alphas/n104_beta0.50_edgePercentage0.50_stats.csv}
			\addruntimeplot{./data/sparse/partition-triangle-cut/alphas/n104_beta1.00_edgePercentage0.50_stats.csv}
			
			\nextgroupplot[
			xmin=0.25,
			xmax=0.7,
			ymin=0,
			ymax=60,
			yticklabel=\empty,
			width=\plotwidth cm,
			height=\plotheight cm,
			xlabel=$\alpha$
			]
			\addruntimeplot{./data/sparse/partition-triangle-cut/alphas/n72_beta0.00_edgePercentage0.75_stats.csv}	
			\addruntimeplot{./data/sparse/partition-triangle-cut/alphas/n72_beta0.01_edgePercentage0.75_stats.csv}
			\addruntimeplot{./data/sparse/partition-triangle-cut/alphas/n72_beta0.50_edgePercentage0.75_stats.csv}
			\addruntimeplot{./data/sparse/partition-triangle-cut/alphas/n72_beta1.00_edgePercentage0.75_stats.csv}
			
			\nextgroupplot[
			xmin=0.25,
			xmax=0.7,
			ymin=0,
			ymax=60,
			yticklabel=\empty,
			width=\plotwidth cm,
			height=\plotheight cm,
			xlabel=$\alpha$
			]
			\addruntimeplot{./data/sparse/partition-triangle-cut/alphas/n56_beta0.00_edgePercentage1.00_stats.csv}	
			\addruntimeplot{./data/sparse/partition-triangle-cut/alphas/n56_beta0.01_edgePercentage1.00_stats.csv}
			\addruntimeplot{./data/sparse/partition-triangle-cut/alphas/n56_beta0.50_edgePercentage1.00_stats.csv}
			\addruntimeplot{./data/sparse/partition-triangle-cut/alphas/n56_beta1.00_edgePercentage1.00_stats.csv}		
		\end{groupplot}
	\end{tikzpicture}
	\\[-2ex]
	\caption{%
		We report above for the partition data set the percentage of fixed triples and runtime after applying Proposition~\ref{lemma:persistency-triplet-cut}. From left to right, the sparsity of the graph decreases as written in the column title $(p_E, \vert V\vert, \bar{E})$.
	}
	\label{figure:partition-triangle-cut-sparse-joint-alphas}
\end{figure}

Instances of the partition data set are defined with respect to a 
graph $G = (V, E)$ with $\vert V \vert = 8n$ nodes, where $n \in \mathbb{N}$, and with respect to a graph partition $\Pi_G = \{U_1, U_2, U_3, U_4\}$
into four sets with $\vert U_1\vert = n$, $\vert U_2\vert = \vert U_3\vert = 2n$ and $\vert U_4 \vert = 3n$ elements, as depicted in Figure~\ref{figure:experiments}a, and with respect to three design parameters, $p_E \in \left[0, 1\right]$, $\alpha\in \left[0, 1\right]$ and $\beta\in \left[0, 1\right]$. 
We remark that we choose the same probability for all the edges in $E$, namely $p_E$. 

The graph $G$ is constructed randomly in the following way: Firstly, we insert every pair $pq\in \tbinom V2$ as edge to $G$ at random with probability $p_E$. Secondly, for every induced subgraph $G\left[U_i\right]$, we iterate over all pairs $pq\in \tbinom{U_i}{2}$ and check if $p$ and $q$ are path-connected in $G\left[U_i\right]$. If this is not the case, then we add the edge $pq$ as an edge to $G$.

We define costs for all edges $E$ and for all triples $T = \{pqr\in \tbinom V3 \mid \{pq, qr, pr\}\subseteq E\}$ that form a 3-clique in the graph.
Both the costs of edges and triples are drawn from two Gaussian distributions with means $-1 + \alpha$ and $1-\alpha$, depending on whether their elements belong to the same set or distinct sets in the partition $\Pi_G$, and with standard deviation $\sigma = \sigma_0 + \alpha(\sigma_1 - \sigma_0)$ with $\sigma_0 = 0.1$ and $\sigma_1 = 0.4$.
Costs of edges are multiplied by $1 - \beta$, and costs of triples by $\beta$. 
For small $\alpha$, the costs of pairs and triples whose elements belong to the same set of the partition $\Pi_G$ are well-separated from the costs of those pairs and triples whose elements belong to different sets of the partition $\Pi_G$. For $\alpha$ close to 1 this is not the case anymore. 
Therefore, the higher $\alpha$ is, the harder the problem becomes.
The higher $\beta$ is, the more important the costs of triples become.

The percentage of edges and triples fixed by applying all conditions jointly, as described in Section~\ref{section:mixing-conditions}, and with respect to $\alpha$ is shown in 
Figure~\ref{figure:partition-sparse-joint-ns-combined} (Rows~1--2). 
It can be seen from this figure that the percentage of fixed variables decreases with increasing $\alpha$. 
As $\alpha$ rises, the runtime increases but remains below one minute for all the instances
(Row~2). 
Varying $\beta$ does not affect the overall trend. 
However, the percentage of fixed variables decreases as soon as triple costs are introduced.
The percentage of variables fixed by applying Proposition~\ref{lemma:persistency-subset-separation}, Proposition~\ref{proposition:edge-cut-persistency} and Corollary~\ref{corollary:subset-join-all-pairs-at-once} separately is shown in 
Figures~\ref{figure:partition-independent-subsets-subset-join-edge-cut-sparse-joint-alphas}
and
the percentage of triples fixed by applying Proposition~\ref{lemma:persistency-triplet-cut} separately is shown in 
Figure~\ref{figure:partition-triangle-cut-sparse-joint-alphas}.
The other partial optimality conditions do not fix any variables of these instances.
While all cut conditions settle the value of some variables, this is not the case for the join statements.
In fact, only one join condition provides partial optimality in this case: Corollary~\ref{corollary:subset-join-all-pairs-at-once}.
Interestingly, this is the one statement that fixes the most variables for almost all instances of this data set.

For $\beta = 0.5$ and with respect to the instance size, the runtime and percentage of variables fixed by applying all conditions jointly are shown in 
Figure~\ref{figure:partition-sparse-joint-ns-combined} (Rows~3--4). 
It can be seen that, as the instance size increases, the number of fixed variables declines while the runtime increases. 
The runtime for complete graphs, $p_E = 1.0$, and $\alpha\in \{0.5, 0.65\}$ converges to a $\mathcal{O}(\vert V\vert^{5.92})$
trend for increasing instance size. As we sparsify the graphs, i.e., for decreasing $p_E$, the percentage of fixed variables increases and the runtime converges to a $\mathcal{O}(\vert V\vert^{4.77})$ trend for $p_E = 0.25$ and $\alpha\in \{0.5, 0.65\}$.
For complete graphs this is close to the expected worst-case complexity of $\mathcal{O}(\vert V\vert^6)$ as we need to solve cubically many max-flow problems if no variables can be fixed by our algorithm.
Therefore, we observe that we can handle larger sparse graphs, induced by larger instances. This suggests that this approach might be able to handle instances of size interesting for practical applications.
\subsection{Geometric Data Set}\label{section:experiments-geometric}
\input{figure-sparse-geometric-joint}
\tikzmath{\plotheight=4.0;}
\tikzmath{\plotwidth=4.6;}

\providecommand{\addvariableplot}{}	
\renewcommand{\addvariableplot}[1]{\addplot[
	only marks,
	mark=*,
	black,
	mark options = {draw=none}
	] table[
	col sep=comma,
	x expr=\thisrow{sigma}, 
	y expr=\thisrow{medianEliminatedVariables}*100,
	y error minus expr=100*(\thisrow{medianEliminatedVariables} - \thisrow{q25EliminatedVariables}),
	y error plus expr=100*(\thisrow{q75EliminatedVariables} - \thisrow{medianEliminatedVariables})
	] {#1};}

\providecommand{\addruntimeplotmilliseconds}{}	
\renewcommand{\addruntimeplotmilliseconds}[1]{\addplot[
	only marks,
	mark=*,
	black,
	mark options = {draw=none}
	] table[
	col sep=comma,
	x expr=\thisrow{sigma}, 
	y expr=\thisrow{medianDuration} / 1e6,
	y error minus expr=(\thisrow{medianDuration} - \thisrow{q25Duration}) / 1e6,
	y error plus expr=(\thisrow{q75Duration} - \thisrow{medianDuration})  / 1e6
	] {#1};}

\providecommand{\addruntimeplot}{}	
\renewcommand{\addruntimeplot}[1]{\addplot[
	only marks,
	mark=*,
	black,
	mark options = {draw=none}
	] table[
	col sep=comma,
	x expr=\thisrow{sigma}, 
	y expr=\thisrow{medianDuration} / 1e9,
	y error minus expr=(\thisrow{medianDuration} - \thisrow{q25Duration}) / 1e9,
	y error plus expr=(\thisrow{q75Duration} - \thisrow{medianDuration})  / 1e9
	] {#1};}

\begin{figure}
	\centering
	
	\begin{tikzpicture}[baseline]
		\begin{groupplot}[group style={group size= 4 by 6, horizontal sep=0.44cm, vertical sep=0.50cm}]
			
			
			\nextgroupplot[
			title={$(5, 99, 2079)$},
			title style={align=center},
			ylabel={Variables [\%]},
			xmin=0.0,
			xmax=0.3,
			ymin=0,
			ymax=100,
			width=\plotwidth cm,
			height=\plotheight cm,
			xticklabels=\empty
			]
			\addvariableplot{./data/sparse/geometric-independent-subsets/sigmas/n99_kn5_kf5_stats.csv}
			
			\nextgroupplot[
			title={$(10, 90, 2210)$},
			title style={align=center},
			xmin=0.0,
			xmax=0.3,
			ymin=0,
			ymax=100,
			yticklabel=\empty,
			width=\plotwidth cm,
			height=\plotheight cm,
			xticklabels=\empty
			]
			\addvariableplot{./data/sparse/geometric-independent-subsets/sigmas/n90_kn10_kf10_stats.csv}
			
			\nextgroupplot[
			title={$(15, 81, 2344)$},
			title style={align=center},
			xmin=0.0,
			xmax=0.3,
			ymin=0,
			ymax=100,
			yticklabel=\empty,
			width=\plotwidth cm,
			height=\plotheight cm,
			xticklabels=\empty
			]
			\addvariableplot{./data/sparse/geometric-independent-subsets/sigmas/n81_kn15_kf15_stats.csv}
			
			\nextgroupplot[
			title={$(\infty, 63, 1953)$},
			title style={align=center},
			xmin=0.0,
			xmax=0.3,
			ymin=0,
			ymax=100,
			yticklabel=\empty,
			width=\plotwidth cm,
			height=\plotheight cm,
			xticklabels=\empty
			]
			\addvariableplot{./data/sparse/geometric-independent-subsets/sigmas/n63_kn500_kf500_stats.csv}
			
			
			\nextgroupplot[
			ylabel={Runtime [ms]},
			xmin=0.0,
			xmax=0.3,
			ymin=0,
			ymax=20,
			width=\plotwidth cm,
			height=\plotheight cm,
			xticklabels=\empty
			]
			\addruntimeplotmilliseconds{./data/sparse/geometric-independent-subsets/sigmas/n99_kn5_kf5_stats.csv}	
			
			\nextgroupplot[
			xmin=0.0,
			xmax=0.3,
			ymin=0,
			ymax=20,
			yticklabel=\empty,
			width=\plotwidth cm,
			height=\plotheight cm,
			xticklabels=\empty
			]
			\addruntimeplotmilliseconds{./data/sparse/geometric-independent-subsets/sigmas/n90_kn10_kf10_stats.csv}	
			
			\nextgroupplot[
			xmin=0.0,
			xmax=0.3,
			ymin=0,
			ymax=20,
			yticklabel=\empty,
			width=\plotwidth cm,
			height=\plotheight cm,
			xticklabels=\empty
			]
			\addruntimeplotmilliseconds{./data/sparse/geometric-independent-subsets/sigmas/n81_kn15_kf15_stats.csv}	
			
			\nextgroupplot[
			xmin=0.0,
			xmax=0.3,
			ymin=0,
			ymax=20,
			yticklabel=\empty,
			width=\plotwidth cm,
			height=\plotheight cm,
			xticklabels=\empty
			]
			\addruntimeplotmilliseconds{./data/sparse/geometric-independent-subsets/sigmas/n63_kn500_kf500_stats.csv}	
			
			
			\nextgroupplot[
			title style={align=center},
			ylabel={Variables [\%]},
			xmin=0.0,
			xmax=0.3,
			ymin=0,
			ymax=100,
			width=\plotwidth cm,
			height=\plotheight cm,
			xticklabels=\empty,
			yshift=-0.25cm
			]
			\addvariableplot{./data/sparse/geometric-edge-cut/sigmas/n99_kn5_kf5_stats.csv}
			
			\nextgroupplot[
			title style={align=center},
			xmin=0.0,
			xmax=0.3,
			ymin=0,
			ymax=100,
			yticklabel=\empty,
			width=\plotwidth cm,
			height=\plotheight cm,
			xticklabels=\empty,
			yshift=-0.25cm
			]
			\addvariableplot{./data/sparse/geometric-edge-cut/sigmas/n90_kn10_kf10_stats.csv}
			
			\nextgroupplot[
			title style={align=center},
			xmin=0.0,
			xmax=0.3,
			ymin=0,
			ymax=100,
			yticklabel=\empty,
			width=\plotwidth cm,
			height=\plotheight cm,
			xticklabels=\empty,
			yshift=-0.25cm
			]
			\addvariableplot{./data/sparse/geometric-edge-cut/sigmas/n81_kn15_kf15_stats.csv}
			
			\nextgroupplot[
			title style={align=center},
			xmin=0.0,
			xmax=0.3,
			ymin=0,
			ymax=100,
			yticklabel=\empty,
			width=\plotwidth cm,
			height=\plotheight cm,
			xticklabels=\empty,
			yshift=-0.25cm
			]
			\addvariableplot{./data/sparse/geometric-edge-cut/sigmas/n63_kn500_kf500_stats.csv}
			
			\nextgroupplot[
			ylabel={Runtime [s]},
			xmin=0.0,
			xmax=0.3,
			ymin=0,
			ymax=1,
			width=\plotwidth cm,
			height=\plotheight cm,
			xticklabel=\empty
			]
			\addruntimeplot{./data/sparse/geometric-edge-cut/sigmas/n99_kn5_kf5_stats.csv}	
			
			\nextgroupplot[
			xmin=0.0,
			xmax=0.3,
			ymin=0,
			ymax=1,
			yticklabel=\empty,
			width=\plotwidth cm,
			height=\plotheight cm,
			xticklabel=\empty
			]
			\addruntimeplot{./data/sparse/geometric-edge-cut/sigmas/n90_kn10_kf10_stats.csv}	
			
			\nextgroupplot[
			xmin=0.0,
			xmax=0.3,
			ymin=0,
			ymax=1,
			yticklabel=\empty,
			width=\plotwidth cm,
			height=\plotheight cm,
			xticklabel=\empty
			]
			\addruntimeplot{./data/sparse/geometric-edge-cut/sigmas/n81_kn15_kf15_stats.csv}	
			
			\nextgroupplot[
			xmin=0.0,
			xmax=0.3,
			ymin=0,
			ymax=1,
			yticklabel=\empty,
			width=\plotwidth cm,
			height=\plotheight cm,
			xticklabel=\empty
			]
			\addruntimeplot{./data/sparse/geometric-edge-cut/sigmas/n63_kn500_kf500_stats.csv}

			\nextgroupplot[
			title style={align=center},
			ylabel={Variables [\%]},
			xmin=0.0,
			xmax=0.3,
			ymin=0,
			ymax=100,
			width=\plotwidth cm,
			height=\plotheight cm,
			xticklabels=\empty,
			yshift=-0.25cm
			]
			\addvariableplot{./data/sparse/geometric-subset-join/sigmas/n99_kn5_kf5_stats.csv}
			
			\nextgroupplot[
			title style={align=center},
			xmin=0.0,
			xmax=0.3,
			ymin=0,
			ymax=100,
			yticklabel=\empty,
			width=\plotwidth cm,
			height=\plotheight cm,
			xticklabels=\empty,
			yshift=-0.25cm
			]
			\addvariableplot{./data/sparse/geometric-subset-join/sigmas/n90_kn10_kf10_stats.csv}
			
			\nextgroupplot[
			title style={align=center},
			xmin=0.0,
			xmax=0.3,
			ymin=0,
			ymax=100,
			yticklabel=\empty,
			width=\plotwidth cm,
			height=\plotheight cm,
			xticklabels=\empty,
			yshift=-0.25cm
			]
			\addvariableplot{./data/sparse/geometric-subset-join/sigmas/n81_kn15_kf15_stats.csv}
			
			\nextgroupplot[
			title style={align=center},
			xmin=0.0,
			xmax=0.3,
			ymin=0,
			ymax=100,
			yticklabel=\empty,
			width=\plotwidth cm,
			height=\plotheight cm,
			xticklabels=\empty,
			yshift=-0.25cm
			]
			\addvariableplot{./data/sparse/geometric-subset-join/sigmas/n63_kn500_kf500_stats.csv}
			
			\nextgroupplot[
			ylabel={Runtime [s]},
			xlabel=$\sigma$,
			xmin=0.0,
			xmax=0.3,
			ymin=0,
			ymax=10,
			width=\plotwidth cm,
			height=\plotheight cm,
			]
			\addruntimeplot{./data/sparse/geometric-subset-join/sigmas/n99_kn5_kf5_stats.csv}	
			
			\nextgroupplot[
			xmin=0.0,
			xmax=0.3,
			ymin=0,
			ymax=10,
			yticklabel=\empty,
			width=\plotwidth cm,
			height=\plotheight cm,
			xlabel=$\sigma$,
			]
			\addruntimeplot{./data/sparse/geometric-subset-join/sigmas/n90_kn10_kf10_stats.csv}	
			
			\nextgroupplot[
			xmin=0.0,
			xmax=0.3,
			ymin=0,
			ymax=10,
			yticklabel=\empty,
			width=\plotwidth cm,
			height=\plotheight cm,
			xlabel=$\sigma$,
			]
			\addruntimeplot{./data/sparse/geometric-subset-join/sigmas/n81_kn15_kf15_stats.csv}	
			
			\nextgroupplot[
			xmin=0.0,
			xmax=0.3,
			ymin=0,
			ymax=10,
			yticklabel=\empty,
			width=\plotwidth cm,
			height=\plotheight cm,
			xlabel=$\sigma$,
			]
			\addruntimeplot{./data/sparse/geometric-subset-join/sigmas/n63_kn500_kf500_stats.csv}	
			
		\end{groupplot}
	\end{tikzpicture}
	\\[-2ex]
	\caption{%
		We report above for the geometric data set the percentage of fixed variables and runtime after applying Proposition~\ref{lemma:persistency-subset-separation} (Rows 1--2),
		Proposition~\ref{proposition:edge-cut-persistency} (Rows 3--4) and 
		Corollary~\ref{corollary:subset-join-all-pairs-at-once} (Rows 5--6), separately.
		From left to right, the sparsity of the graph decreases as written in the column title $(k, \vert V\vert, \bar{E})$.}
	\label{figure:geometric-independent-subsets-subset-join-edge-cut-sparse-joint-alphas-combined}
\end{figure}
\tikzmath{\plotheight=4.0;}
\tikzmath{\plotwidth=4.6;}

\providecommand{\addtriangleplot}{}	
\renewcommand{\addtriangleplot}[1]{\addplot[
	only marks,
	mark=*,
	black,
	mark options = {draw=none}
	] table[
	col sep=comma,
	x expr=\thisrow{sigma}, 
	y expr=100*\thisrow{medianEliminatedTriangles},
	y error minus expr=100*\thisrow{medianEliminatedTriangles} - 100*\thisrow{q25EliminatedTriangles},
	y error plus expr=100*\thisrow{q75EliminatedTriangles} - 100*\thisrow{medianEliminatedTriangles} 
	] {#1};}

\providecommand{\addruntimeplot}{}	
\renewcommand{\addruntimeplot}[1]{\addplot[
	only marks,
	mark=*,
	black,
	mark options = {draw=none}
	] table[
	col sep=comma,
	x expr=\thisrow{sigma}, 
	y expr=\thisrow{medianDuration} / 1e9,
	y error minus expr=(\thisrow{medianDuration} - \thisrow{q25Duration}) / 1e9,
	y error plus expr=(\thisrow{q75Duration} - \thisrow{medianDuration})  / 1e9
	] {#1};}

\begin{figure}[!t]
	\centering
	
	\begin{tikzpicture}[baseline]
		\begin{groupplot}[group style={group size= 4 by 2, horizontal sep=0.44cm, vertical sep=0.50cm}]
			
			
			\nextgroupplot[
			title={$(5, 99, 2079)$},
			title style={align=center},
			ylabel={Triples [\%]},
			xmin=0.0,
			xmax=0.3,
			ymin=0,
			ymax=100,
			width=\plotwidth cm,
			height=\plotheight cm,
			xticklabels=\empty
			]
			\addtriangleplot{./data/sparse/geometric-triangle-cut/sigmas/n99_kn5_kf5_stats.csv}
			
			\nextgroupplot[
			title={$(10, 90, 2210)$},
			title style={align=center},
			xmin=0.0,
			xmax=0.3,
			ymin=0,
			ymax=100,
			yticklabel=\empty,
			width=\plotwidth cm,
			height=\plotheight cm,
			xticklabels=\empty
			]
			\addtriangleplot{./data/sparse/geometric-triangle-cut/sigmas/n90_kn10_kf10_stats.csv}
			
			\nextgroupplot[
			title={$(15, 81, 2344)$},
			title style={align=center},
			xmin=0.0,
			xmax=0.3,
			ymin=0,
			ymax=100,
			yticklabel=\empty,
			width=\plotwidth cm,
			height=\plotheight cm,
			xticklabels=\empty
			]
			\addtriangleplot{./data/sparse/geometric-triangle-cut/sigmas/n81_kn15_kf15_stats.csv}
			
			\nextgroupplot[
			title={$(\infty, 63, 1953)$},
			title style={align=center},
			xmin=0.0,
			xmax=0.3,
			ymin=0,
			ymax=100,
			yticklabel=\empty,
			width=\plotwidth cm,
			height=\plotheight cm,
			xticklabels=\empty
			]
			\addtriangleplot{./data/sparse/geometric-triangle-cut/sigmas/n63_kn500_kf500_stats.csv}
			
			
			\nextgroupplot[
			ylabel={Runtime [s]},
			xlabel=$\sigma$,
			xmin=0.0,
			xmax=0.3,
			ymin=0,
			ymax=30,
			width=\plotwidth cm,
			height=\plotheight cm
			]
			\addruntimeplot{./data/sparse/geometric-triangle-cut/sigmas/n99_kn5_kf5_stats.csv}	
			
			\nextgroupplot[
			xmin=0.0,
			xmax=0.3,
			ymin=0,
			ymax=30,
			yticklabel=\empty,
			width=\plotwidth cm,
			height=\plotheight cm,
			xlabel=$\sigma$
			]
			\addruntimeplot{./data/sparse/geometric-triangle-cut/sigmas/n90_kn10_kf10_stats.csv}	
			
			\nextgroupplot[
			xmin=0.0,
			xmax=0.3,
			ymin=0,
			ymax=30,
			yticklabel=\empty,
			width=\plotwidth cm,
			height=\plotheight cm,
			xlabel=$\sigma$
			]
			\addruntimeplot{./data/sparse/geometric-triangle-cut/sigmas/n81_kn15_kf15_stats.csv}	
			
			\nextgroupplot[
			xmin=0.0,
			xmax=0.3,
			ymin=0,
			ymax=30,
			yticklabel=\empty,
			width=\plotwidth cm,
			height=\plotheight cm,
			xlabel=$\sigma$
			]
			\addruntimeplot{./data/sparse/geometric-triangle-cut/sigmas/n63_kn500_kf500_stats.csv}	
			
		\end{groupplot}
        \end{tikzpicture}
        \\[-2ex]
	\caption{%
		We report above for the geometric data set the percentage of fixed triples and runtime after applying Proposition~\ref{lemma:persistency-triplet-cut}. From left to right, the sparsity of the graph decreases as written in the column title $(k, \vert V\vert, \bar{E})$.}
	\label{figure:geometric-triangle-cut-sparse-joint-alphas}
\end{figure}

Next, we consider a data set of instances that arise from the geometric problem of finding equilateral triangles in a noisy point cloud; see Figure~\ref{figure:experiments}b.
For this, we fix three equilateral triangles in the plane. 
For each vertex $\vec{a}$ of a triangle, we draw a number of points from a Gaussian distribution with mean $\vec{a}$ and covariance matrix $\sigma^2\mathbbm{1}$. 
This defines the set of nodes, $V$.
We build a graph $G = (V, E)$ from these points in the following way. 
First we add all pairs $pq\in \tbinom V2$ as edges if $p$ and $q$ were drawn from vertices of the same triangle. 
Second, with respect to a parameter $k\in \mathbb{N}$, we compute for each point $\vec{a}_p$ its $k$ nearest neighbors, $V_p^n$, and its $k$ farthest neighbors, $V_p^f$. Then we add for every $p\in V$ the pairs $\{pq\mid q\in V_p^n\cup V_p^f\}$ as edges to the graph $G$.
Adding both the $k$ nearest and the $k$ farthest neighbors is motivated by the following intuition. Three points that are mutually close likely are drawn from the same triangle. For three points that are mutually far apart we can make the decision whether these points belong to the same triangle or to different triangles based on the deviation of interior angles from $\pi / 3$. For three points that are neither mutually close or far apart, this decision is ambiguous. Therefore, we prefer to include both edges between far apart and close points. 
We let $T = \{pqr\in \tbinom V3\mid pq\in E, qr\in E, pr\in E \}$, i.e., we introduce costs for all triples that form a 3-clique.

For any three points $\vec{a}_p, \vec{a}_q, \vec{a}_r$ such that $pqr\in T$, 
let $\varphi_p$, $\varphi_q, \varphi_r$ be the interior angles of the triangle spanned by these points, and let $d^{max}_{pqr}$ and $d^{min}_{pqr}$ be the maximum and minimum length of edges in this triangle.
If the three points are mutually close, $d^{max}_{pqr} \leq 4\sigma$, we reward solutions in which these belong to the same set by letting $c_{pqr} = -1 + \frac{d^{max}_{pqr}}{4\sigma}$. 
If only two points are close, $d^{max}_{pqr} > 4\sigma$ and $d^{min}_{pqr} \leq 4\sigma$, we let $c_{pqr} = 0$. 
If the three points are mutually far apart, $d^{min}_{pqr} > 4\sigma$, we calculate the sum of the deviations of the inner angles from $\frac \pi 3, $ $\delta_{pqr} = \sum_i\vert \varphi_i - \frac{\pi}{3}\vert$. 
If $\delta_{pqr} \leq \frac \pi 6$, we let $c_{pqr} = -1 + \frac{6\delta_{pqr}}{\pi}$. Otherwise, $c_{pqr} = \frac{6}{7}\frac{\delta_{pqr} - \frac{\pi}{6}}{\pi}$. 

The percentage of variables fixed by applying all conditions jointly, as described in Section~\ref{section:mixing-conditions}, 
and with respect to $\sigma$ is reported in  Figure~\ref{figure:geometric-sparse-joint-sigmas-combined} (Rows 1--2).
Here, the hardness of the instances is embodied by $\sigma$. 
Depending on the sparsity of the instances, the number of nodes ranges from $\vert V\vert = 54$ (for $k = \infty$) to $\vert V\vert = 90$ (for $k = 5$).
As $\sigma$ increases, the percentage of fixed variables decreases. 
The runtime increases, as can be seen from Figure~\ref{figure:geometric-sparse-joint-sigmas-combined} (Row~2), and stays below one minute for all these instances. We also observe that for sparser instances (lower $k$) we can fix less variables at the same noise level $\sigma$.
The percentage of variables fixed by applying Proposition~\ref{lemma:persistency-subset-separation}, Proposition~\ref{proposition:edge-cut-persistency} and Corollary~\ref{corollary:subset-join-all-pairs-at-once} separately is shown in 
Figures~\ref{figure:geometric-independent-subsets-subset-join-edge-cut-sparse-joint-alphas-combined}
and
the percentage of triples fixed by applying Proposition~\ref{lemma:persistency-triplet-cut} separately is shown in 
Figure~\ref{figure:geometric-triangle-cut-sparse-joint-alphas}.
Also here, all the cut conditions are effective whereas the only useful join condition is Corollary~\ref{corollary:subset-join-all-pairs-at-once}.
Moreover, Corollary~\ref{corollary:subset-join-all-pairs-at-once} is the most effective 
for complete graphs. In contrast, Propositions~\ref{lemma:persistency-subset-separation} and~\ref{proposition:edge-cut-persistency} are more effective for sparse instances.

The runtime and percentage of variables fixed by applying all conditions jointly and with respect to the instance size are shown in 
Figure~\ref{figure:geometric-sparse-joint-sigmas-combined} (Rows 3--4).
Similar to the partition data set, we see that, as the instance size increases, the number of fixed variables declines while the runtime increases. 
For some $\sigma$, our conditions can fix more variables in instances defined on complete graphs than in instances defined on sparse graphs.
For increasing instance size, the runtime for $\sigma = 0.1$ and $k = 5$ converges to a $\mathcal{O}(\vert V\vert^{5.77})$ trend and for $k=10$ to a $\mathcal{O}(\vert V\vert^{5.36})$ trend. 
Again, this is close to the expected worst-case complexity of $\mathcal{O}(\vert V\vert^6)$.
Moreover, we are able to handle larger sparse graphs, induced by larger instances.

\section{Connection to the Cubic Multicut Problem}\label{sec: multicut}
In this section, we turn to a natural question that arises in the context of correlation clustering and multicuts.
When the objective function is linear, it is well-known that the two problems are equivalent, as mentioned also in the introduction. 
In particular, the feasible vectors of the two problems are in a 1-to-1 correspondence given by the affine map $x \mapsto 1 - x$.
This implies that the optimal value of the (linear) correlation clustering problem is equal to a constant minus the optimal value of the (linear) multicut problem. 
It is therefore reasonable to ask ourselves whether the same behaviour occurs also in the nonlinear setting. 
One important thing to realize is that, while the cubic term in \eqref{eq: pb} tells us exactly what we want, i.e., it is equal to one if and only if all the three corresponding points are in the same node set, this expressiveness property does not transfer to the higher-order multicut problem.
In fact, if we simply replace the variables with their affine transformation $x \mapsto 1 - x$, the cubic term becomes $(1 - x_{pq})(1 - x_{pr})(1 - x_{qr})$ for $pqr \in T$. 
We observe that this product equals one if and only if all the correlation clustering variables $x_{pq}$, $x_{pr}$, $x_{qr}$ have value zero, meaning that the three points end up in three different node subsets. 
However, we would desire the cubic term in the multicut problem to evaluate $1$ also when two points out of three are together but the third one is in another node subset. 
This suggests that the higher-order multicut problem cannot be expressed naturally without the introduction of additional variables, contrary to the higher-order correlation clustering problem. 
Next, we define the variables of the cubic multicut problem, and show the relationship between the optimal values of this problem and of the cubic correlation clustering. 
\begin{gather}
\forall pq \in E, \qquad z_{pq} := 1 - x_{pq} = \begin{cases}
1, \qquad & \text{if } p, q \text{ are in different subsets,} \\
0, \qquad & \text{if } p, q \text{ otherwise}, 
\end{cases} \\
\forall pqr \in T, \qquad y_{pqr} := 1 - x_{pq}x_{pr}x_{qr} = \begin{cases}
1, \qquad & \text{if } p, q, r \text{ are not all together,} \\
0, \qquad & \text{if } p, q, r \text{ otherwise}.
\end{cases}
\end{gather}
The cubic multicut problem w.r.t.~to a graph $G = (V, E)$, triples $T$, and costs $c \in \mathbb{R}^{I(G)}$ is: 
\begin{align}\label{eq: homc}
	\smallmath{\min_{\substack{z: E \to \{0,1\} \\ y: T \to \{0,1\}}} \quad} & \smallmath{\sum_{pqr \in T} c_{pqr} y_{pqr} + \sum_{pq \in E} c_{pq} z_{pq} + c_{\emptyset}} \\
	\smallmath{\text{s.t.} \quad \ } & \smallmath{\forall (V_C, E_C)\in \mathrm{cycles}(G)\, \forall e'\in E_C\colon \quad z_{e'} \leq \sum_{e\in E_C\setminus \{e'\}}z_e }\\
	\label{eq:cubic-multicut-edges-not-cut-triplet-not-cut-constraint}
	&\smallmath{\forall pqr\in T\colon \quad y_{pqr} \leq z_{pq} + z_{qr} + z_{pr}}\\
	\label{eq:cubic-multicut-one-edge-cut-triplet-cut-constraint}
	&\smallmath{\forall pqr\in T\colon \quad y_{pqr} \geq z_{pq} \text{\, and \,} y_{pqr} \geq z_{qr} \text{\, and \,} y_{pqr} \geq z_{pr}}
\end{align}
We denote the objective value of \eqref{eq: pb} by $\text{obj}_{\text{\ccp}}$ and the value of \eqref{eq: homc} by $\text{obj}_\text{CMCP}$:
\begin{align}
\text{obj}_{\text{\ccp}} & = \sum_{pqr \in T} c_{pqr} x_{pq}x_{pr}x_{qr} + \sum_{pq \in E} c_{pq} x_{pq} + c_{\emptyset} \\ 
& = \sum_{pqr \in T} c_{pqr} - \sum_{pqr \in T} c_{pqr} (1 - x_{pq}x_{pr}x_{qr}) + \sum_{pq \in E} c_{pq} - \sum_{pq \in E} c_{pq} (1 - x_{pq}) - c_{\emptyset} + 2c_{\emptyset} \\
& = \sum_{pqr \in T} c_{pqr} - \sum_{pqr \in T} c_{pqr} y_{pqr} + \sum_{pq \in E} c_{pq} - \sum_{pq \in E} c_{pq} z_{pq} - c_{\emptyset} + 2c_{\emptyset} = - \text{obj}_\text{CMCP} + C, 
\end{align}
where $C$ is a constant (specifically, $C = \sum_{pqr \in T} c_{pqr} + \sum_{pq \in E} c_{pq} + 2c_{\emptyset}$).

We conclude that the two problems are equivalent, similarly to the linear case. 
However, the two problems have a different number of variables in the non-linear setting.
Hence, all partial optimality results obtained for one problem can be transferred to the other problem once we flip the sign of the costs and the binary values of the variables.
\section{Conclusion}
We establish partial optimality conditions for the cubic correlation clustering problem for general graphs. 
In particular, we generalize the conditions and the implementation established by~\cite{stein-2023} for the case of complete graphs to that of general graphs. 
We empirically evaluate these conditions on two data sets and make the following observations:
Firstly, we observe for the specific instances investigated in this article that all cut conditions are effective and the only join condition that is effective is Proposition~\ref{proposition:subset-join-proposition} (in its simplified form of Corollary~\ref{corollary:subset-join-all-pairs-at-once}). 
In the interest of objectivity, we discuss all conditions generalized from prior work unconditionally and report also the ineffectiveness of some conditions. 
Secondly, we observe that our implementation can handle larger instances, with respect to the number of elements, as we sparsify the instances.

One question we do not answer here is in which order to apply certified joins. 
The fraction of fixed variables depends on this order, and looking into heuristics for choosing an order appears promising.
Toward future work, it would be interesting to see whether partial optimality conditions are completely or partially transferable to higher-order correlation clustering or to the case where triple costs additionally depend on node labels.

\acks{%
	David Stein acknowledges partial support by the Federal Ministry of Education and Research of Germany through grant 16LW0079K.
	Bjoern Andres acknowledges partial support by the Federal Ministry of Education and Research of Germany through grant 16KIS2332K.}

\newpage

\appendix
\section{Reductions to Minimum Cut Problems}\label{sec:tech-min-cut} 
Here, we discuss how, for a given pair or triple, we reduce the search for subsets $U, U' \subseteq V$ that satisfy \eqref{eq:assumption-edge-cut-inequality}--\eqref{eq:triangle-edge-join-2} maximally to the min $st$-cut problem.
In any of these cases, we have $G = (V, E)$ be a graph and $c\in \mathbb{R}^{I(G)}$ such that $c_{pqr} \geq 0$ for all $pqr\in T$, and $c_{pq}\geq 0$ for all $pq\in E$. 
Moreover, in the cases of~\eqref{eq:assumption-edge-cut-inequality} and \eqref{eq:edge-join-inequality}~we have an edge $ij \in E$ such that $i\in U$ and $j\notin U$, while in the cases of~\eqref{eq:triplet-cut-condition}, \eqref{eq:assumption-triplet-join}, \eqref{eq:triangle-edge-join-1} and \eqref{eq:triangle-edge-join-2} we have a triple $ijk \in T$ such that $ij, ik\in \delta(U)$. Note that in the latter case we can assume that $i\in U$ and $j, k\notin U$ without loss of generality. We handle both cases jointly by  denoting the set of vertices not in $U$ by $V_0 \subseteq V$. The search for subsets that satisfy \eqref{eq:assumption-edge-cut-inequality}--\eqref{eq:triangle-edge-join-2} maximally then takes the form:
\begin{equation}
	\min_{\substack{U \subseteq V \colon \\ i \in U,\\ j \notin U, \forall j \in V_0}}\quad \sum_{pqr\in T_{\delta(U)}}c_{pqr} + \sum_{pq\in \delta(U)}c_{pq} 
	\enspace .
	\label{eq:verification-problem-unified}
\end{equation}
To begin with, we move costs of triples to costs of pairs:
\begin{proposition}\label{prop: reduction costs}
	\label{eq:cubic-st-cut-reduction-min-st-cut}
	Let $G = (V, E)$ be a graph, $c\in \mathbb{R}^{I(G)}$ and $U \subseteq V$. Then
	$\sum_{pqr\in T_{\delta(U)}} c_{pqr} = \frac{1}{2}\sum_{pq\in \delta(U)}\sum_{\substack{r\in V \setminus \{p, q\} : \\ pqr \in T}}c_{pqr}$.
\end{proposition} 
\begin{proof}
	Let $U \subseteq V$.
	Observe that
	\begin{align} \label{eq:identity-sum-all-TdeltaR-triplets}
		&\sum_{pqr\in T_{\delta(U)}}c_{pqr} 
		= 
		\sum_{pq\in \tbinom U2}\sum_{\substack{r\in V \setminus U : \\ pqr \in T}}c_{pqr} 
		+ \sum_{pq\in \tbinom{V \setminus U}{2}}\sum_{\substack{r\in U : \\ pqr \in T}} c_{pqr} \\
		= &\frac 12 \sum_{p\in U}\sum_{q\in U \setminus \{p\}}\sum_{\substack{r\in V \setminus U : \\ pqr \in T}}c_{pqr} + \frac 12 \sum_{p\in V \setminus U}\sum_{q\in V \setminus \left(U \cup \{p\}\right)}\sum_{\substack{r\in U : \\ pqr \in T}}c_{pqr} \\
		= &\frac 12 \sum_{p\in U}\sum_{q\in V \setminus U}
		\Bigl(\sum_{\substack{r\in U \setminus \{p\} : \\ pqr \in T}}c_{pqr} + \sum_{\substack{r\in V \setminus \left(U \cup \{q\}\right) : \\ pqr \in T}}c_{pqr}\Bigr) 
		= \frac 12 \sum_{pq\in \delta(U)}\sum_{\substack{r\in V \setminus \{p, q\} : \\ pqr \in T}}c_{pqr}. \qedhere
	\end{align}
\end{proof}
Consequently, \eqref{eq:verification-problem-unified} is equivalent to 
	$\displaystyle\min_{\substack{U \subseteq V \colon \\ i \in U,\\ j \notin U, \forall j \in V_0}}\quad \sum_{pq\in \delta(U)}c'_{pq}$,
where $c'_{pq} = c_{pq} + \frac{1}{2}\sum_{\substack{r\in V \setminus \{p, q\} \\ pqr \in T}}c_{pqr}$ for all $pq \in E$.
Next, we reduce the above minimization problem to a \emph{quadratic unconstrained binary optimization problem}, by applying the following proposition:
\begin{proposition}
	\label{lemma:qpbo-translation}
	Let $G = (V, E)$ be a graph and $c\in \mathbb{R}^E$. 
	Moreover, let $i\in V$ and $V_0 \subseteq V\setminus \{i\}$. 
	Furthermore, let $V' = V \setminus \left(V_0 \cup \{i\}\right)$ and $c' \colon (E \cap \tbinom{V'}{2}) \cup V	' \cup \{\emptyset\} \to \mathbb{R}$ such that
	\begin{gather}
		c'_p = \sum_{q\in V \setminus \{p\} :  pq \in E}c_{pq} - 2c_{pi} \: \forall p \in V' \colon pi \in E, \qquad
		c'_p = \sum_{q\in V \setminus \{p\} :  pq \in E}c_{pq} \: \forall p \in V' \colon pi \notin E
		\\
		c'_{pq} = -2c_{pq} \: \forall pq \in E \cap \tbinom{V'}{2},
		\qquad
		c'_\emptyset = \sum_{q\in V \setminus \{i\} : qi \in E}c_{qi}.
	\end{gather}
	Then:
	\begin{equation}
		\min_{\substack{U \subseteq V \colon \\ i \in U, \\ j \not\in U, \forall j \in V_0}} \sum_{pq\in \delta(U)} c_{pq} 
		= \min_{y\in \{0, 1\}^{V'}} 
		\sum_{pq\in E \cap \tbinom{V'}{2}} c'_{pq}y_p y_q + \sum_{p\in V'}c'_p y_p + c'_\emptyset
		\enspace .
		\label{eq:node-encoding-equality-2} 
	\end{equation}
\end{proposition}
\begin{proof}
	Let $U \subseteq V$ such that $i \in U$ and $\forall j \in V_0 \colon j\not \in U$. 
	We define $y\in \{0, 1\}^V$ such that $y = \mathbbm{1}_U$. 
	Then, we have $y_i = 1$ and $\forall j \in V_0 \colon y_j = 0$. 
	Moreover, it follows 
\small{	\begin{align}
		&\smashoperator[r]{\sum_{pq\in \delta(U)}} c_{pq}
		= 
		\hspace{-0.5ex}\sum_{\mathclap{pq\in E}}c_{pq}\left(y_p (1-y_q) + y_q (1-y_p)\right) 
		= 
		\hspace{-0.5ex}\sum_{\mathclap{pq\in E}}c_{pq}\left(y_p + y_q - 2y_p y_q\right) 
		= -2
		\sum_{\mathclap{pq\in E}}c_{pq}y_p y_q + 
		\sum_{\mathclap{\substack{p,q \in V : \\ p \neq q, \\ pq \in E}}} c_{pq}y_p \\
		& = -2 \sum_{pq\in E \cap \tbinom{V'}{2}}c_{pq}y_p y_q - 2\sum_{\substack{p\in V' : \\pi\in E}}c_{pi} y_p + \sum_{p\in V'}\sum_{\substack{q\in V \setminus \{p\} : \\ pq \in E}}c_{pq}y_p + \sum_{\substack{q\in V \setminus \{i\} : \\ qi \in E}} c_{qi} \\
		& =-2\sum_{pq \in E \cap \tbinom{V'}{2}}c_{pq}y_p y_q + \sum_{\substack{p\in V' : \\pi \in E}}
		\Bigl(-2c_{pi} + \sum_{\substack{q\in V \setminus \{p\} : \\ pq\in E}}c_{pq}\Bigr)y_p 
		+
		\sum_{\substack{p\in V' : \\pi \notin E}}\sum_{\substack{q\in V \setminus \{p\} : \\ pq\in E}}c_{pq}y_p 
		+
		\sum_{\substack{q\in V \setminus \{i\} : \\ qi\in E}}c_{qi} \\
		&= \sum_{pq\in \tbinom{V'}{2}}c'_{pq}y_p y_q + \sum_{p\in V'} c'_p y_p + c'_\emptyset.\qedhere
	\end{align}%
	}\normalsize%
\end{proof}%
For the instances of \eqref{eq:node-encoding-equality-2} that arise from deciding  \eqref{eq:assumption-edge-cut-inequality}--\eqref{eq:triangle-edge-join-2}, we have $\forall pq\in E \cap \tbinom{V'}{2} \colon c'_{pq}\leq 0$.
Thus, the right-hand side of \eqref{eq:node-encoding-equality-2} is \emph{submodular} and can be minimized in strongly polynomial time \citep{boros-2008, kolmogorov-2004}.
	For completeness we describe the reduction of quadratic unconstrained binary optimization to an instance of min-$st$-cut in detail: Lemma~\ref{lemma:qubo-to-max-flow-analogue} reformulates the objective function of quadratic unconstrained binary optimization. Proposition~\ref{prop:last-prop} shows a reduction of this reformulation to min-$st$-cut.
	\begin{lemma}
		\label{lemma:qubo-to-max-flow-analogue}
		Let $G = (V, E)$ be a graph and $c\in \mathbb{R}^{V\cup E}$. We define $c'\colon \mathbb{R}^{V\cup E}$ as 
		\begin{gather}
			c'_p = c_p + \frac{1}{2}\sum_{\substack{q\in V\setminus \{p\} : \\ pq\in E}}c_{pq} \:\:\:\: \forall p\in V, \qquad
			c'_{pq} = - \frac{1}{2}c_{pq} \:\:\:\: \forall pq\in E.
		\end{gather}
		Then, for any $y\in \{0, 1\}^V$ we have that
		\begin{equation}
			\sum_{pq\in E}c_{pq} y_py_q + \sum_{p\in V} c_p y_p =  \sum_{p\in V}\sum_{\substack{q\in V\setminus \{p\} : \\ pq\in E}} c'_{pq}y_p (1-y_q) + \sum_{\substack{p\in V : \\ c'_p > 0}} c'_p y_p -\sum_{\substack{p\in V : \\ c'_p< 0}} c'_p (1-y_p) + \sum_{\substack{p\in V : \\ c'_p < 0}}c'_p\enspace.
		\end{equation}
	\end{lemma}
	\begin{proof}
		Let $y\in \{0, 1\}^V$. 
		We have that
		\small{\begin{align}
			&\sum_{pq\in E}c_{pq}y_p y_q + \sum_{p\in V}c_py_p = \frac{1}{2}\sum_{p\in V}\sum_{\substack{q\in V \setminus \{p\} : \\ pq\in E}}c_{pq}y_p y_q + \sum_{p\in V}c_py_p \\
			 = &-\frac{1}{2}\sum_{p\in V} \sum_{\substack{q\in V \setminus \{p\} : \\ pq\in E}} c_{pq}y_p (1-y_q) 
			+ 
			\frac{1}{2}\sum_{p\in V} \sum_{\substack{q\in V \setminus \{p\} : \\ pq\in E}} c_{pq}y_p  
			+ 
			\sum_{p\in V} c_{p}y_p \\
			= &-\frac{1}{2}\sum_{p\in V} 
			\smashoperator[r]{\sum_{\substack{q\in V \setminus \{p\} : \\ pq\in E}}} c_{pq}y_p (1-y_q) 
			+ 
			\sum_{p\in V} 
			\Bigl(c_p + \frac{1}{2}\sum_{\substack{q\in V \setminus \{p\} : \\ pq\in E}} c_{pq}\Bigr)y_p 
			= \sum_{p\in V} 
			\sum_{\substack{q\in V \setminus \{p\} : \\ pq\in E}} c'_{pq}y_p (1-y_q) + \sum_{p\in V} c'_p y_p \\
			= &\sum_{p\in V} \sum_{\substack{q\in V \setminus \{p\} : \\ pq\in E}} c'_{pq}y_p (1-y_q) + \sum_{\substack{p\in V : \\ c'_p > 0}} c'_p y_p -\sum_{\substack{p\in V : \\ c'_p< 0}} c'_p (1-y_p) + \sum_{\substack{p\in V : \\ c'_p < 0}}c'_p. \qedhere
		\end{align}}\normalsize%
	\end{proof}
	\begin{proposition}\label{prop:last-prop}
		Let $G = (V, E)$ be a graph and $c\in \mathbb{R}^{V \cup E}$. 
		We define $\phi_c \colon \{0, 1\}^V \to \mathbb{R}$ such that for all $y\in \{0, 1\}^V$ it holds that
			\begin{align}
				\phi_c(y) = &\sum_{p\in V} \sum_{\substack{q\in V \setminus \{p\} : \\ pq \in E}} c_{pq}y_p (1-y_q) 
				+ \sum_{\substack{p\in V : \\ c_p > 0}} c_p y_p - \sum_{\substack{p\in V : \\ c_p< 0}} c_p(1-y_p).
			\end{align}%
		Furthermore, we define $V' = V \cup \{s, t\}$, $E'\subset V'\times V'$ such that $
			(s, p)\in E' \Leftrightarrow c_p < 0,\: \forall p \in V; \: 
			(p, t)\in E' \Leftrightarrow c_p > 0,\: \forall p \in V; \:$ and $
			(p, q)\in E' \land (q, p)\in E', \: \forall pq\in E
$.		
		We also define $c' \in \mathbb{R}^{E'}$ as 
		$c'_{(s, p)} = -c_p, \: \forall (s, p) \in E'; \:
		c'_{(p, t)} = c_p, \: \forall (p, t)\in E'; $ and $
		c'_{(p, q)} = c'_{(q, p)}= c_{pq}, \: \forall pq\in E .$
		Moreover, we define the function $\varphi_{c'}: \{0, 1\}^{V'} \to \mathbb{R}$ such that for all $y\in \{0, 1\}^{V'}$ it holds that 
		$\varphi_{c'}(y) = \sum_{(p,q)\in E'}c'_{(p, q)}y_p(1- y_q).$
		Then we have that
		\begin{equation}
			\min_{x\in \{0, 1\}^V}\phi_c(x) = \min_{\substack{y\in \{0, 1\}^{V'} : \\ y_s = 1  , \,y_t = 0}} \varphi_{c'}(y).
		\end{equation}%
	\end{proposition}
	\begin{proof}
		First, the map $\chi\colon \{0, 1\}^V \to \{y \in \{0, 1\}^{V'}\mid y_s = 1 \,\land \, y_t = 0\}$ such that $\chi(y)_s = 1$, $\chi(y)_t = 0$ and $\chi(y)_p = y_p$ for any $p\in V$ is bijective. 
		Second, for any $y\in \{0, 1\}^V$ it holds that
\small{		\begin{align}
			&\varphi_{c'}(\chi(y))
			= \sum_{\mathclap{(p, q)\in E'}}c'_{pq}\chi(y)_p (1- \chi(y)_q) 
			= 
			\sum_{p\in V}\hspace{-0.25ex}
			\smashoperator[r]{\sum_{\substack{q\in V \setminus \{p\} : \\ pq\in E}}}c'_{(p, q)}y_p(1- y_q) + \sum_{\mathclap{\substack{p\in V : \\ c_p > 0}}}c'_{(p, t)}y_p 
			+ 
			\sum_{\mathclap{\substack{p\in V : \\ c_p < 0}}}c'_{(s, p)}(1- y_p) \\
			= &\sum_{p\in V}\sum_{\substack{q\in V\setminus \{p\} : \\ pq\in E}}c_{pq}y_p(1- y_q) +\sum_{\substack{p\in V : \\ c_p > 0}}c_py_p - \sum_{\substack{p\in V : \\ c_p < 0}}c_p(1- y_p) = \phi_c(y)\qedhere
		\end{align}}%
\normalsize
	\end{proof}
	
	If $c_{pq} \leq 0$, $pq\in E$, and therefore $c'_{pq} \geq 0$, $\forall pq\in E$, in Lemma~\ref{lemma:qubo-to-max-flow-analogue} the resulting instance can be solved efficiently.	
\section{Reproduction of Experiments from \cite{stein-2023}}\label{sec:reproduction-experiments-conference-article} 
We run experiments from~\cite{stein-2023} on the same hardware as used for our experiments in this manuscript. Moreover, we impose a runtime limit of $T = 60\mathrm{s}$.

\subsection{Partition Data Set}
In Figure~\ref{figure:partition-complete-joint} and Figure~\ref{figure:partition-goemetric-complete-individual-combined} (Row~1) we show the same experiments as described in Section~\ref{section:experiments-partition} for complete graphs, but with the code from~\cite{stein-2023}. We observe a slightly better runtime trend of $\mathcal{O}(\vert V\vert^{5.79})$ than for our new implementation applied to complete graphs.

\subsection{Geometric Data Set}
In Figure~\ref{figure:geometric-complete-joint} and Figure~\ref{figure:partition-goemetric-complete-individual-combined} (Row~2) we show the same experiments as described in Section~\ref{section:experiments-geometric} for complete graphs, but with the code from~\cite{stein-2023}. We observe a worse runtime behavior of $\mathcal{O}(\vert V\vert^{7.45})$, in contrast to the behavior of $\mathcal{O}(\vert V\vert^{5.8})$ reported in \cite{stein-2023}. We conjecture that the reason for this difference lies in the different runtime limits imposed, i.e., in this manuscript we impose a runtime limit of 60s and in~\cite{stein-2023} it is 1000s. Therefore, we assume that the runtimes for small $\vert V\vert$ affect the runtime trend disproportionately. Indeed, if we ignore the first two points in Figure~\ref{figure:geometric-complete-joint}d), then we obtain a runtime trend of $\mathcal{O}(\vert V\vert^{5.54})$ that is closer to $\mathcal{O}(\vert V\vert^{5.8})$ reported in~\cite{stein-2023}.
\tikzmath{\plotheight=3.7;}
\tikzmath{\plotwidth=4.45;}

\providecommand{\addvariableplot}{}
\renewcommand{\addvariableplot}[1]{\addplot+[
	only marks,
	mark=*,
	mark options={draw=none}
	] table[
	col sep=comma,
	x expr=\thisrow{numberOfPoints}, 
	y expr=\thisrow{medianEliminatedVariables}*100,
	y error minus expr=100*(\thisrow{medianEliminatedVariables} - \thisrow{q25EliminatedVariables}),
	y error plus expr=100*(\thisrow{q75EliminatedVariables} - \thisrow{medianEliminatedVariables})
	] {#1};}

\providecommand{\addtriangleplot}{}
\renewcommand{\addtriangleplot}[1]{\addplot+[
	only marks,
	mark=*,
	mark options={draw=none}
	] table[
	col sep=comma,
	x expr=\thisrow{numberOfPoints}, 
	y expr=100*\thisrow{medianEliminatedTriangles},
	y error minus expr=100*\thisrow{medianEliminatedTriangles} - 100*\thisrow{q25EliminatedTriangles},
	y error plus expr=100*\thisrow{q75EliminatedTriangles} - 100*\thisrow{medianEliminatedTriangles} 
	] {#1};}

\providecommand{\addruntimeplot}{}
\renewcommand{\addruntimeplot}[1]{\addplot+[
	only marks,
	mark=*,
	mark options={draw=none}
	] table[
	col sep=comma,
	x expr=\thisrow{numberOfPoints}, 
	y expr=\thisrow{medianDuration} / 1e9,
	y error minus expr=(\thisrow{medianDuration} - \thisrow{q25Duration}) / 1e9,
	y error plus expr=(\thisrow{q75Duration} - \thisrow{medianDuration})  / 1e9
	] {#1};}

\providecommand{\addpartitionvariableplot}{}
\renewcommand{\addpartitionvariableplot}[1]{\addplot+[
	only marks,
	mark=*,
	mark options={draw=none}
	] table[
	col sep=comma,
	x expr=\thisrow{alpha}, 
	y expr=\thisrow{medianEliminatedVariables}*100,
	y error minus expr=100*(\thisrow{medianEliminatedVariables} - \thisrow{q25EliminatedVariables}),
	y error plus expr=100*(\thisrow{q75EliminatedVariables} - \thisrow{medianEliminatedVariables})
	] {#1};}

\providecommand{\addpartitionruntimeplot}{}
\renewcommand{\addpartitionruntimeplot}[1]{\addplot+[
	only marks,
	mark=*,
	mark options={draw=none}
	] table[
	col sep=comma,
	x expr=\thisrow{alpha}, 
	y expr=\thisrow{medianDuration} / 1e9,
	y error minus expr=(\thisrow{medianDuration} - \thisrow{q25Duration}) / 1e9,
	y error plus expr=(\thisrow{q75Duration} - \thisrow{medianDuration})  / 1e9
	] {#1};}

\begin{figure}[!b]
	\centering
	
	\begin{tikzpicture}[baseline]
		\begin{groupplot}[group style={group size= 4 by 1, horizontal sep=0.80cm, vertical sep=0.75cm}]
			
			\nextgroupplot[
			title={Variables [\%]},
			xlabel=$\alpha$,
			xmin=0.25,
			xmax=0.7,
			ymin=0,
			ymax=100,
			width=\plotwidth cm,
			height=\plotheight cm,
			legend style={legend columns=1, font=\tiny}
			]
			\addpartitionvariableplot{./data/complete/partition/alphas/n7_beta0.00_stats.csv}
			\addlegendentry{$\beta = 0.00$}
			\addpartitionvariableplot{./data/complete/partition/alphas/n7_beta0.01_stats.csv}
			\addlegendentry{$\beta = 0.01$}
			\addpartitionvariableplot{./data/complete/partition/alphas/n7_beta0.50_stats.csv}
			\addlegendentry{$\beta = 0.50$}
			\addpartitionvariableplot{./data/complete/partition/alphas/n7_beta1.00_stats.csv}
			\addlegendentry{$\beta = 1.00$}
			
			\nextgroupplot[
			title={Runtime [s]},
			xlabel=$\alpha$,
			xmin=0.25,
			xmax=0.7,
			ymin=0,
			ymax=60,
			width=\plotwidth cm,
			height=\plotheight cm,
			domain=8:750,
			domain y=1e-6:1e2,
			samples=100,
			legend pos=south east,
			legend style={font=\scriptsize},
			xshift=-0.10cm
			]
			\addpartitionruntimeplot{./data/complete/partition/alphas/n7_beta0.00_stats.csv}
			\addpartitionruntimeplot{./data/complete/partition/alphas/n7_beta0.01_stats.csv}
			\addpartitionruntimeplot{./data/complete/partition/alphas/n7_beta0.50_stats.csv}
			\addpartitionruntimeplot{./data/complete/partition/alphas/n7_beta1.00_stats.csv}
			
			\nextgroupplot[
			title={Variables [\%]},
			xlabel=$\vert V\vert$,
			xmin=0,
			xmax=110,
			ymin=0,
			ymax=100,
			width=\plotwidth cm,
			height=\plotheight cm,
			legend style={legend columns=1, font=\tiny},
			xlabel shift=-0.1cm
			]
			\addvariableplot{./data/complete/partition/ns/alpha0.00_beta0.50_stats.csv}
			\addlegendentry{$\alpha = 0.00$}
			\addvariableplot{./data/complete/partition/ns/alpha0.40_beta0.50_stats.csv}
			\addlegendentry{$\alpha = 0.40$}
			\addvariableplot{./data/complete/partition/ns/alpha0.50_beta0.50_stats.csv}
			\addlegendentry{$\alpha = 0.50$}
			\addvariableplot{./data/complete/partition/ns/alpha0.65_beta0.50_stats.csv}
			\addlegendentry{$\alpha = 0.65$}
			
			\nextgroupplot[
			xmode=log,
			ymode=log,
			title={Runtime [s]},
			xlabel=$\vert V\vert$,
			xmin=6,
			xmax=110,
			ymin=1e-6,
			ymax=1e2,
			width=\plotwidth cm,
			height=\plotheight cm,
			domain=8:750,
			domain y=1e-6:1e2,
			samples=100,
			legend pos=south east,
			legend style={font=\scriptsize},
			xshift=0.30cm,
			xlabel shift=-0.2cm
			]
			\addruntimeplot{./data/complete/partition/ns/alpha0.00_beta0.50_stats.csv}
			\addruntimeplot{./data/complete/partition/ns/alpha0.40_beta0.50_stats.csv}
			\addruntimeplot{./data/complete/partition/ns/alpha0.50_beta0.50_stats.csv}
			\addruntimeplot{./data/complete/partition/ns/alpha0.65_beta0.50_stats.csv}
			
			\addplot[no markers, color=black] {3.64e-9*x^5.79};
			\legend{,,,,$O(\lvert V\rvert^{5.79})$}
		\end{groupplot}		
		
		\node[left=0.50cm, above=0.25cm] at (group c1r1.north west) {a)};
		\node[left=0.35cm, above=0.25cm] at (group c2r1.north west) {b)};
		\node[left=0.50cm, above=0.25cm] at (group c3r1.north west) {c)};
		\node[left=0.70cm, above=0.25cm] at (group c4r1.north west) {d)};
	\end{tikzpicture}
	\\[-2.75ex]
	\caption{
		We report above for the partition dataset the percentage of fixed variables and the runtime after applying all conditions jointly, as described in Section~\ref{section:mixing-conditions}, for the implementation from~\cite{stein-2023} for complete graphs and on the new hardware. 
		a-b) show these with respect to $\alpha$ and $\vert V\vert = 56$. c-d) show these with respect to $\vert V\vert$.
	}
	\label{figure:partition-complete-joint}
\end{figure}
\tikzmath{\plotheight=3.7;}
\tikzmath{\plotwidth=4.45;}

\providecommand{\addgeometricvariableplot}{}
\renewcommand{\addgeometricvariableplot}[1]{\addplot[
	only marks,
	mark=*,
	black,
	mark options={draw=none}
	] table[
	col sep=comma,
	x expr=\thisrow{sigma}, 
	y expr=\thisrow{medianEliminatedVariables}*100,
	y error minus expr=100*(\thisrow{medianEliminatedVariables} - \thisrow{q25EliminatedVariables}),
	y error plus expr=100*(\thisrow{q75EliminatedVariables} - \thisrow{medianEliminatedVariables})
	] {#1};}

\providecommand{\addgeometricruntimeplot}{}
\renewcommand{\addgeometricruntimeplot}[1]{\addplot[
	only marks,
	mark=*,
	black,
	mark options={draw=none}
	] table[
	col sep=comma,
	x expr=\thisrow{sigma}, 
	y expr=\thisrow{medianDuration} / 1e9,
	y error minus expr=(\thisrow{medianDuration} - \thisrow{q25Duration}) / 1e9,
	y error plus expr=(\thisrow{q75Duration} - \thisrow{medianDuration})  / 1e9
	] {#1};}

\begin{figure}[!b]
	\centering
	
	\begin{tikzpicture}[baseline]
		\begin{groupplot}[group style={group size= 4 by 1, horizontal sep=0.8cm, vertical sep=1cm}]			
			\nextgroupplot[
			xlabel=$\sigma$,
			title={Variabes [\%]},
			xmin=0,
			xmax=0.1,
			ymin=0,
			ymax=100,
			width=\plotwidth cm,
			height=\plotheight cm,
			xtick distance=0.05,
			ticklabel style={
				/pgf/number format/.cd,
				/pgf/number format/fixed
			}
			]
			\addgeometricvariableplot{./data/complete/geometric/sigmas/n54_stats.csv}
			
			\nextgroupplot[
			xlabel=$\sigma$,
			title={Runtime [s]},
			xmin=0,
			xmax=0.1,
			ymin=0,
			ymax=60,
			width=\plotwidth cm,
			height=\plotheight cm,
			legend pos=south east,
			legend style={font=\scriptsize},
			xtick distance=0.05,
			ticklabel style={
				/pgf/number format/.cd,
				/pgf/number format/fixed
			},
			xshift=-0.10cm
			]
			
			\addgeometricruntimeplot{./data/complete/geometric/sigmas/n54_stats.csv}
			
			\nextgroupplot[
			title={Variables [\%]},
			xlabel=$\vert V\vert$,
			xmin=0,
			xmax=110,
			ymin=0,
			ymax=100,
			width=\plotwidth cm,
			height=\plotheight cm,
			legend pos = south east,
			legend style={legend columns=1, font=\tiny},
			xlabel shift=-0.1cm
			]
			\addvariableplot{./data/complete/geometric/ns/sigma0.0100_stats.csv}
			\addlegendentry{$\sigma = 0.01$}
			\addvariableplot{./data/complete/geometric/ns/sigma0.0600_stats.csv}
			\addlegendentry{$\sigma = 0.06$}
			\addvariableplot{./data/complete/geometric/ns/sigma0.1000_stats.csv}
			\addlegendentry{$\sigma = 0.10$}
			
			\nextgroupplot[
			xmode=log,
			ymode=log,
			title={Runtime [s]},
			xlabel=$\vert V\vert$,
			xmin=6,
			xmax=110,
			ymin=1e-6,
			ymax=1e2,
			width=\plotwidth cm,
			height=\plotheight cm,
			domain=8:750,
			domain y=1e-6:1e2,
			samples=100,
			legend pos=south east,
			legend style={font=\scriptsize},
			xshift=0.3cm,
			xlabel shift=-0.2cm
			]
			\addruntimeplot{./data/complete/geometric/ns/sigma0.0100_stats.csv}
			\addruntimeplot{./data/complete/geometric/ns/sigma0.0600_stats.csv}
			\addruntimeplot{./data/complete/geometric/ns/sigma0.1000_stats.csv}
			
			\addplot[no markers, color=black] {5.11e-12*x^7.45};
			\legend{,,,$O(\lvert V\rvert^{7.45})$}
			
		\end{groupplot}
		
		\node[left=0.25cm, above=0.25cm] at (group c1r1.north west) {a)};
		\node[left=0.25cm, above=0.25cm] at (group c2r1.north west) {b)};
		\node[left=0.25cm, above=0.25cm] at (group c3r1.north west) {c)};
		\node[left=0.25cm, above=0.25cm] at (group c4r1.north west) {d)};
	\end{tikzpicture}
	\\[-2.75ex]
	\caption{
		We report above for the geometric dataset the percentage of fixed variables and the runtime after applying all conditions jointly, as described in Section~\ref{section:mixing-conditions}, for the implementation from~\cite{stein-2023} for complete graphs and on the new hardware. a-b) show these with respect to $\sigma$ and $\vert V \vert = 54$. c-d) show these with respect to $\vert V\vert$.
	}
	\label{figure:geometric-complete-joint}
\end{figure}
\tikzmath{\plotheight=3.7;}
\tikzmath{\plotwidth=4.55;}

\providecommand{\addvariableplot}{}	
\renewcommand{\addvariableplot}[1]{\addplot+[
	only marks,
	mark=*,
	mark options={draw=none}
	] table[
	col sep=comma,
	x expr=\thisrow{alpha}, 
	y expr=\thisrow{medianEliminatedVariables}*100,
	y error minus expr=100*(\thisrow{medianEliminatedVariables} - \thisrow{q25EliminatedVariables}),
	y error plus expr=100*(\thisrow{q75EliminatedVariables} - \thisrow{medianEliminatedVariables})
	] {#1};}

\providecommand{\addtriangleplot}{}	
\renewcommand{\addtriangleplot}[1]{\addplot+[
	only marks,
	mark=*,
	mark options={draw=none}
	] table[
	col sep=comma,
	x expr=\thisrow{alpha}, 
	y expr=100*\thisrow{medianEliminatedTriples},
	y error minus expr=100*\thisrow{medianEliminatedTriples} - 100*\thisrow{q25EliminatedTriples},
	y error plus expr=100*\thisrow{q75EliminatedTriples} - 100*\thisrow{medianEliminatedTriples} 
	] {#1};}

\providecommand{\addvariableplotsigma}{}	
\renewcommand{\addvariableplotsigma}[1]{\addplot[
	only marks,
	mark=*,
	black,
	mark options={draw=none}
	] table[
	col sep=comma,
	x expr=\thisrow{sigma}, 
	y expr=\thisrow{medianEliminatedVariables}*100,
	y error minus expr=100*(\thisrow{medianEliminatedVariables} - \thisrow{q25EliminatedVariables}),
	y error plus expr=100*(\thisrow{q75EliminatedVariables} - \thisrow{medianEliminatedVariables})
	] {#1};}

\providecommand{\addtriangleplotsigma}{}	
\renewcommand{\addtriangleplotsigma}[1]{\addplot[
	only marks,
	mark=*,
	black,
	mark options={draw=none}
	] table[
	col sep=comma,
	x expr=\thisrow{sigma}, 
	y expr=100*\thisrow{medianEliminatedTriangles},
	y error minus expr=100*\thisrow{medianEliminatedTriangles} - 100*\thisrow{q25EliminatedTriangles},
	y error plus expr=100*\thisrow{q75EliminatedTriangles} - 100*\thisrow{medianEliminatedTriangles} 
	] {#1};}

\begin{figure}[!b]
	\centering
	
	\begin{tikzpicture}[baseline]
		\begin{groupplot}[group style={group size= 4 by 2, horizontal sep=0.7cm}]
			
			\nextgroupplot[
			title={Variables [\%]},
			xlabel=$\alpha$,
			xmin=0.25,
			xmax=0.7,
			ymin=0,
			ymax=100,
			width=\plotwidth cm,
			height=\plotheight cm,
			legend pos=north east,
			legend style={font=\tiny}
			]
			\addvariableplot{./data/complete/partition-individual/alphas/n7_beta0.00_findIndependentSubsets_stats.csv}
			\addlegendentry{$\beta = 0.00$}
			
			\addvariableplot{./data/complete/partition-individual/alphas/n7_beta0.01_findIndependentSubsets_stats.csv}
			\addlegendentry{$\beta = 0.01$}
			
			\addvariableplot{./data/complete/partition-individual/alphas/n7_beta0.50_findIndependentSubsets_stats.csv}
			\addlegendentry{$\beta = 0.50$}
			
			\addvariableplot{./data/complete/partition-individual/alphas/n7_beta1.00_findIndependentSubsets_stats.csv}
			\addlegendentry{$\beta = 1.00$}

			\nextgroupplot[
			title={Variables [\%]},
			xlabel=$\alpha$,
			xmin=0.25,
			xmax=0.7,
			ymin=0,
			ymax=100,
			width=\plotwidth cm,
			height=\plotheight cm,
			yticklabels=\empty
			]
			\addvariableplot{./data/complete/partition-individual/alphas/n7_beta0.00_edgeCut_stats.csv}
			\addvariableplot{./data/complete/partition-individual/alphas/n7_beta0.01_edgeCut_stats.csv}
			\addvariableplot{./data/complete/partition-individual/alphas/n7_beta0.50_edgeCut_stats.csv}
			\addvariableplot{./data/complete/partition-individual/alphas/n7_beta1.00_edgeCut_stats.csv}
			
			\nextgroupplot[
			title={Triples [\%]},
			xlabel=$\alpha$,
			yticklabels=\empty,
			xmin=0.25,
			xmax=0.7,
			ymin=0,
			ymax=100,
			width=\plotwidth cm,
			height=\plotheight cm
			]
			\addtriangleplot{./data/complete/partition-individual/alphas/n7_beta0.00_tripletCut_stats.csv}
			\addtriangleplot{./data/complete/partition-individual/alphas/n7_beta0.01_tripletCut_stats.csv}
			\addtriangleplot{./data/complete/partition-individual/alphas/n7_beta0.50_tripletCut_stats.csv}
			\addtriangleplot{./data/complete/partition-individual/alphas/n7_beta1.00_tripletCut_stats.csv}
			
			\nextgroupplot[
			title={Variables [\%]},
			xlabel=$\alpha$,
			yticklabels=\empty,
			xmin=0.25,
			xmax=0.7,
			ymin=0,
			ymax=100,
			width=\plotwidth cm,
			height=\plotheight cm
			]
			\addvariableplot{./data/complete/partition-individual/alphas/n7_beta0.00_subsetJoin_stats.csv}
			\addvariableplot{./data/complete/partition-individual/alphas/n7_beta0.01_subsetJoin_stats.csv}
			\addvariableplot{./data/complete/partition-individual/alphas/n7_beta0.50_subsetJoin_stats.csv}
			\addvariableplot{./data/complete/partition-individual/alphas/n7_beta1.00_subsetJoin_stats.csv}

			\nextgroupplot[
			xlabel=$\sigma$,
			xmin=0.0,
			xmax=0.1,
			ymin=0,
			ymax=100,
			width=\plotwidth cm,
			height=\plotheight cm,
			xtick distance=0.05,
			ticklabel style={
				/pgf/number format/.cd,
				/pgf/number format/fixed
			}
			]
			\addvariableplotsigma{./data/complete/geometric_individual/sigmas/n54_findIndependentSubsets_stats.csv}

			\nextgroupplot[
			xlabel=$\sigma$,
			xmin=0.0,
			xmax=0.1,
			ymin=0,
			ymax=100,
			width=\plotwidth cm,
			height=\plotheight cm,
			yticklabels=\empty,
			xtick distance=0.05,
			ticklabel style={
				/pgf/number format/.cd,
				/pgf/number format/fixed
			}
			]
			\addvariableplotsigma{./data/complete/geometric_individual/sigmas/n54_edgeCut_stats.csv}
			
			\nextgroupplot[
			xlabel=$\sigma$,
			yticklabels=\empty,
			xmin=0.0,
			xmax=0.1,
			ymin=0,
			ymax=100,
			width=\plotwidth cm,
			height=\plotheight cm,
			xtick distance=0.05,
			ticklabel style={
				/pgf/number format/.cd,
				/pgf/number format/fixed
			}
			]
			\addtriangleplotsigma{./data/complete/geometric_individual/sigmas/n54_tripletCut_stats.csv}
						
			\nextgroupplot[
			xlabel=$\sigma$,
			yticklabels=\empty,
			xmin=0.0,
			xmax=0.1,
			ymin=0,
			ymax=100,
			width=\plotwidth cm,
			height=\plotheight cm,
			xtick distance=0.05,
			ticklabel style={
				/pgf/number format/.cd,
				/pgf/number format/fixed
			}
			]
			\addvariableplotsigma{./data/complete/geometric_individual/sigmas/n54_subsetJoin_stats.csv}
			
		\end{groupplot}
	\end{tikzpicture}
	\\[-2.75ex]
	\caption{
		We report above the percentage of fixed variables and triples when applying Propositions~\ref{lemma:persistency-subset-separation} to \ref{lemma:persistency-triplet-cut} and Corollary~\ref{corollary:subset-join-all-pairs-at-once} (Columsn~1--4), separately, and for the implementation from~\cite{stein-2023} for complete graphs and on the new hardware.
		Row 1 shows these for the partition data set with 56 elements. 
		Row 2 shows these for the geometric data set with 54 elements.
	}
	\label{figure:partition-goemetric-complete-individual-combined}
\end{figure}
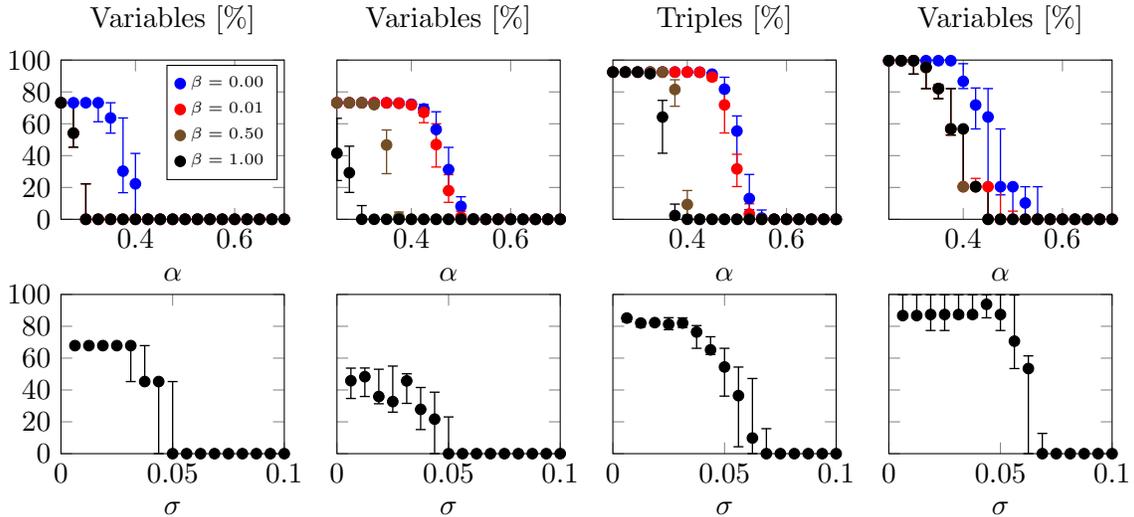

\vskip 0.2in
\bibliography{arxiv}

\end{document}